\documentclass[11pt]{article}
\pdfoutput=1
\usepackage{xifthen}

\newcommand{\me}{miforbes}

	\newenvironment{proof-sketch}{\medskip\noindent{\em Sketch of Proof.}\hspace*{1em}}{\qed\bigskip}
	\newenvironment{proof-attempt}{\medskip\noindent{\em Proof attempt.}\hspace*{1em}}{\bigskip}

	{\vspace{1ex}\noindent{\emph{Proof.}}\hspace{0.5em}\def\myproof@name{#1}\newif\ifqedhere\qedherefalse}%
	{\ifqedhere \else \hfill{\tiny \qed\ (\myproof@name)}\vspace{1ex} \qedherefalse \fi}

\usepackage{todonotes}
\usepackage{lipsum}
\usepackage{setspace}
\usepackage{mdframed}
\ifthenelse{\equal{\me}{iddo}}{
	\usepackage[usenames,pdf]{pstricks}
	\usepackage{epsfig}
	\usetikzlibrary{decorations.pathreplacing,backgrounds}
}{}

	\usepackage[hyphens]{url}
	\usepackage[margin=1in]{geometry}
	\usepackage[noadjust]{cite}
	\usepackage[normalem]{ulem}
	
	\usepackage[small]{complexity}
	\usepackage{algpseudocode}
	\usepackage{amsfonts, amsmath, amssymb, amsthm}
	\usepackage{bbm}
	\usepackage{etoolbox}
	\usepackage{graphicx}
	\usepackage{intcalc}
	\usepackage{lmodern}
	\usepackage{microtype}
	\usepackage{nag}
	\usepackage{nicefrac}
	\usepackage[super]{nth}
	\usepackage{stmaryrd}
	\ifthenelse{\NOT\equal{\me}{iddo}}{\usepackage{thmtools,thm-restate}}{}
	\usepackage{tikz}
	\usepackage{tocbibind}
	\usepackage{version}
	\usepackage{xparse}
	\usepackage{xspace}
	\usepackage{xstring}

	\usepackage[colorlinks=true]{hyperref}
	\hypersetup{
		linkcolor=[rgb]{0.5, 0.2, 0.2},
		citecolor=[rgb]{0.2, 0.5, 0.2},
		urlcolor =[rgb]{0.2, 0.2, 0.5}
	}
	\usepackage{algorithm}

	\makeatletter
	\global\let\tikz@ensure@dollar@catcode=\relax
	\makeatother

	\SetSymbolFont{stmry}{bold}{U}{stmry}{m}{n}

	\PackageWarning{miforbes}{overfullrule enabled}
	\newcommand{\ignore}[1]{}
	\newcommand{\demph}[1]{\textbf{#1}}

	\ifthenelse{\NOT\equal{\me}{iddo}}{
	\numberwithin{equation}{section}
	\declaretheoremstyle[bodyfont=\it,qed=$\qedsymbol$]{noproofstyle} 
	\declaretheoremstyle[bodyfont=\it,qed=$\lozenge$]{defstyle} 
	\declaretheoremstyle[qed=$\lozenge$]{rmkstyle} 

	\theoremstyle{plain}

	\declaretheorem[name=Theorem,numberlike=equation]{theorem}
	\declaretheorem[name=Theorem,numberlike=equation,style=noproofstyle]{theoremwp}
	\declaretheorem[name=Theorem,unnumbered]{theorem*}
	\declaretheorem[name=Theorem,unnumbered,style=noproofstyle]{theoremwp*}

	\declaretheorem[name=Lemma,numberlike=equation]{lemma}
	\declaretheorem[name=Lemma,numberlike=equation,style=noproofstyle]{lemmawp}
	\declaretheorem[name=Lemma,unnumbered]{lemma*}
	\declaretheorem[name=Lemma,unnumbered,style=noproofstyle]{lemmawp*}

	\declaretheorem[name=Corollary,numberlike=equation]{corollary}
	\declaretheorem[name=Corollary,numberlike=equation,style=noproofstyle]{corollarywp}
	\declaretheorem[name=Corollary,unnumbered]{corollary*}
	\declaretheorem[name=Corollary,unnumbered,style=noproofstyle]{corollarywp*}

	\declaretheorem[name=Proposition,numberlike=equation]{proposition}
	
	\declaretheorem[name=Proposition,unnumbered]{proposition*}
	\declaretheorem[name=Proposition,unnumbered,style=noproofstyle]{propositionwp*}

	\declaretheorem[name=Claim,numberlike=equation]{claim}
	\declaretheorem[name=Claim,numberlike=equation,style=noproofstyle]{claimwp}
	\declaretheorem[name=Claim,unnumbered]{claim*}
	\declaretheorem[name=Claim,unnumbered,style=noproofstyle]{claimwp*}

	\declaretheorem[name=Subclaim,numberlike=equation]{subclaim}
	
	\declaretheorem[name=Subclaim,unnumbered]{subclaim*}
	\declaretheorem[name=Subclaim,unnumbered,style=noproofstyle]{subclaimwp*}

	\declaretheorem[name=Theorem,numberlike=equation,style=noproofstyle]{theorem-cited}

	\declaretheorem[name=Fact,numberlike=equation,style=noproofstyle]{fact}

	\declaretheorem[name=Definition,style=defstyle,numberlike=equation]{definition}
	\declaretheorem[name=Definition,style=defstyle,unnumbered]{definition*}

	\declaretheorem[name=Conjecture,style=defstyle,unnumbered]{conjecture*}

	\declaretheorem[name=Construction,style=defstyle,numberlike=equation]{construction}
	\declaretheorem[name=Construction,style=defstyle,unnumbered]{construction*}

	\declaretheorem[name=Open Question,style=defstyle,numberlike=equation]{open}
	\declaretheorem[name=Open Question,style=defstyle,unnumbered]{open*}

	\declaretheorem[name=Remark,style=rmkstyle,numberlike=equation]{remark}
	\declaretheorem[name=Remark,style=rmkstyle,unnumbered]{remark*}

	\declaretheorem[name=Example,style=rmkstyle,numberlike=equation]{example}
	\declaretheorem[name=Example,style=rmkstyle,unnumbered]{example*}

	\declaretheorem[name=Notation,style=rmkstyle,numberlike=equation]{notation}
	\declaretheorem[name=Notation,style=rmkstyle,unnumbered]{notation*}

	\declaretheorem[name=Question,style=rmkstyle,unnumbered]{question*}
	}{}

	\newcommand*{\subproofname}{Sub-Proof:}
	\newenvironment{subproof}[1][\subproofname]{\begin{proof}[#1]\renewcommand*{\qedsymbol}{\(\scalebox{.6}{$\square$}\)}}{\end{proof}}

	\newcommand{\la}{\langle}
	\newcommand{\ra}{\rangle}

	\newcommand{\ind}[1]{\mathbbm{1}_{#1}}

	\newcommand{\ceil}[1]{{\lceil{#1}\rceil}}			
	\newcommand{\nc}[1]{{\la #1\ra}}			
	
	\newcommand{\subsets}[1]{2^{#1}}

	\newcommand{\bits}{\ensuremath{\{0,1\}}}

	\DeclareMathOperator{\chara}{char}

	\DeclareMathOperator{\ideg}{ideg}

	\DeclareMathOperator{\perm}{perm}
	\DeclareMathOperator{\rank}{rank}

	\DeclareMathOperator{\spn}{span}
	
	\DeclareMathOperator{\supp}{Supp}

	\newcommand{\ellzero}[1]{{\ensuremath{|#1|_0}}}
	\newcommand{\ellone}[1]{{\ensuremath{|#1|_1}}}
	
	\newcommand{\ellinfty}[1]{{\ensuremath{|#1|_\infty}}}

	\newcommand{\coeff}[1]{{\mathrm{Coeff}}_{#1}}

	\DeclareMathOperator{\LM}{LM}

	\newcommand{\coeffs}[1]{{\mathbf{Coeff}}_{#1}}
	\newcommand{\evals}[1]{{\mathbf{Eval}}_{#1}}

	\newcommand{\pow}{\ensuremath{\bigwedge}}

	\newcommand{\sumpowc}{\sumpow{\O(1)}}
	\newcommand{\sumpowsum}{\ensuremath{\sum\pow\sum}\xspace}
	
	\newcommand{\sumpow}[1]{\ensuremath{\sum\pow\sum\prod^{#1}}\xspace}
	\newcommand{\sumprod}{\ensuremath{\sum\prod}\xspace}

	\newlang{\PIT}{PIT}

	\newcommand{\tsum}{{\textstyle\sum}}

	\newcommand{\e}{\mathrm{e}}
	\newcommand{\T}{\mathrm{tr}}
	\newcommand{\eps}{\epsilon}

	\renewcommand{\O}{\mathcal{O}}
	
	\newcommand{\Sn}{\mathfrak{S}}

	\newcommand{\eqdef}{:=}
	\newcommand{\defeq}{=:}

	\PackageWarning{miforbes}{get rid of matrix commands}

	\newcommand{\rv}{\mathsf}
	\renewcommand{\vec}[1]{\overline{#1}}

	\renewcommand{\E}{\mathbb{E}}			
	
	\newcommand{\F}{\mathbb{F}}
	
	\renewcommand{\K}{\mathbb{K}}
	\renewcommand{\R}{\mathbb{R}}
	\newcommand{\N}{\mathbb{N}}

	\newcommand{\cC}{\mathcal{C}}
	\newcommand{\cD}{\mathcal{D}}

	\newcommand{\cG}{\mathcal{G}}

	\newcommand{\cvG}{{\vec{\cG}}}

	\newcommand{\rX}{\rv{X}}

	\makeatletter
	\newcommand{\va}{{\vec{a}}\@ifnextchar{^}{\!\:}{}}
	\newcommand{\vb}{{\vec{b}}\@ifnextchar{^}{\!\:}{}}
	\newcommand{\vc}{{\vec{c}}\@ifnextchar{^}{\!\:}{}}
	\newcommand{\vd}{{\vec{d}}\@ifnextchar{^}{\!\:}{}}
	\newcommand{\ve}{{\vec{e}}\@ifnextchar{^}{\!\:}{}}

	\newcommand{\vf}{{\vec{f}}}
	
	\newcommand{\vg}{{\vec{g}}\@ifnextchar{^}{\!\:}{}}
	\newcommand{\vh}{{\vec{h}}\@ifnextchar{^}{\!\:}{}}
	\newcommand{\vi}{{\vec{i}}\@ifnextchar{^}{\!\:}{}}
	\newcommand{\vj}{{\vec{j}}\@ifnextchar{^}{\!\:}{}}
	\newcommand{\vk}{{\vec{k}}\@ifnextchar{^}{\!\:}{}}
	\newcommand{\vl}{{\vec{\ell}}\@ifnextchar{^}{\!\:}{}}
	\newcommand{\vm}{{\vec{m}}\@ifnextchar{^}{\!\:}{}}
	\newcommand{\vn}{{\vec{n}}\@ifnextchar{^}{\!\:}{}}
	\newcommand{\vo}{{\vec{o}}\@ifnextchar{^}{\!\:}{}}
	\newcommand{\vp}{{\vec{p}}\@ifnextchar{^}{\!\:}{}}
	\newcommand{\vq}{{\vec{q}}\@ifnextchar{^}{\!\:}{}}
	\newcommand{\vr}{{\vec{r}}\@ifnextchar{^}{\!\:}{}}
	\newcommand{\vs}{{\vec{s}}\@ifnextchar{^}{\!\:}{}}
	\newcommand{\vt}{{\vec{t}}\@ifnextchar{^}{\!\:}{}}
	\newcommand{\vu}{{\vec{u}}\@ifnextchar{^}{\!\:}{}}
	\newcommand{\vv}{{\vec{v}}\@ifnextchar{^}{\!\:}{}}
	\newcommand{\vw}{{\vec{w}}\@ifnextchar{^}{\!\:}{}}
	\newcommand{\vy}{{\vec{y}}\@ifnextchar{^}{\!\:}{}}
	\newcommand{\vx}{{\vec{x}}\@ifnextchar{^}{}{}}		
	\newcommand{\vz}{{\vec{z}}\@ifnextchar{^}{\!\:}{}}

	\newcommand{\vA}{{\vec{A}}\@ifnextchar{^}{\!\:}{}}
	\newcommand{\vB}{{\vec{B}}\@ifnextchar{^}{\!\:}{}}
	\newcommand{\vC}{{\vec{C}}\@ifnextchar{^}{\!\:}{}}
	\newcommand{\vD}{{\vec{D}}\@ifnextchar{^}{\!\:}{}}
	\newcommand{\vE}{{\vec{E}}\@ifnextchar{^}{\!\:}{}}
	\newcommand{\vF}{{\vec{F}}\@ifnextchar{^}{\!\:}{}}
	\newcommand{\vG}{{\vec{G}}\@ifnextchar{^}{\!\:}{}}
	\newcommand{\vH}{{\vec{H}}\@ifnextchar{^}{\!\:}{}}
	\newcommand{\vI}{{\vec{I}}\@ifnextchar{^}{\!\:}{}}
	\newcommand{\vJ}{{\vec{J}}\@ifnextchar{^}{\!\:}{}}
	\newcommand{\vK}{{\vec{K}}\@ifnextchar{^}{\!\:}{}}
	\newcommand{\vL}{{\vec{L}}\@ifnextchar{^}{\!\:}{}}
	\newcommand{\vM}{{\vec{M}}\@ifnextchar{^}{\!\:}{}}
	\newcommand{\vN}{{\vec{N}}\@ifnextchar{^}{\!\:}{}}
	\newcommand{\vO}{{\vec{O}}\@ifnextchar{^}{\!\:}{}}
	\newcommand{\vP}{{\vec{P}}\@ifnextchar{^}{\!\:}{}}
	\newcommand{\vQ}{{\vec{Q}}\@ifnextchar{^}{\!\:}{}}
	\newcommand{\vR}{{\vec{R}}\@ifnextchar{^}{\!\:}{}}
	\newcommand{\vS}{{\vec{S}}\@ifnextchar{^}{\!\:}{}}
	\newcommand{\vT}{{\vec{T}}\@ifnextchar{^}{\!\:}{}}
	\newcommand{\vU}{{\vec{U}}\@ifnextchar{^}{\!\:}{}}
	\newcommand{\vV}{{\vec{V}}\@ifnextchar{^}{\!\:}{}}
	\newcommand{\vW}{{\vec{W}}\@ifnextchar{^}{\!\:}{}}
	\newcommand{\vY}{{\vec{Y}}\@ifnextchar{^}{\!\:}{}}
	\newcommand{\vX}{{\vec{X}}\@ifnextchar{^}{}{}}		
	\newcommand{\vZ}{{\vec{Z}}\@ifnextchar{^}{\!\:}{}}
	\makeatother

	\newcommand{\vno}{{\vec{1}}}
	\newcommand{\vnz}{{\vec{0}}}

	\newcommand{\vaa}{{\vec{\alpha}}}
	\newcommand{\vbb}{{\vec{\beta}}}

\hypersetup{
	linkcolor=[rgb]{0.2,0.2,0.5},
	citecolor=[rgb]{0.2, 0.5, 0.2},
	urlcolor=[rgb]{0.5, 0.2, 0.2}
}

\usepackage{nth}
\usepackage{intcalc}
\usepackage{etoolbox}
\usepackage{xstring}

\newcommand{\ECCCurl}[2]{http://eccc.hpi-web.de/report/\ifnumcomp{#1}{>}{93}{19}{20}#1/#2/}

\newcommand{\shortECCC}[2]{\texttt{\href{http://eccc.hpi-web.de/report/\ifnumcomp{#1}{>}{93}{19}{20}#1/#2/}{eccc:TR#1-#2}}}

\newcommand{\parseECCC}[1]{
\StrSubstitute{#1}{TR}{}[\tmpstring]%
\IfSubStr{\tmpstring}{/}{ 
\StrBefore{\tmpstring}{/}[\ecccyear]%
\StrBehind{\tmpstring}{/}[\ecccreport]%
}{
\StrBefore{\tmpstring}{-}[\ecccyear]%
\StrBehind{\tmpstring}{-}[\ecccreport]%
}%
\shortECCC{\ecccyear}{\ecccreport}}

\PackageWarning{miforbes}{LATER: is deferred}

\newcommand{\LATERR}[2]{}

	\DeclareMathOperator{\AND}{AND}
	
	\DeclareMathOperator{\XOR}{XOR}
	\DeclareMathOperator{\size}{size}
	\newcommand{\aaaa}[1]{}
	\newcommand{\Pudlak}{Pudl{\'{a}}k\xspace}
	\newcommand{\Hrubes}{Hrube\v{s}\xspace}
	\newcommand{\Krajicek}{Kraj\'{i}\v{c}ek\xspace}
	\newcommand{\ml}{\operatorname{ml}}
	\makeatletter
	\renewcommand\paragraph{\@startsection{paragraph}{4}{\z@}%
		{1ex \@plus1ex \@minus.2ex}%
		{-1em}%
		{\normalfont\normalsize\bfseries}}
	\makeatother
	\ifthenelse{\equal{\me}{iddo}}{}
	{\declaretheorem[name=Goal,style=rmkstyle,numberlike=equation]{goal}}

	\newcommand{\nx}{{\neg x}}
	\newcommand{\ny}{{\neg y}}
	\newcommand{\vnx}{{\vec{\nx}}}
	\newcommand{\vny}{{\vec{\ny}}}

\ifthenelse{\equal{\me}{iddo}}{
	\newtheorem{theorem}{Theorem}[section]
	\newtheorem*{theorem*}{Theorem}
	\newtheorem{theoremwp}{Theorem}
	\newtheorem*{theoremwp*}{Theorem}
	\newtheorem{lemma}[theorem]{Lemma}
	\newtheorem{lemmawp}[theorem]{Lemma}

	\newtheorem{subclaim}[theorem]{Sub claim}
	\newtheorem{definition}{Definition}[section]
	\newtheorem{proposition}[theorem]{Proposition}
	\newtheorem{corollary}[theorem]{Corollary}
	\newtheorem{corollarywp}[theorem]{Corollary}
	\newtheorem{construction}[theorem]{Construction}
	\newtheorem{fact}[theorem]{Fact}
	\newtheorem{remark}[theorem]{Remark}
	\newtheorem{assumption}[theorem]{Assumption}
	\newtheorem{goal}[theorem]{Goal}

\sloppy

	\newenvironment{example}{\QuadSpace\par\noindent{\bf Example}:}{\HalfSpace}

}{
	
}

\def\I{{\mathbf{I}}}

\def\_{\,\,\,\,\,}

\allowdisplaybreaks

\def\RCD0#1{{\mbox{\rm R$_{#1}$(lin)}}}

\newcommand{\QuadSpace}{\smallskip}
\newcommand{\HalfSpace}{\medskip}

\usepackage{color}

\newcommand{\GSV}[1]{\cvG^{\mathrm{\scriptscriptstyle SV}}_{#1}}
\newcommand{\GSVb}[1]{\cvG^{\mathrm{\scriptscriptstyle SV'}}_{#1}}

\DeclareMathOperator{\TM}{TM} 
\DeclareMathOperator{\TD}{TD}
\renewcommand{\TC}{\operatorname{TC}}
\DeclareMathOperator{\LD}{LD}
\DeclareMathOperator{\LC}{LC}
\newcommand{\nTM}[1]{\left|\TM(#1)\right|}
\newcommand{\nLM}[1]{\left|\LM(#1)\right|}
\newcommand{\LMp}[1]{\LM\left(#1\right)}

\newcommand{\lIPS}{\texorpdfstring{IPS$_{\text{LIN}}$}{IPS-LIN}\xspace}
\newcommand{\lbIPS}{\texorpdfstring{\ensuremath{\text{IPS}_{\text{LIN}'}}\xspace}{IPS-LIN'}}

\title{Proof Complexity Lower Bounds from Algebraic Circuit Complexity}
\author{%
	Michael A. Forbes~\thanks{Email: \texttt{miforbes@csail.mit.edu}. Department of Computer Science, Princeton University. Supported by the Princeton Center for Theoretical Computer Science.}
	\and
	Amir Shpilka~\thanks{Email: \texttt{shpilka@post.tau.ac.il}. Department of Computer Science, Tel Aviv University, Tel Aviv, Israel. The research leading to these results has received funding from the European Community's Seventh Framework Programme (FP7/2007-2013) under grant agreement number 257575.}
	\and
	Iddo Tzameret~\thanks{Email: \texttt{iddo.tzameret@rhul.ac.uk}. Department of Computer Science, Royal Holloway, University of London, UK.}
	\and
	Avi Wigderson~\thanks{Email: \texttt{avi@math.ias.edu}. Institute for Advanced Study, Princeton. This research was partially supported by NSF grant CCF-1412958.}
}
\usepackage{amsmath}
\usepackage[normalem]{ulem}

\begin{document}

\maketitle 

\begin{abstract}
	We give upper and lower bounds on the power of subsystems of the \emph{Ideal Proof System (IPS)}, the algebraic proof system recently proposed by Grochow and Pitassi~\cite{GrochowPitassi14}, where the circuits comprising the proof come from various restricted algebraic circuit classes. This mimics an established research direction in the boolean setting for subsystems of \emph{Extended Frege} proofs, where proof-lines are circuits from restricted boolean circuit classes. Except one, all of the subsystems considered in this paper can simulate the well-studied \emph{Nullstellensatz} proof system, and prior to this work there were no known lower bounds when measuring proof size by the algebraic complexity of the polynomials (except with respect to degree, or to sparsity).

	We give two general methods of converting certain algebraic lower bounds into proof complexity ones. Our methods require stronger notions of lower bounds, which lower bound a polynomial as well as an entire family of polynomials it defines. Our techniques are reminiscent of existing methods for converting boolean circuit lower bounds into related proof complexity results, such as \emph{feasible interpolation}. We obtain the relevant types of lower bounds for a variety of classes (\emph{sparse polynomials}, \emph{depth-3 powering formulas}, \emph{read-once oblivious algebraic branching programs}, and \emph{multilinear formulas}), and infer the relevant proof complexity results. We complement our lower bounds by giving short refutations of the previously-studied \emph{subset-sum} axiom using IPS subsystems, allowing us to conclude strict separations between some of these subsystems.

	Our first method is a \emph{functional lower bound}, a notion of Grigoriev and Razborov~\cite{GrigorievRazborov00}, which is a function $\hat{f}:\bits^n\to\F$ such that any polynomial $f$ agreeing with $\hat{f}$ on the boolean cube requires large algebraic circuit complexity.  For our classes of interest, we develop functional lower bounds where $\hat{f}(\vx)$ equals $\nicefrac{1}{p(\vx)}$ where $p$ is a constant-degree polynomial, which in turn yield corresponding IPS lower bounds for proving that $p$ is nonzero over the boolean cube. In particular, we show super-polynomial lower bounds for refuting variants of the subset-sum axiom in various IPS subsystems.

	Our second method is to give \emph{lower bounds for multiples}, that is, to give explicit polynomials whose all (nonzero) multiples require large algebraic circuit complexity. By extending known techniques, we are able to obtain such lower bounds for our classes of interest, which we then use to derive corresponding IPS lower bounds.  Such lower bounds for multiples are of independent interest, as they have tight connections with the algebraic hardness versus randomness paradigm.
\end{abstract}

\newpage

\tableofcontents

\newpage

\section{Introduction}

Propositional proof complexity aims to understand and analyze the computational resources required to prove propositional tautologies, in the same way that circuit complexity studies the resources required to compute boolean functions. A typical goal would be to establish, for a given proof system, super-polynomial lower bounds on the \emph{size} of any proof of some propositional tautology. The seminal work of Cook and Reckhow~\cite{CookReckhow79} showed that this goal relates quite directly to fundamental hardness questions in computational complexity such as the \NP\ vs.~\coNP\ question: establishing super-polynomial lower bounds for \emph{every} propositional proof system would separate \NP\ from \coNP\ (and thus also \P\ from \NP). We refer the reader to \Krajicek~\cite{Krajicek95} for more on this subject.

Propositional proof systems come in a large variety, as different ones capture different forms of reasoning, either reasoning used to actually prove theorems, or reasoning used by algorithmic techniques for different types of search problems (as failure of the algorithm to find the desired object constitutes a proof of its nonexistence). Much of the research in proof complexity deals with propositional proof systems originating from logic or geometry.  Logical proof systems include such systems as \emph{resolution} (whose variants are related to popular algorithms for automated theory proving and SAT solving), as well as the \emph{Frege} proof system (capturing the most common logic text-book systems) and its many subsystems. Geometric proof systems include \emph{cutting-plane proofs}, capturing reasoning used in algorithms for integer programming, as well as proof systems arising from systematic strategies for rounding linear- or semidefinite-programming such as the \emph{lift-and-project} or \emph{sum-of-squares} hierarchies. 

In this paper we focus on algebraic proof systems, in which propositional tautologies (or rather contradictions) are expressed as unsatisfiable systems of polynomial equations and algebraic tools are used to refute them. This study originates with the work of  Beame, Impagliazzo, \Krajicek, Pitassi and \Pudlak~\cite{BeameIKPP96}, who introduced the  Nullstellensatz refutation system (based on Hilbert's Nullstellensatz), followed by the Polynomial Calculus system of Clegg, Edmonds, and Impagliazzo~\cite{CleggEI96}, which is a ``dynamic'' version of Nullstellensatz. In both systems the main measures of proof size that have been studied are the \emph{degree} and \emph{sparsity} of the polynomials appearing in the proof. Substantial work has lead to a very good understanding of the power of these systems with respect to these measures (see for example \cite{BussIKPRS96,Razborov98,Grigoriev98,IPS99,BussGIP01,AlekhnovichRazborov01} and references therein). 

However, the above measures of degree and sparsity are rather rough measures of a complexity of a proof. As such, Grochow and Pitassi~\cite{GrochowPitassi14} have recently advocated measuring the complexity of such proofs by their algebraic circuit size and shown that the resulting proof system can polynomially simulate strong proof systems such as the Frege system.  This naturally leads to the question of establishing lower bounds for this stronger proof system, even for restricted classes of algebraic circuits.  

In this work we establish such lower bounds for previously studied restricted classes of algebraic circuits, and show that these lower bounds are interesting by providing non-trivial \emph{upper} bounds in these proof systems for refutations of interesting sets of polynomial equations.  This provides what are apparently the first examples of lower bounds on the algebraic circuit size of propositional proofs in the Ideal Proof System (IPS) framework of Grochow and Pitassi~\cite{GrochowPitassi14}.

We note that obtaining proof complexity lower bounds from circuit complexity lower bounds is an established tradition that takes many forms. Most prominent are the lower bounds for subsystems of the Frege proof system defined by low-depth boolean circuits, and lower bounds of \Pudlak~\cite{Pudlak97} on Resolution and Cutting Planes system using the so-called feasible interpolation method.  We refer the reader again to \Krajicek~\cite{Krajicek95} for more details. Our approach here for algebraic systems shares features with both of these approaches.

The rest of this introduction is arranged as follows. In \autoref{sec:Nullstellensatz} we give the necessary background in algebraic proof complexity, and explain the IPS system. In \autoref{sec:intro:circuits} we define the algebraic complexity classes that will underlie the subsystems of IPS we will study. In \autoref{sec:results} we state our results and explain our techniques, for both the algebraic and proof complexity worlds.

\subsection{Algebraic Proof Systems}\label{sec:Nullstellensatz}

We now describe the algebraic proof systems that are the subject of this paper.  If one has a set of polynomials (called \emph{axioms}) $f_1,\ldots,f_m\in\F[x_1,\ldots,x_n]$ over some field $\F$, then (the weak version of) Hilbert's Nullstellensatz shows that the system $f_1(\vx)=\cdots=f_m(\vx)=0$ is unsatisfiable (over the algebraic closure of $\F$) if and only if there are polynomials $g_1,\ldots,g_m\in\F[\vx]$ such that $\sum_j g_j(\vx)f_j(\vx)=1$ (as a formal identity), or equivalently, that 1 is in the ideal generated by the $\{f_j\}_j$.

Beame, Impagliazzo, \Krajicek, Pitassi, and \Pudlak~\cite{BeameIKPP96} suggested to treat these $\{g_j\}_j$ as a \emph{proof} of the unsatisfiability of this system of equations, called a \emph{Nullstellensatz refutation}.  This is in particular relevant for complexity theory as one can restrict attention to \emph{boolean} solutions to this system by adding the \emph{boolean axioms}, that is, adding the polynomials $\{x_i^2-x_i\}_{i=1}^n$ to the system.  As such, one can then naturally encode $\NP$-complete problems such as the satisfiability of 3CNF formulas as the satisfiability of a system of constant-degree polynomials, and a Nullstellensatz refutation is then an equation of the form $\sum_{j=1}^m g_j(\vx)f_j(\vx)+\sum_{i=1}^n h_i(\vx)(x_i^2-x_i)=1$ for $g_j,h_i\in\F[\vx]$.  This proof system is sound (only refuting unsatisfiable systems over $\bits^n$) and complete (refuting any unsatisfiable system, by Hilbert's Nullstellensatz). 

Given that the above proof system is sound and complete, it is then natural to ask what is its power to refute unsatisfiable systems of polynomial equations over $\bits^n$.  To understand this question one must define the notion of the \emph{size} of the above refutations.  Two popular notions are that of the \emph{degree}, and the \emph{sparsity} (number of monomials). One can then show (see for example Pitassi~\cite{Pitassi97}) that for any unsatisfiable system which includes the boolean axioms, there exist a refutation where the $g_j$ are multilinear and where the $h_i$ have degree at most $O(n+d)$, where each $f_j$ has degree at most $d$.  In particular, this implies that for any unsatisfiable system with $d=O(n)$ there is a refutation of degree $O(n)$ and involving at most $\exp(O(n))$ monomials.  This intuitively agrees with the fact that $\coNP$ is a subset of non-deterministic exponential time.

Building on the suggestion of Pitassi~\cite{Pitassi97} and various investigations into the power of strong algebraic proof systems (\cite{GrigorievHirsch03,RazTzameret08a,RazTzameret08b}), Grochow and Pitassi~\cite{GrochowPitassi14} have recently considered more \emph{succinct} descriptions of polynomials where one measures the size of a polynomial by the size of an algebraic circuit needed to compute it.  This is potentially much more powerful as there are polynomials such as the determinant which are of high degree and involve exponentially many monomials and yet can be computed by small algebraic circuits. They named the resulting system the \emph{Ideal Proof System (IPS)} which we now define.

\begin{definition}[Ideal Proof System (IPS), Grochow-Pitassi~\cite{GrochowPitassi14}]\label{def:orig-IPS}
	Let $f_1(\vx),\ldots,f_m(\vx)\in\F[x_1,\ldots,x_n]$ be a system of polynomials. An \demph{IPS refutation} for showing that the polynomials $\{f_j\}_j$ have no common solution in $\bits^n$ is an algebraic circuit $C(\vx,\vy,\vz)\in\F[\vx,y_1,\ldots,y_m,z_1,\ldots,z_n]$, such that
	\begin{enumerate}
		\item $C(\vx,\vnz,\vnz) = 0$.
		\item $C(\vx,f_1(\vx),\ldots,f_m(\vx),x_1^2-x_1,\ldots,x_n^2-x_n)=1$.
	\end{enumerate}
	The \demph{size} of the IPS refutation is the size of the circuit $C$. If $C$ is of individual degree $\le 1$ in each $y_j$ and $z_i$, then this is a \demph{linear} IPS refutation (called \emph{Hilbert} IPS by Grochow-Pitassi~\cite{GrochowPitassi14}), which we will abbreviate as \lIPS. If $C$ is of individual degree $\le 1$ only in the $y_j$ then we say this is a \lbIPS refutation. If $C$ comes from a restricted class of algebraic circuits $\cC$, then this is a called a $\cC$-IPS refutation, and further called a $\cC$-\lIPS refutation if $C$ is linear in $\vy,\vz$, and a $\cC$-\lbIPS refutation if $C$ is linear in $\vy$.
\end{definition}

Notice also that our definition here by default adds the equations $\{x_i^2-x_i\}_i$ to the system $\{f_j\}_j$.  For convenience we will often denote the equations $\{x_i^2-x_i\}_i$ as $\vx^2-\vx$. One need not add the equations $\vx^2-\vx$ to the system in general, but this is the most interesting regime for proof complexity and thus we adopt it as part of our definition.

The $\cC$-IPS system is sound for any $\cC$, and Hilbert's Nullstellensatz shows that $\cC$-\lIPS is complete for any complete class of algebraic circuits $\cC$ (that is, classes which can compute any polynomial, possibly requiring exponential complexity). We note that we will also consider non-complete classes such as multilinear-formulas (which can only compute \emph{multilinear} polynomials, but are complete for multilinear polynomials), where we will show that the multilinear-formula-\lIPS system is not complete for the language of all unsatisfiable sets of multilinear polynomials (\autoref{ex:multi-form:incomplete}), while the stronger multilinear-formula-\lbIPS version is complete (\autoref{res:multilin:simulate-sparse}).  However, for the standard conversion of unsatisfiable CNFs into polynomial systems of equations, the multilinear-formula-\lIPS system is complete (\autoref{thm:GrochowPitassi14}).

Grochow-Pitassi~\cite{GrochowPitassi14} proved the following theorem, showing that the IPS system has surprising power and that lower bounds on this system give rise to \emph{computational} lower bounds.

\begin{theoremwp}[Grochow-Pitassi~\cite{GrochowPitassi14}]\label{thm:GrochowPitassi14}
	Let $\varphi=C_1\wedge\cdots\wedge C_m$ be an unsatisfiable CNF on $n$-variables, and let $f_1,\ldots,f_m\in\F[x_1,\ldots,x_m]$ be its encoding as a polynomial system of equations. If there is a size-$s$ Frege proof (resp.\ Extended Frege) that $\{f_j\}_j,\{x_i^2-x_i\}_i$ is unsatisfiable, then there is a formula-\lIPS (resp.\ circuit-\lIPS) refutation of size $\poly(n,m,s)$ that is checkable in randomized $\poly(n,m,s)$ time.\footnote{We note that Grochow and Pitassi~\cite{GrochowPitassi14} proved this for Extended Frege and circuits, but essentially the same proof relates Frege and formula size.} 
	
	Further, $\{f_j\}_j,\{x_i^2-x_i\}_i$ has a \lIPS refutation, where the refutation uses multilinear polynomials in $\VNP$.  Thus, if every IPS refutation of $\{f_j\}_j,\{x_i^2-x_i\}_i$ requires formula (resp.\ circuit) size $\ge s$, then there is an explicit polynomial (that is, in $\VNP$) that requires size $\ge s$ algebraic formulas (resp.\ circuits).
\end{theoremwp}

\begin{remark}\label{rmk:EF-degree}
	One point to note is that the transformation from Extended Frege to IPS refutations yields circuits of polynomial size but without any guarantee on their degree. In particular, such circuits can compute polynomials of exponential degree. In contrast, the conversion from Frege to IPS refutations yields polynomial sized algebraic formulas and those compute polynomials of polynomially bounded degree. This range of parameters, polynomials of polynomially bounded degree, is the more common setting studied in algebraic complexity.
\end{remark}

The fact that $\cC$-IPS refutations are efficiently checkable (with randomness) follows from the fact that we need to verify the polynomial identities stipulated by the definition.  That is, one needs to solve an instance of the \emph{polynomial identity testing (PIT)} problem for the class $\cC$: given a circuit from the class $\cC$ decide whether it computes the identically zero polynomial.  This problem is solvable in probabilistic polynomial time ($\BPP$) for general algebraic circuits, and there are various restricted classes for which deterministic algorithms are known (see \autoref{sec:PIT}). 

Motivated by the fact that PIT of non-commutative formulas\,\footnote{These are formulas over a set of non-commuting variables.} can be solved deterministically (\cite{RazShpilka05}) and admit
exponential-size lower bounds (\cite{Nisan91}), Li, Tzameret and Wang~\cite{LiTW15} have shown that IPS over \emph{non-commutative} polynomials (along with additional \emph{commutator} axioms) can simulate Frege (they also provided a quasipolynomial simulation of IPS over non-commutative formulas by Frege; see Li, Tzameret and Wang~\cite{LiTW15} for more details).

\begin{theoremwp}[Li, Tzameret and Wang~\cite{LiTW15}]\label{thm:LTW}
	Let $\varphi=C_1\wedge\cdots\wedge C_m$ be an unsatisfiable CNF on $n$-variables, and let $f_1,\ldots,f_m\in\F[x_1,\ldots,x_m]$ be its encoding as a polynomial system of equations. If there is a size-$s$ Frege proof that $\{f_j\}_j,\{x_i^2-x_i\}_i$ is unsatisfiable, then there is a non-commutative-IPS refutation of formula-size $\poly(n,m,s)$, where the commutator axioms $x_ix_j-x_jx_i$ are also included in the polynomial system being refuted. Further, this refutation is checkable in deterministic $\poly(n,m,s)$ time.
\end{theoremwp}

The above results naturally motivate studying $\cC$-IPS for various restricted classes of algebraic circuits, as lower bounds for such proofs then intuitively correspond to restricted lower bounds for the Extended Frege proof system.  In particular, as exponential lower bounds are known for non-commutative formulas (\cite{Nisan91}), this possibly suggests that such methods could even attack the full Frege system itself.  

\subsection{Algebraic Circuit Classes}\label{sec:intro:circuits}

Having motivated $\cC$-IPS for restricted circuit classes $\cC$, we now give formal definitions of the algebraic circuit classes of interest to this paper, all of which were studied independently in algebraic complexity. Some of them capture the state-of-art in our ability to prove lower bounds and provide efficient deterministic identity tests, so it is natural to attempt to fit them into the proof complexity framework. We define each and briefly explain what we know about it. As the list is long, the reader may consider skipping to the results (\autoref{sec:results}), and refer to the definitions of these classes as they arise.

Algebraic circuits and formula (over a fixed chosen field) compute polynomials via addition and multiplication gates, starting from the input variables and constants from the field. For background on algebraic circuits in general and their complexity measures we refer the reader to the survey of Shpilka and Yehudayoff~\cite{SY10}. We next define the restricted circuit classes that we will be studying in this paper. 

\subsubsection{Low Depth Classes}

We start by defining what are the simplest and most restricted classes of algebraic circuits. The first class simply represents polynomials as a sum of monomials. This is also called the \emph{sparse representation} of the polynomial. Notationally we call this model $\sumprod$ formulas (to capture the fact that polynomials computed in the class are represented simply as sums of products), but we will more often call these polynomials ``sparse''.

\begin{definition}\label{def:sparse}
	The class $\cC=\sumprod$ compute polynomials in their \demph{sparse} representation, that is, as a sum of monomials. The graph of computation has two layers with an addition gate at the top and multiplication gates at the bottom. The \demph{size} of a $\sumprod$ circuit of a polynomial $f$ is the multiplication of the number of monomials in $f$, the number of variables, and the degree.
\end{definition}

This class of circuits is what is used in the Nullstellensatz proof system. In our terminology $\sumprod$-\lIPS is exactly the Nullstellensatz proof system.

Another restricted class of algebraic circuits is that of \emph{depth-$3$ powering formulas} (sometimes also called ``diagonal depth-$3$ circuits''). We will sometimes abbreviate this name as a ``$\sumpowsum$ formula'', where $\bigwedge$ denotes the powering operation. Specifically, polynomials that are efficiently computed by small formulas from this class can be represented as sum of powers of linear functions. This model appears implicitly in Shpilka~\cite{Shpilka02} and explicitly in the work of Saxena \cite{Saxena08}.

\begin{definition}\label{def:diagonal}
	The class of depth-$3$ powering formulas, denoted $\sumpowsum$, computes polynomials of the following form $$f(\vx)=\sum_{i=1}^{s} \ell_i(\vx)^{d_i},$$ where $\ell_i(\vx)$ are linear functions. The degree of this $\sumpowsum$ representation of $f$ is $\max_i\{d_i\}$ and its size  is $n\cdot \sum_{i=1}^{s}(d_i+1)$.
\end{definition}

One reason for considering this class of circuits is that it is a simple, but non-trivial model that is somewhat well-understood.  In particular, the partial derivative method of Nisan-Wigderson~\cite{NisanWigderson96} implies lower bounds for this model and efficient polynomial identity testing algorithms are known (\cite{Saxena08,AgrawalSS13,ForbesShpilka13a,ForbesShpilka13b,ForbesSS14}, as discussed further in \autoref{sec:PIT}).

We also consider a generalization of this model where we allow powering of low-degree polynomials.

\begin{definition}\label{def:sumpowlowdeg}
	The class $\sumpow{t}$ computes polynomials of the following form 
	\[
		f(\vx)=\sum_{i=1}^{s} f_i(\vx)^{d_i}
		\;,
	\]
	where the degree of the $f_i(\vx)$ is at most $t$. The size of this representation is $\binom{n+t}{t} \cdot \sum_{i=1}^{s}(d_i+1)$.
\end{definition}

We remark that the reason for defining the size this way is that we think of the $f_i$ as represented as sum of monomials (there are $\binom{n+t}{t}$ $n$-variate monomials of degree at most $t$) and the size captures the complexity of writing this as an algebraic formula. This model is the simplest that requires the method of \emph{shifted partial derivatives} of Kayal~\cite{Kayal12,GuptaKKS14} to establish lower bounds, and this has recently been generalized to obtain polynomial identity testing algorithms (\cite{Forbes15}, as discussed further in \autoref{sec:PIT}).

\subsubsection{Oblivious Algebraic Branching Programs}

Algebraic branching programs (ABPs) form a model whose computational power lies between that of algebraic circuits and algebraic formulas, and certain \emph{read-once} and \emph{oblivious} ABPs are a natural setting for studying the \emph{partial derivative matrix} lower bound technique of Nisan~\cite{Nisan91}.

\begin{definition}[Nisan~\cite{Nisan91}]\label{def:roABP}
	An \demph{algebraic branching program (ABP) with unrestricted weights} of \demph{depth} $D$ and \demph{width} $\le r$, on the variables $x_1,\ldots,x_n$, is a directed acyclic graph such that:
	\begin{itemize}
		\item The vertices are partitioned in $D+1$ layers $V_0,\ldots,V_D$, so that $V_0=\{s\}$ ($s$ is the source node), and $V_D=\{t\}$ ($t$ is the sink node). Further, each edge goes from $V_{i-1}$ to $V_{i}$ for some $0< i\le D$.
		\item $\max|V_i|\le r$.
		\item Each edge $e$ is weighted with a polynomial $f_e\in\F[\vx]$.
	\end{itemize}
	The \demph{(individual) degree} $d$ of the ABP is the maximum (individual) degree of the edge polynomials $f_e$. The \demph{size} of the ABP is the product $n\cdot r\cdot d\cdot D$, 

	Each $s$-$t$ path is said to compute the polynomial which is the product of the labels of its edges, and the algebraic branching program itself computes the sum over all $s$-$t$ paths of such polynomials.

	There are also the following restricted ABP variants.
	\begin{itemize}
		\item An algebraic branching program is said to be \demph{oblivious} if for every layer $\ell$, all the edge labels in that layer are univariate polynomials in a single variable $x_{i_\ell}$.
		\item An oblivious branching program is said to be a \demph{read-once} oblivious ABP (roABP) if each $x_i$ appears in the edge label of exactly one layer, so that $D=n$. That is, each $x_i$ appears in the edge labels in at exactly one layer. The layers thus define a \demph{variable order}, which will be $x_1<\cdots<x_n$ if not otherwise specified.
		\item An oblivious branching program is said to be a \demph{read-$k$} oblivious ABP if each variable $x_i$ appears in the edge labels of exactly $k$ layers, so that $D=kn$.
		\item An ABP is \textbf{non-commutative} if each $f_e$ is from the ring $\F\nc{\vx}$ of non-commuting variables and has $\deg f_e\le 1$, so that the ABP computes a non-commutative polynomial.
			\qedhere
	\end{itemize}
\end{definition}

Intuitively, roABPs are the algebraic analog of read-once boolean branching programs, the non-uniform model of the class $\RL$, which are well-studied in boolean complexity.  Algebraically, roABPs are also well-studied. In particular, roABPs are essentially equivalent to non-commutative ABPs (\cite{ForbesShpilka13b}), a model at least as strong as non-commutative formulas.  That is, as an roABP reads the variables in a fixed order (hence not using commutativity) it can be almost directly interpreted as a non-commutative ABP\@. Conversely, as non-commutative multiplication is ordered, one can interpret a non-commutative polynomial in a read-once fashion by (commutatively) exponentiating a variable to its index in a monomial.  For example, the non-commutative $xy-yx$ can be interpreted commutatively as $x^1y^2-y^1x^2=xy^2-x^2y$, and one can show that this conversion preserves the relevant ABP complexity (\cite{ForbesShpilka13b}). The study of non-commutative ABPs dates to Nisan~\cite{Nisan91}, who proved lower bounds for non-commutative ABPs (and thus also for roABPs, in any order). In a sequence of more recent papers, polynomial identity testing algorithms were devised for roABPs (\cite{RazShpilka05,ForbesShpilka12,ForbesShpilka13b,ForbesSS14,AgrawalGKS15}, see also \autoref{sec:PIT}). In terms of proof complexity, Tzameret~\cite{Tzameret11} studied a proof system with lines given by roABPs, and recently Li, Tzameret and Wang~\cite{LiTW15} (\autoref{thm:LTW}) showed that IPS over non-commutative formulas is essentially equivalent in power to the Frege proof system.  Due to the close connections between non-commutative ABPs and roABPs, this last result suggests the importance of proving lower bounds for roABP-IPS as a way of attacking lower bounds for the Frege proof system (although our work obtains roABP-\lIPS lower bounds without obtaining non-commutative-\lIPS lower bounds).

Finally, we mention that recently Anderson, Forbes, Saptharishi, Shpilka, and Volk~\cite{AndersonFSSV16} obtained exponential lower bounds for read-$k$ oblivious ABPs (when $k=o(\log n/\log\log n)$) as well as a slightly subexponential polynomial identity testing algorithm.

\subsubsection{Multilinear Formulas}

The last model that we consider is that of multilinear formulas. 

\begin{definition}[Multilinear formula]\label{def-ml-fmla}
	An algebraic formula is a \demph{multilinear formula} if the polynomial computed by \emph{each} gate of the formula is multilinear (as a formal polynomial, that is, as an element of $\,\mathbb{F}[x_1,\ldots,x_n]$). The \demph{product depth} is the maximum number of multiplication gates on any input-to-output path in the formula.
\end{definition}

Raz~\cite{Raz09,Raz06} proved quasi-polynomial lower bounds for multilinear formulas and separated multilinear formulas from multilinear circuits. Raz and Yehudayoff proved exponential lower bounds for small depth multilinear formulas \cite{RazYehudayoff09}. Only slightly sub-exponential polynomial identity testing algorithms are known for small-depth multilinear formulas (\cite{OliveiraSV15}).

\subsection{Our Results and Techniques}\label{sec:results}

We now briefly summarize our results and techniques, stating some results in less than full generality to more clearly convey the result.
We present the results in the order that we later prove them. We start by giving upper bounds for the IPS (\autoref{sec:results:upper}). We then describe our functional lower bounds and the \lIPS lower bounds they imply (\autoref{sec:results:funct}). Finally, we discuss lower bounds for multiples and state our lower bounds for IPS (\autoref{sec:intro:lbs-mult}).

\subsubsection{Upper Bounds for Proofs within Subclasses of IPS}\label{sec:results:upper}

Various previous works have studied restricted algebraic proof systems and shown non-trivial upper bounds.  The general simulation (\autoref{thm:GrochowPitassi14}) of Grochow and Pitassi~\cite{GrochowPitassi14} showed that the formula-IPS and circuit-IPS systems can simulate powerful proof systems such as Frege and Extended Frege, respectively. The work of Li, Tzameret and Wang~\cite{LiTW15} (\autoref{thm:LTW}) show that even non-commutative-formula-IPS can simulate Frege.  The work of Grigoriev and Hirsch~\cite{GrigorievHirsch03} showed that proofs manipulating depth-3 algebraic formulas can refute hard axioms such as the \emph{pigeonhole principle}, the \emph{subset-sum axiom}, and \emph{Tseitin tautologies}. The work of Raz and Tzameret~\cite{RazTzameret08a,RazTzameret08b} somewhat strengthened their results by restricting the proof to depth-3 \emph{multilinear} proofs (in a \emph{dynamic} system, see \autoref{sec:alg-proofs}).  

However, these upper bounds are for proof systems (IPS or otherwise) for which no proof lower bounds are known.  As such, in this work we also study upper bounds for restricted subsystems of IPS.  In particular, we compare linear-IPS versus the full IPS system, as well as showing that even for restricted $\cC$, $\cC$-IPS can refute interesting unsatisfiable systems of equations arising from $\NP$-complete problems (and we will obtain corresponding proof lower bounds for these $\cC$-IPS systems).

Our first upper bound is to show that linear-IPS can simulate the full IPS proof system when the axioms are computationally simple, which essentially resolves a question of Grochow and Pitassi~\cite[Open Question 1.13]{GrochowPitassi14}.

\begin{theoremwp*}[\autoref{res:h-ips_v_ips}]
	For $|\F|\ge\poly(d)$, if $f_1,\ldots,f_m\in\F[x_1,\ldots,x_n]$ are degree-$d$ polynomials computable by size-$s$ algebraic formulas (resp.\ circuits) and they have a size-$t$ formula-IPS (resp.\ circuit-IPS) refutation, then they also have a size-$\poly(d,s,t)$ formula-\lIPS (resp.\ circuit-\lIPS) refutation.
\end{theoremwp*}

This theorem is established by pushing the ``non-linear'' dependencies on the axioms into the IPS refutation itself, which is possible as the axioms are assumed to themselves be computable by small circuits.  We note that Grochow and Pitassi~\cite{GrochowPitassi14} showed such a conversion, but only for IPS refutations computable by sparse polynomials.  Also, we remark that this result holds even for circuits of unbounded degree, as opposed to just those of polynomial degree.

We then turn our attention to IPS involving only restricted classes of algebraic circuits, and show that they are complete proof systems.  This is clear for complete models of algebraic circuits such as sparse polynomials, depth-3 powering formulas\,\footnote{Showing that depth-3 powering formulas are complete (in large characteristic) can be seen from the fact that any multilinear monomial can be computed in this model, see for example Fischer~\cite{Fischer94}.} and roABPs.  The models of sparse-\lIPS and roABP-\lIPS can efficiently simulate the Nullstellensatz proof system measured in terms of number of monomials, as the former is equivalent to this system, and the latter follows as sparse polynomials have small roABPs.  Note that depth-3 powering formulas cannot efficiently compute sparse polynomials in general (\autoref{res:lbs-mult:sumpowsum}) so cannot efficiently simulate the Nullstellensatz system. For multilinear formulas, showing completeness (much less an efficient simulation of sparse-\lIPS) is more subtle as not every polynomial is multilinear, however the result can be obtained by a careful multilinearization. 

\begin{theoremwp*}[\autoref{ex:multi-form:incomplete}, \autoref{res:multilin:simulate-sparse}]
	The proof systems of sparse-\lIPS, $\sumpowsum$-\lIPS (in large characteristic fields), and roABP-\lIPS are complete proof systems (for systems of polynomials with no boolean solutions).  The multilinear-formula-\lIPS proof system is not complete, but the depth-2 multilinear-formula-\lbIPS proof system is complete (for multilinear axioms) and can polynomially simulate sparse-\lIPS (for low-degree axioms).
\end{theoremwp*}

However, we recall that multilinear-formula-\lIPS \emph{is} complete when refuting unsatisfiable CNF formulas (\autoref{thm:GrochowPitassi14}).

We next consider the equation $\sum_{i=1}^n \alpha_i x_i-\beta$ along with the boolean axioms $\{x_i^2-x_i\}_i$.  Deciding whether this system of equations is satisfiable is the $\NP$-complete \emph{subset-sum} problem, and as such we do not expect small refutations in general (unless $\NP=\coNP$). Indeed, Impagliazzo, \Pudlak, and Sgall~\cite{IPS99} (\autoref{thm:IPS99}) have shown lower bounds for refutations in the \emph{polynomial calculus} system (and thus also the Nullstellensatz system) even when $\vaa=\vno$. Specifically, they showed that such refutations require both $\Omega(n)$-degree and $\exp(\Omega(n))$-many monomials, matching the worst-case upper bounds for these complexity measures.  In the language of this paper, they gave $\exp(\Omega(n))$-size lower bounds for refuting this system in $\sum\prod$-\lIPS (which is equivalent to the Nullstellensatz proof system).  In contrast, we establish here $\poly(n)$-size refutations for $\vaa=\vno$ in the restricted proof systems of roABP-\lIPS and \emph{depth-3} multilinear-formula-\lIPS (despite the fact that multilinear-formula-\lIPS is not complete).

\begin{theoremwp*}[\autoref{res:ips-ubs:subset:roABP}, \autoref{res:ips-ubs:subset:mult-form}]
	Let $\F$ be a field of characteristic $\chara(\F)>n$. Then the system of polynomial equations $\sum_{i=1}^n x_i-\beta$, $\{x_i^2-x_i\}_{i=1}^n$ is unsatisfiable for $\beta\in\F\setminus\{0,\ldots,n\}$, and there are explicit $\poly(n)$-size refutations within roABP-\lIPS, as well as within depth-3 multilinear-formula-\lIPS.
\end{theoremwp*}

This theorem is proven by noting that the polynomial $p(t)\eqdef\prod_{k=0}^n (t-k)$ vanishes on $\sum_i x_i$ modulo $\{x_i^2-x_i\}_{i=1}^n$, but $p(\beta)$ is a nonzero constant.  This implies that $\sum_i x_i-\beta$ divides $p(\sum_i x_i)-p(\beta)$.  Denoting the quotient by $f(\vx)$, it follows that $\frac{1}{-p(\beta)}\cdot f(\vx)\cdot (\sum_i x_i-\beta)\equiv 1\mod \{x_i^2-x_i\}_{i=1}^n$, which is nearly a linear-IPS refutation except for the complexity of establishing this relation over the boolean cube.  We show that the quotient $f$ is easily expressed as a depth-3 powering circuit.  Unfortunately, proving the above equivalence to 1 modulo the boolean cube is not possible in the depth-3 powering circuit model.  However, by moving to more powerful models (such as roABPs and multilinear formulas) we can give proofs of this multilinearization to 1 and thus give proper IPS refutations.

\subsubsection{Linear-IPS Lower Bounds via Functional Lower Bounds}\label{sec:results:funct}

Having demonstrated the power of various restricted classes of IPS refutations by refuting the subset-sum axiom, we now turn to lower bounds.  We give two paradigms for establishing lower bounds, the first of which we discus here, named a \emph{functional circuit lower bound}. This idea appeared in the work of Grigoriev and Razborov~\cite{GrigorievRazborov00} as well as in the recent work of Forbes, Kumar and Saptharishi~\cite{ForbesKS16}. We briefly motivate this type of lower bound as a topic of independent interest in algebraic circuit complexity, and then discuss the lower bounds we obtain and their implications to obtaining proof complexity lower bounds.

In boolean complexity, the primary object of interest are \emph{functions}.  Generalizing slightly, one can even seek to compute functions $f:\bits^n\to\F$ for some field $\F$.  In contrast, in algebraic complexity one seeks to compute \emph{polynomials} as elements of the ring $\F[x_1,\ldots,x_n]$. These two regimes are tied by the fact that every polynomial $f\in\F[\vx]$ induces a function $\hat{f}:\bits^n\to \F$ via the evaluation $\hat{f}:\vx\mapsto f(\vx)$.  That is, the polynomial $f$ \emph{functionally computes} the function $\hat{f}$.  As an example, $x^2-x$ functionally computes the zero function despite being a nonzero polynomial.

Traditional algebraic circuit lower bounds for the $n\times n$ permanent are lower bounds for computing $\perm_n$ as an element in the ring $\F[\{x_{i,j}\}_{1\le i,j\le n}]$. This is a strong notion of ``computing the permanent'', while one can consider the weaker notion of functionally computing the permanent, that is, a polynomial $f\in\F[\{x_{i,j}\}]$ such that $f=\perm_n$ over $\bits^{n\times n}$, where $f$ is not required to equal $\perm_n$ as a polynomial.  As $\perm_n:\bits^{n\times n}\to\F$ is $\#\P$-hard (for fields of large characteristic), assuming plausible conjectures (such as the polynomial hierarchy being infinite) it follows that \emph{any} polynomial $f$ which functionally computes $\perm_n$ must require large algebraic circuits.  Unconditionally obtaining such a result is what we term a \emph{functional lower bound}.

\begin{goal}[Functional Circuit Lower Bound (\cite{GrigorievRazborov00,ForbesKS16})]
	Obtain an explicit function $\hat{f}:\bits^n\to\F$ such that for any polynomial $f\in\F[x_1,\ldots,x_n]$ satisfying $f(\vx)=\hat{f}(\vx)$ for all $\vx\in\bits^n$, it must be that $f$ requires large algebraic circuits.
\end{goal}

Obtaining such a result is challenging, in part because one must lower bound \emph{all} polynomials agreeing with the function $\hat{f}$ (of which there are infinitely many).  Prior work (\cite{GrigorievKarpinski98,GrigorievRazborov00,KumarSaptharishi15}) has established functional lower bounds for functions when computing with polynomials over constant-sized finite fields, and the recent work of Forbes, Kumar and Saptharishi~\cite{ForbesKS16} has established some lower bounds for any field. 

While it is natural to hope that existing methods would yield such lower bounds, many lower bound techniques inherently use that algebraic computation is \emph{syntactic}.  In particular, techniques involving partial derivatives (which include the partial derivative method of Nisan-Wigderson~\cite{NisanWigderson96} and the shifted partial derivative method of Kayal~\cite{Kayal12,GuptaKKS14}) cannot as is yield functional lower bounds as 
knowing a polynomial on $\bits^n$ is not enough  to conclude information about its partial derivatives.

We now explain how functional lower bounds imply lower bounds for linear-IPS refutations in certain cases.  Suppose one considers refutations of the unsatisfiable polynomial system $f(\vx),\{x_i^2-x_i\}_{i=1}^n$.  A linear-IPS refutation would yield an equation of the form $g(\vx)\cdot f(\vx)+\sum_i h_i(\vx)\cdot (x_i^2-x_i)=1$ for some polynomials $g,h_i\in\F[\vx]$.  Viewing this equation modulo the boolean cube, we have that $g(\vx)\cdot f(\vx)\equiv 1\mod \{x_i^2-x_i\}_i$.  Equivalently, since $f(\vx)$ is unsatisfiable over $\bits^n$, we see that $g(\vx)=\nicefrac{1}{f(\vx)}$ for $\vx\in\bits^n$, as $f(\vx)$ is never zero so this fraction is well-defined.  It follows that \emph{if} the function $\vx\mapsto \nicefrac{1}{f(\vx)}$ induces a functional lower bound then $g(\vx)$ must require large complexity, yielding the desired linear-IPS lower bound.

Thus, it remains to instantiate this program.  While we are successful, we should note that this approach as is seems to only yield proof complexity lower bounds for systems with one non-boolean axiom and thus cannot encode polynomial systems where each equation depends on $O(1)$ variables (such as those naturally arising from 3CNFs).

Our starting point is to observe that the subset-sum axiom already induces a weak form of functional lower bound, where the complexity is measured by degree.

\begin{theoremwp*}[\autoref{res:subsetsum:deg:ge}]
	Let $\F$ be a field of a characteristic at least $\poly(n)$ and $\beta\notin\{0,\ldots,n\}$. Then $\sum_i x_i-\beta,\{x_i^2-x_i\}_i$ is unsatisfiable and any polynomial $f\in\F[x_1,\ldots,x_n]$ with $f(\vx)=\frac{1}{\sum_i x_i-\beta}$ for $\vx\in\bits^n$, satisfies $\deg f\ge n$.
\end{theoremwp*}

A lower bound of $\ceil{\frac{n}{2}}+1$ was previously established by Impagliazzo, \Pudlak, and Sgall~\cite{IPS99} (\autoref{thm:IPS99}), but the bound of `$n$' (which is tight) will be crucial for our results.

We then lift this result to obtain lower bounds for stronger models of algebraic complexity. In particular, by replacing ``$x_i$'' with ``$x_iy_i$'' we show that the function $\frac{1}{\sum_i x_iy_i-\beta}$ has maximal \emph{evaluation dimension} between $\vx$ and $\vy$, which is some measure of interaction between the variables in $\vx$ and those in $\vy$ (see \autoref{sec:eval-dim}).  This measure is essentially \emph{functional}, so that one can lower bound this measure by understanding the functional behavior of the polynomial on finite sets such as the boolean cube. Our lower bound for evaluation dimension follows by examining the above degree bound. Using known relations between this complexity measure and algebraic circuit classes, we can obtain lower bounds for depth-3 powering linear-IPS.

\begin{theoremwp*}[\autoref{res:lbs-fn:lbs-ips:fixed-order}]
	Let $\F$ be a field of characteristic $\ge\poly(n)$ and $\beta\notin\{0,\ldots,n\}$.  Then $\sum_{i=1}^n x_iy_i-\beta,\{x_i^2-x_i\}_i,\{y_i^2-y_i\}_i$ is unsatisfiable and any $\sumpowsum$-\lIPS refutation requires size $\ge\exp(\Omega(n))$.
\end{theoremwp*}

The above axiom only gets maximal interaction between the variables across a \emph{fixed} partition of the variables.  By introducing auxiliary variables we can create such interactions in variables across \emph{any} partition of (some) of the variables.  By again invoking results showing such structure implies computational hardness we obtain more linear-IPS lower bounds.

\begin{theoremwp*}[\autoref{res:lbs-fn:lbs-ips:vary-order}]
	Let $\F$ be a field of characteristic $\ge\poly(n)$ and $\beta\notin\{0,\ldots,\binom{2n}{2}\}$. Then $\sum_{i<j} z_{i,j}x_ix_j-\beta,\{x_i^2-x_i\}_{i=1}^n,\{z_{i,j}^2-z_{i,j}\}_{i<j}$ is unsatisfiable, and any roABP-\lIPS refutation (in any variable order) requires $\exp(\Omega(n))$-size. Further, any multilinear-formula-IPS refutation requires $n^{\Omega(\log n)}$-size, and any product-depth-$d$ multilinear-formula-IPS refutation requires $n^{\Omega((\nicefrac{n}{\log n})^{1/d}/d^2)}$-size.
\end{theoremwp*}

Note that our result for multilinear-formulas is not just for the linear-IPS system, but actually for the full multilinear-formula-IPS system. Thus, we show that even though roABP-\lIPS and depth-3 multilinear formula-\lbIPS can refute the subset-sum axiom in polynomial size, slight variants of this axiom do not have polynomial-size refutations.

\subsubsection{Lower Bounds for Multiples}\label{sec:intro:lbs-mult}

While the above paradigm can establish super-polynomial lower bounds for \emph{linear}-IPS, it does not seem able to establish lower bounds for the general IPS proof system over non-multilinear polynomials, even for restricted classes.  This is because such systems would induce equations such as $h(\vx) f(\vx)^2+g(\vx) f(\vx)\equiv 1 \mod \{x_i^2-x_i\}_{i=1}^n$, where we need to design a computationally simple axiom $f$ so that this equation implies at least one of $h$ or $g$ is of large complexity.  In a linear-IPS proof it must be that $h$ is zero, so that for any $\vx\in\bits^n$ we can solve for $g(\vx)$, that is, $g(\vx)=\nicefrac{1}{f(\vx)}$.  However, in general knowing $f(\vx)$ does not uniquely determine $g(\vx)$ or $h(\vx)$, which makes this approach significantly more complicated.  Further, even though we can efficiently simulate IPS by linear-IPS (\autoref{res:h-ips_v_ips}) in general, this simulation increases the complexity of the proof so that even if one started with a $\cC$-IPS proof for a restricted circuit class $\cC$ the resulting \lIPS proof may not be in $\cC$-\lIPS.

As such, we introduce a second paradigm, called \emph{lower bounds for multiples}, which can yield $\cC$-IPS lower bounds for various restricted classes $\cC$.  We begin by defining this question.

\begin{goal}[Lower Bounds for Multiples]
	Design an explicit polynomial $f(\vx)$ such that for any nonzero $g(\vx)$ we have that the multiple $g(\vx)f(\vx)$ is hard to compute.
\end{goal}

We now explain how such lower bounds yield IPS lower bounds.  Consider the system $f,\{x_i^2-x_i\}_i$ with a single non-boolean axiom.  An IPS refutation is a circuit $C(\vx,y,\vz)$ such that $C(\vx,0,\vnz)=0$ and $C(\vx,f,\vx^2-\vx)=1$, where (as mentioned) $\vx^2-\vx$ denotes $\{x_i^2-x_i\}_i$.  Expressing $C(\vx,f,\vx^2-\vx)$ as a univariate in $f$, we obtain that $\sum_{i\ge 1} C_i(\vx,\vx^2-\vx) f^i = 1-C(\vx,0,\vx^2-\vx)$ for some polynomials $C_i$.  For most natural measures of circuit complexity $1-C(\vx,0,\vx^2-\vx)$ has complexity roughly bounded by that of $C$ itself.  Thus, we see that a multiple of $f$ has a small circuit, as $\left(\sum_{i\ge 1} C_i(\vx,\vx^2-\vx) f^{i-1}\right)\cdot f=1-C(\vx,0,\vx^2-\vx)$, and one can use the properties of the IPS refutation to show this is in fact a \emph{nonzero} multiple.  Thus, if we can show that all nonzero multiples of $f$ require large circuits then we rule out a small IPS refutation.

We now turn to methods for obtaining polynomials with hard multiples. Intuitively if a polynomial $f$ is hard then so should small modifications such as $f^2+x_1f$, and this intuition is supported by the result of Kaltofen~\cite{Kaltofen89} which shows that if a polynomial has a small algebraic circuit then so do all of its factors.  As a consequence, if a polynomial requires super-polynomially large algebraic circuits then so do all of its multiples. However, Kaltofen's~\cite{Kaltofen89} result is about \emph{general} algebraic circuits, and there are very limited results about the complexity of factors of \emph{restricted} algebraic circuits (\cite{DvirSY09,Oliveira15}) so that obtaining polynomials for hard multiples via factorization results seems difficult.

However, note that lower bound for multiples has a different order of quantifiers than the factoring question.  That is, Kaltofen's~\cite{Kaltofen89} result speaks about the factors of \emph{any} small circuit, while the lower bound for multiples speaks about the multiples of a \emph{single} polynomial.  Thus, it seems plausible that existing methods could yield such explicit polynomials, and indeed we show this is the case.

We begin by noting that obtaining lower bounds for multiples is a natural instantiation of the algebraic \emph{hardness versus randomness} paradigm. In particular, Heintz-Schnorr~\cite{HeintzSchnorr80} and Agrawal~\cite{Agrawal05} showed that obtaining deterministic (black-box) polynomial identity testing algorithms implies lower bounds (see \autoref{sec:PIT} for more on PIT), and we strengthen that connection here to lower bounds for multiples.  We can actually instantiate this connection, and we use slight modifications of existing PIT algorithms to show that multiples of the determinant are hard in some models.

\begin{theoremwp*}[Informal Version of \autoref{res:generator_to_lbs-mult}, \autoref{res:lbs-mult:pit:det}]
	Let $\cC$ be a restricted class of $n$-variate algebraic circuits. Full derandomization of PIT algorithms for $\cC$ yields a (weakly) explicit polynomial whose nonzero multiples require $\exp(\Omega(n))$-size as $\cC$-circuits. 

	In particular, when $\cC$ is the class of sparse polynomials, depth-3 powering formulas, $\sumpowc$ formulas (in characteristic zero), or ``every-order'' roABPs, then all nonzero multiples of the $n\times n$ determinant are $\exp(\Omega(n))$-hard in these models.
\end{theoremwp*}

The above statement shows that \emph{derandomization} implies \emph{hardness}.  We also partly address the converse direction by arguing (\autoref{sec:lbs-mult:hard-v-rand}) that hardness-to-randomness construction of Kabanets and Impagliazzo~\cite{KabanetsImpagliazzo04} only requires lower bounds for multiples to derandomize PIT\@.  Unfortunately, this direction is harder to instantiate for restricted classes as it requires lower bounds for classes with suitable closure properties.\footnote{Although, we note that one can instantiate this connection with depth-3 powering formulas (or even $\sumpowc$ formulas) using the lower bounds for multiples developed in this paper, building on the work of Forbes~\cite{Forbes15}.  However, the resulting PIT algorithms are worse than those developed by Forbes~\cite{Forbes15}.}

Unfortunately the above result is slightly unsatisfying from a proof complexity standpoint as the (exponential-size) lower bounds for the subclasses of IPS one can derive from the above result would involve the determinant polynomial as an axiom.  While the determinant is efficiently computable, it is not computable by the above restricted circuit classes (indeed, the above result proves that).  As such, this would not fit the real goal of proof complexity which seeks to show that there are statements whose proofs must be \emph{super-polynomial larger} than the length of the statement.  Thus, if we measure the size of the IPS proof and the axioms with respect to the same circuit measure, the lower bounds for multiples approach \emph{cannot} establish such super-polynomial lower bounds. 

However, we believe that lower bounds for multiples could lead, with further ideas, to proof complexity lower bounds in the conventional sense.  That is, it seems plausible that by adding \emph{extension variables} we can convert complicated axioms to simple, local axioms by tracing through the computation of that axiom.  That is, consider the axiom $xyzw$.  This can be equivalently written as $\{a-xy,b-zw,c-ab,c\}$, where this conversion is done by considering a natural algebraic circuit for $xyzw$, replacing each gate with a new variable, and adding an axiom ensuring the new variables respect the computation of the circuit.  While we are unable to understand the role of extension variables in this work, we aim to give as simple axioms as possible whose multiples are all hard as this may facilitate future work on extension variables.

We now discuss the lower bounds for multiples we obtain.\footnote{While we discussed functional lower bounds for multilinear formulas, this class is not interesting for the lower bounds for multiples question.  This is because a multiple of a multilinear polynomial may not be multilinear, and thus clearly cannot have a multilinear formula.}

\begin{theoremwp*}[Corollaries~\ref{res:lbs-mult:sumpowsum}, \ref{res:lbs-mult:sumpowt}, \ref{res:lbs-mult:sparse-LM}, \ref{prop:good-f-roabp}, and \ref{cor:r2abp-multiples}]
	We obtain the following lower bounds for multiples.

	\begin{itemize}
		\item All nonzero multiples of $x_1\cdots x_n$ require $\exp(\Omega(n))$-size as a depth-3 powering formula (over any field), or as a $\sumpowc$ formula (in characteristic zero).
		\item All nonzero multiples of $(x_1+1)\cdots (x_n+1)$ require $\exp(\Omega(n))$-many monomials.
		\item All nonzero multiples of $\prod_i (x_i+y_i)$ require $\exp(\Omega(n))$-width as an roABP in any variable order where $\vx$ precedes $\vy$.
		\item All nonzero multiples of $\prod_{i<j} (x_i+x_j)$ require $\exp(\Omega(n))$-width as an roABP in any variable order, as well as $\exp(\Omega(n))$-width as a read-twice oblivious ABP.
			\qedhere
	\end{itemize}
\end{theoremwp*}

We now briefly explain our techniques for obtaining these lower bounds, focusing on the simplest case of depth-3 powering formulas.  It follows from the partial derivative method of Nisan and Wigderson~\cite{NisanWigderson94} (see Kayal~\cite{Kayal08}) that such formulas require exponential size to compute the monomial $x_1\ldots x_n$ \emph{exactly}.  Forbes and Shpilka~\cite{ForbesShpilka13a}, in giving a PIT algorithm for this class, showed that this lower bound can be \emph{scaled down} and \emph{made robust}.  That is, if one has a size-$s$ depth-3 powering formula, it follows that \emph{if} it computes a monomial $x_{i_1}\cdots x_{i_\ell}$ for distinct $i_j$ then $\ell\le O(\log s)$ (so the lower bound is scaled down).  One can then show that regardless of what this formula actually computes the \emph{leading} monomial $x_{i_1}^{a_{i_1}}\cdots x_{i_\ell}^{a_{i_\ell}}$ (for distinct $i_j$ and positive $a_{i_j}$) must have that $\ell\le O(\log s)$.  One then notes that leading monomials are \emph{multiplicative}.  Thus, for any nonzero $g$ the leading monomial of $g\cdot x_1\ldots x_n$ involves $n$ variables so that if $g\cdot x_1\ldots x_n$ is computed in size-$s$ then $n\le O(\log s)$, giving $s\ge\exp(\Omega(n))$ as desired.  One can then obtain the other lower bounds using the same idea, though for roABPs one needs to define a leading \emph{diagonal} (refining an argument of  Forbes-Shpilka~\cite{ForbesShpilka12}).

We now conclude our IPS lower bounds.

\begin{theoremwp*}[\autoref{res:lbs-ips:mult:sumpowsum}, \autoref{res:lbs-ips:mult:roABP}]
	We obtain the following lower bounds for subclasses of IPS.

	\begin{itemize}
		\item In characteristic zero, the system of polynomials $x_1\cdots x_n,x_1+\cdots+x_n-n,\{x_i^2-x_i\}_{i=1}^n$ is unsatisfiable, and any $\sumpowsum$-IPS refutation requires $\exp(\Omega(n))$-size.

		\item In characteristic $>n$, the system of polynomials, $\prod_{i<j}(x_i+x_j-1),x_1+\cdots+x_n-n,\{x_i^2-x_i\}_i$ is unsatisfiable, and any roABP-IPS refutation (in any variable order) must be of size $\exp(\Omega(n))$.
			\qedhere
	\end{itemize}
\end{theoremwp*}

Note that the first result is a non-standard encoding of $1=\AND(x_1,\ldots,x_n)=0$. Similarly, the second is a non-standard encoding of $\AND(x_1,\ldots,x_n)=1$ yet $\XOR(x_i,x_j)=1$ for all $i,j$.

\subsection{Organization}

The rest of the paper is organized as follows. \autoref{sec:notation} contains the basic notation for the paper. In \autoref{sec:background} we give background from algebraic complexity, including several important complexity measures such as coefficient dimension and evaluation dimension (see \autoref{sec:coeff-dim} and \autoref{sec:eval-dim}). We present our upper bounds for IPS in \autoref{sec:h-ips:ubs}. In \autoref{sec:lbs-fn} we give our functional lower bounds and from them obtain lower bounds for \lIPS. \autoref{sec:lbs-mult} contains our lower bounds for multiples of polynomials and in \autoref{sec:ips-mult} we derive lower bounds for IPS using them.  In \autoref{sec:open-problems} we list some problems which were left open by this work.

In \autoref{sec:alg-proofs} we describe various other algebraic proof systems and their relations to IPS, such as the dynamic Polynomial Calculus of Clegg, Edmonds, and Impagliazzo~\cite{CleggEI96}, the ordered formula proofs of Tzameret~\cite{Tzameret11}, and the multilinear proofs of Raz and Tzameret~\cite{RazTzameret08a}. In \autoref{sec:appendix} we give an explicit description of a multilinear polynomial occurring in our IPS upper bounds.

\section{Notation}\label{sec:notation}

In this section we briefly describe notation used in this paper. We denote $[n]\eqdef\{1,\ldots,n\}$.  For a vector $\va\in\N^n$, we denote $\vx^\va\eqdef x_1^{a_1}\cdots x_n^{a_n}$ so that in particular $\vx^\vno=\prod_{i=1}^n x_i$. The (total) degree of a monomial $\vx^\va$, denoted $\deg\vx^\va$, is equal to $\ellone{\va}\eqdef\sum_i a_i$, and the individual degree, denoted $\ideg \vx^\va$, is equal to $\ellinfty{\va}\eqdef\max\{a_i\}_i$. A monomial $\vx^\va$ depends on $\ellzero{\va}\eqdef|\{i:a_i\ne 0\}|$ many variables. Degree and individual degree can be defined for a polynomial $f$, denoted $\deg f$ and $\ideg f$ respectively, by taking the maximum over all monomials with nonzero coefficients in $f$. We will sometimes compare vectors $\va$ and $\vb$ as ``$\va\le\vb$'', which is to be interpreted coordinate-wise.  We will use $\prec$ to denote a monomial order on $\F[\vx]$, see \autoref{sec:mon-ord}.

Polynomials will often be written out in their monomial expansion.  At various points we will need to extract coefficients from polynomials.  When ``taking the coefficient of $\vy^\vb$ in $f\in\F[\vx,\vy]$'' we mean that both $\vx$ and $\vy$ are treated as variables and thus the coefficient returned is a scalar in $\F$, and this will be denoted $\coeff{\smash{\vy^\vb}}(f)$.  However, when ``taking the coefficient of $\vy^{\vb}$ in $f\in\F[\vx][\vy]$'' we mean that $\vx$ is now part of the ring of scalars, so the coefficient will be an element of $\F[\vx]$, and this coefficient will be denoted $\coeff{\vx|\vy^\vb}(f)$.

For a vector $\va\in\N^n$ we denote $\va_{\le i}\in\N^i$ to be the restriction of $\va$ to the first $i$ coordinates. For a set $S \subseteq [n]$ we let $\overline{S}$ denote the complement set.  We will denote the size-$k$ subsets of $[n]$ by $\binom{[n]}{k}$.  We will use $\ml:\F[\vx]\to\F[\vx]$ to denote the multilinearization operator, defined by \autoref{fact:multilinearization}.  We will use $\vx^2-\vx$ to denote the set of equations $\{x_i^2 -x_i\}_i$. 

To present algorithms that are field independent, this paper works in a model of computation where field operations (such as addition, multiplication, inversion and zero-testing) over $\F$ can be computed at unit cost, see for example Forbes~\cite[Appendix A]{Forbes14}. We say that an algebraic circuit is \demph{$t$-explicit} if it can be constructed in $t$ steps in this unit-cost model.

\section{Algebraic Complexity Theory Background}\label{sec:background}

In this section we state some known facts regarding the algebraic circuit classes that we will be studying. We also give some important definitions that will be used later in the paper. 

\subsection{Polynomial Identity Testing}\label{sec:PIT}

In the \emph{polynomial identity testing (PIT)} problem, we are given an algebraic circuit computing some polynomial $f$, and we have to determine whether ``$f\equiv 0$''.  That is, we are asking whether $f$ is the zero polynomial in $\F[x_1,\ldots,x_n]$. By the Schwartz-Zippel-DeMillo-Lipton Lemma~\cite{Zippel79,Schwartz80,DeMilloLipton78}, if $0\ne f \in \F[\vx]$ is a polynomial of degree $\le d$ and $S\subseteq \F$, and $\vaa\in S^n$ is chosen uniformly at random, then $f(\vaa) =0$ with probability at most\,\footnote{Note that this is non-trivial only if $d < |S| \leq |\F|$, which in particular implies that $f$ is not the zero function.} $d/|S|$. Thus, given the circuit, we can perform these evaluations efficiently, giving an efficient randomized procedure for deciding whether ``$f\equiv 0$?''.  It is an important open problem to find a derandomization of this algorithm, that is, to find a {\em deterministic} procedure for PIT that runs in polynomial time (in the size of circuit).

Note that in the randomized algorithm of Schwartz-Zippel-DeMillo-Lipton we only use the circuit to compute the evaluation $f(\vaa)$.  Such algorithms are said to run in the {\em black-box} model. In contrast, an algorithm that can access the internal structure of the circuit runs in the {\em white-box} model. It is a folklore result that efficient deterministic black-box algorithms are equivalent to constructions of small \emph{hitting sets}. That is, a hitting set is set of inputs so that any nonzero circuit from the relevant class evaluates to nonzero on at least one of the inputs in the set. For more on PIT we refer to the survey of Shpilka and Yehudayoff~\cite{SY10}. 

A related notion to that of a hitting set is that of a \emph{generator}, which is essentially a low-dimensional curve whose image contains a hitting set.  The equivalence between hitting sets and generators can be found in the above mentioned survey.

\begin{definition}\label{defn:generator}
	Let $\cC\subseteq\F[x_1,\ldots,x_n]$ be a set of polynomials.  A polynomial $\cvG:\F^\ell\to\F^n$ is a \demph{generator for $\cC$} with \demph{seed length $\ell$} if for all $f\in\cC$,
	$f\equiv 0 \text{ iff } f\circ \cvG\equiv 0$. That is, $f(\vx)=0$ in $\F[\vx]$ iff $f(\cvG(\vy))=0$ in $\F[\vy]$.
\end{definition}

In words, a generator for a circuit class $\cC$ is a mapping $\cvG:\F^\ell \to \F^n$, such that for any nonzero polynomial $f$, computed by a circuit from $\cC$, it holds that the composition $f(\cvG)$ is nonzero as well. By considering the image of $\cvG$ on $S^\ell$, where $S \subseteq \F$ is of polynomial size, we obtain a hitting set for $\cC$.

We now list some existing work on derandomizing PIT for some of the classes of polynomials we study in this paper.

\paragraph{Sparse Polynomials:} There are many papers giving efficient black-box PIT algorithms for $\sumprod$ formulas. For example, Klivans and Spielman \cite{KlivansSpielman01} gave a hitting set of polynomial size.

\paragraph{Depth-$3$ Powering Formulas:} Saxena~\cite{Saxena08} gave a polynomial time white-box PIT algorithm and  Forbes, Shpilka, and Saptharishi~\cite{ForbesSS14} gave a $s^{O(\lg\lg s)}$-size hitting set for size-$s$ depth-$3$ powering formulas.

\paragraph{$\sumpowc$ Formulas:} Forbes~\cite{Forbes15} gave an $s^{O(\lg s)}$-size hitting set for size-$s$ $\sumpowc$ formulas (in large characteristic).

\paragraph{Read-once Oblivious ABPs:} Raz and Shpilka~\cite{RazShpilka05} gave a polynomial time white-box PIT algorithm. A long sequence of papers calumniated in the work of Agrawal, Gurjar, Korwar, and Saxena~\cite{AgrawalGKS15}, who gave a $s^{O(\lg s)}$-sized hitting set for size-$s$ roABPs.

\paragraph{Read-$k$ Oblivious ABPs:} Recently, Anderson, Forbes, Saptharishi, Shpilka and Volk~\cite{AndersonFSSV16} obtained a white-box PIT algorithm running in time $2^{\tilde{O}(n^{1-1/2^{k-1}})}$ for $n$-variate $\poly(n)$-sized read-$k$ oblivious ABPs.

\subsection{Coefficient Dimension and roABPs}\label{sec:coeff-dim}

This paper proves various lower bounds on roABPs using a complexity measures known as \emph{coefficient dimension}.  In this section, we define this measures and recall basic properties. Full proofs of these claims can be found for example in the thesis of Forbes~\cite{Forbes14}.

We first define the \emph{coefficient matrix} of a polynomial, called the ``partial derivative matrix'' in the prior work of Nisan~\cite{Nisan91} and Raz~\cite{Raz09}. This matrix is formed from a polynomial $f\in\F[\vx,\vy]$ by arranging its coefficients into a matrix.  That is, the coefficient matrix has rows indexed by monomials $\vx^\va$ in $\vx$, columns indexed by monomials $\vy^\vb$ in $\vy$, and the $(\vx^\va,\vy^\vb)$-entry is the coefficient of $\vx^\va\vy^\vb$ in the polynomial $f$.  We now define this matrix, recalling that $\coeff{\vx^\va\vy^\vb}(f)$ is the coefficient of $\vx^\va\vy^\vb$ in $f$.

\begin{definition}
	Consider $f\in\F[\vx,\vy]$.  Define the \demph{coefficient matrix of $f$} as the scalar matrix 
	\[
		(C_f)_{\va,\vb}
		\eqdef
		\coeff{\vx^\va\vy^\vb}(f)
		\;,
	\]
	where coefficients are taken in $\F[\vx,\vy]$, for $\ellone{\va},\ellone{\vb}\le \deg f$.  
\end{definition}

We now give the related definition of \emph{coefficient dimension}, which looks at the dimension of the row- and column-spaces of the coefficient matrix. Recall that $\coeff{\vx|\vy^\vb}(f)$ extracts the coefficient of $\vy^\vb$ in $f$ as a polynomial in $\F[\vx][\vy]$.

\begin{definition}\label{defn:coefficient-space}
	Let $\coeffs{\vx|\vy}:\F[\vx,\vy]\to\subsets{\F[\vx]}$ be the \demph{space of $\F[\vx][\vy]$ coefficients}, defined by
	\[
		\coeffs{\vx|\vy}(f)
		\eqdef
		\left\{
			\coeff{\vx|\vy^\vb}(f)
		\right\}_{\vb\in\N^n}
		\;,
	\]
	where coefficients of $f$ are taken in $\F[\vx][\vy]$.

	Similarly, define $\coeffs{\vy|\vx}:\F[\vx,\vy]\to\subsets{\F[\vy]}$ by taking coefficients in $\F[\vy][\vx]$. 
\end{definition}

The following basic lemma shows that the rank of the coefficient matrix equals the coefficient dimension, which follows from simple linear algebra.

\begin{lemmawp}[Nisan~\cite{Nisan91}]\label{res:y-dim_eq-x-dim}
	Consider $f\in\F[\vx,\vy]$.  Then the rank of the coefficient matrix $C_f$ obeys
	\[
		\rank C_f
		=
		\dim\coeffs{\vx|\vy}(f)=\dim\coeffs{\vy|\vx}(f)
		\;.
		\qedhere
	\]
\end{lemmawp}

Thus, the ordering of the partition ($(\vx,\vy)$ versus $(\vy,\vx)$) does not matter in terms of the resulting dimension. The above matrix-rank formulation of coefficient dimension can be rephrased in terms of low-rank decompositions.

\begin{lemmawp}\label{res:nisan_dim-eq-width}
	Let $f\in\F[\vx,\vy]$. Then $\dim\coeffs{\vx|\vy}(f)$ equals the minimum $r$ such that there are $\vg\in\F[\vx]^r$ and $\vh\in\F[\vy]^r$ such that $f$ can be written as $f(\vx,\vy)=\sum_{i=1}^r g_i(\vx)h_i(\vy)$.
\end{lemmawp}

We now state a convenient normal form for roABPs (see for example Forbes~\cite[Corollary 4.4.2]{Forbes14}).

\begin{lemmawp}\label{res:roABP-normal-form}
	A polynomial $f\in\F[x_1,\ldots,x_n]$ is computed by width-$r$ roABP iff there exist matrices $A_i(x_{i})\in\F[x_{i}]^{r\times r}$ of (individual) degree $\le \deg f$ such that $f=(\prod_{i=1}^n A_i(x_{i}))_{1,1}$. Further, this equivalence preserves explicitness of the roABPs up to $\poly(n,r,\deg f)$-factors.
\end{lemmawp}

By splitting an roABP into such variable-disjoint inner-products one can obtain a lower bound for roABP width via coefficient dimension.  In fact, this complexity measure \emph{characterizes} roABP width.

\begin{lemmawp}\label{res:roABP-width_eq_dim-coeffs}
	Let $f\in\F[x_1,\ldots,x_n]$ be a polynomial. If $f$ is computed by a width-$r$ roABP then $r \ge \max_i\dim\coeffs{\vx_{\le i}|\vx_{>i}}(f)$. Further, $f$ is computable width-$\left(\max_i\dim\coeffs{\vx_{\le i}|\vx_{>i}}(f)\right)$ roABP.
\end{lemmawp}

Using this complexity measure it is rather straightforward to prove the following closure properties of roABPs.

\begin{fact}\label{fact:roABP:closure}
	If $f,g\in\F[\vx]$ are computable by width-$r$ and width-$s$ roABPs respectively, then
	\begin{itemize}
		\item $f+g$ is computable by a width-$(r+s)$ roABP\@.
		\item $f\cdot g$ is computable by a width-$(rs)$ roABP\@.
	\end{itemize}

	\noindent Further, roABPs are also closed under the follow operations.
	\begin{itemize}
		\item If $f(\vx,\vy)\in\F[\vx,\vy]$ is computable by a width-$r$ roABP in some variable order then the partial substitution $f(\vx,\vaa)$, for $\vaa\in\F^{|\vy|}$, is computable by a width-$r$ roABP in the induced order on $\vx$, where the degree of this roABP is bounded by the degree of the roABP for $f$.
		\item If $f(z_1,\ldots,z_n)$ is computable by a width-$r$ roABP in variable order $z_1<\cdots<z_n$, then $f(x_1y_1,\ldots,x_ny_n)$ is computable by a $\poly(r,\ideg f)$-width roABP in variable order $x_1<y_1<\cdots<x_n<y_n$.
	\end{itemize}
	Further, these operations preserve the explicitness of the roABPs up to polynomial factors in all relevant parameters.
\end{fact}

We now state the extension of these techniques which yield lower bounds for read-$k$ oblivious ABPs, as recently obtained by Anderson, Forbes, Saptharishi, Shpilka and Volk~\cite{AndersonFSSV16}.

\begin{theoremwp}[{\cite{AndersonFSSV16}}]\label{thm:read-k-eval-dim}
	Let $f \in \F[x_1, \ldots, x_n]$ be a polynomial computed by a width-$w$ read-$k$ oblivious ABP\@. Then there exists a partition $\vx=(\vu,\vv,\vw)$ such that
	\begin{enumerate}
		\item $|\vu|, |\vv| \ge n/k^{O(k)}$.
		\item $|\vw| \le n/10$.
		\item $\dim_{\F(\vw)} \coeffs{\vu|\vv}(f_\vw) \le w^{2k}$, where $f_\vw$ is $f$ as a polynomial in $\F(\vw)[\vu,\vv]$.
			\qedhere
	\end{enumerate}
\end{theoremwp}

\subsection{Evaluation Dimension}\label{sec:eval-dim}

While coefficient dimension measures the size of a polynomial $f(\vx,\vy)$ by taking all coefficients in $\vy$, \emph{evaluation dimension} is a complexity measure due to Saptharishi~\cite{Saptharishi12} that measures the size by taking all possible evaluations in $\vy$ over the field.  This measure will be important for our applications as one can restrict such evaluations to the boolean cube and obtain circuit lower bounds for computing $f(\vx,\vy)$ as a \emph{polynomial} via its induced \emph{function} on the boolean cube. We begin with the definition.

\begin{definition}[Saptharishi~\cite{Saptharishi12}]\label{defn:evaluation-space}
	Let $S\subseteq \F$. Let $\evals{\vx|\vy,S}:\F[\vx,\vy]\to\subsets{\F[\vx]}$ be the \demph{space of $\F[\vx][\vy]$ evaluations over $S$}, defined by
	\[
		\evals{\vx|\vy,S}(f(\vx,\vy))
		\eqdef
		\left\{
			f(\vx,\vbb)
		\right\}_{\vbb\in S^{|\vy|}}
		\;.
	\]
	Define $\evals{\vx|\vy}:\F[\vx,\vy]\to\subsets{\F[\vx]}$ to be $\evals{\vx|\vy,S}$ when $S=\F$.

	Similarly, define $\evals{\vy|\vx,S}:\F[\vx,\vy]\to\subsets{\F[\vy]}$ by replacing $\vx$ with all possible evaluations $\vaa\in S^{|\vx|}$, and likewise define $\evals{\vy|\vx}:\F[\vx,\vy]\to\subsets{\F[\vy]}$.
\end{definition}

The equivalence between evaluation dimension and coefficient dimension was shown by Forbes-Shpilka~\cite{ForbesShpilka13b} by appealing to interpolation.

\begin{lemmawp}[Forbes-Shpilka~\cite{ForbesShpilka13b}]\label{res:evals_eq-coeffs}
	Let $f\in\F[\vx,\vy]$.  For any $S\subseteq\F$ we have that $\evals{\vx|\vy,S}(f)\subseteq\spn \coeffs{\vx|\vy}(f)$ so that $\dim \evals{\vx|\vy,S}(f)\le \dim \coeffs{\vx|\vy}(f)$. In particular, if $|S|>\ideg f$ then $\dim\evals{\vx|\vy,S}(f)=\dim\coeffs{\vx|\vy}(f)$.
\end{lemmawp}

While evaluation dimension and coefficient dimension are equivalent when the field is large enough, when restricting our attention to inputs from the boolean cube this equivalence no longer holds (in particular, we have to consider all polynomials that obtain the same values on the boolean cube and not just one polynomial), but evaluation dimension will be still be helpful as it will always lower bound coefficient dimension.

\subsection{Multilinear Polynomials and Multilinear Formulas}

We now turn to multilinear polynomials and classes that respect multilinearity such as multilinear formulas. We first state some well-known facts about multilinear polynomials.

\begin{fact}\label{fact:multilinearization}
	For any two multilinear polynomials $f,g\in\F[x_1,\ldots,x_n]$, $f=g$ as polynomials iff they agree on the boolean cube $\bits^n$.  That is, $f=g$ iff  $f|_{\bits^n}=g|_{\bits^n}$.

	Further, there is a \demph{multilinearization} map $\ml:\F[\vx]\to\F[\vx]$ such that for any $f,g\in\F[\vx]$,
	\begin{enumerate}
		\item $\ml(f)$ is multilinear.
		\item $f$ and $\ml(f)$ agree on the boolean cube, that is, $f|_{\bits^n}=\ml(f)|_{\bits^n}$.
		\item $\deg \ml(f)\le \deg f$.
		\item $\ml(fg)=\ml(\ml(f)\ml(g))$.
		\item $\ml$ is linear, so that for any $\alpha,\beta\in \F$, $\ml(\alpha f+\beta g)=\alpha \ml(f)+\beta\ml(g)$.
		\item $\ml(x_1^{a_1}\cdots x_n^{a_n})=\prod_i x_i^{\max\{a_i,1\}}$.
		\item If $f$ is the sum of at most $s$ monomials ($s$-sparse) then so is $\ml(f)$.
	\end{enumerate}
	Also, if $\hat{f}$ is a function $\bits^n\to\F$ that only depends on the coordinates in $S\subseteq[n]$, then the unique multilinear polynomial $f$ agreeing with $\hat{f}$ on $\bits^n$ is a polynomial only in $\{x_i\}_{i\in S}$.

	One can also extend the multilinearization map $\ml:\F[\vx]\to\F[\vx]$ to matrices $\ml:\F[\vx]^{r\times r}\to\F[\vx]^{r\times r}$ by applying the map entry-wise, and the above properties still hold.
\end{fact}

Throughout the rest of this paper `$\ml$' will denote the multilinearization operator. Raz~\cite{Raz09,Raz06} gave lower bounds for multilinear formulas using the above notion of coefficient dimension, and Raz-Yehudayoff~\cite{RazYehudayoff08,RazYehudayoff09} gave simplifications and extensions to constant-depth multilinear formulas.

\begin{theoremwp}[Raz-Yehudayoff~\cite{Raz09,RazYehudayoff09}]\label{thm:full-rank-lb}
	Let $f\in\F[x_1,\ldots,x_{2n},\vz]$ be a multilinear polynomial in the set of variables $\vx$ and auxiliary variables $\vz$.  Let $f_\vz$ denote the polynomial $f$ in the ring $\F[\vz][\vx]$.  Suppose that for any partition $\vx=(\vu,\vv)$ with $|\vu|=|\vv|=n$ that
	\[
		\dim_{\F(\vz)} \coeffs{\vu|\vv} f_\vz \ge 2^n
		\;.
	\]
	Then $f$ requires $\ge n^{\Omega(\log n)}$-size to be computed as a multilinear formula, and for $d=o(\nicefrac{\log n}{\log\log n})$, $f$ requires $n^{\Omega((\nicefrac{n}{\log n})^{\nicefrac{1}{d}}/d^2)}$-size to be computed as a multilinear formula of product-depth-$d$.
\end{theoremwp}

\subsection{Depth-3 Powering Formulas}

In this section we review facts about depth-3 powering formulas.  We begin with the \emph{duality trick} of Saxena~\cite{Saxena08}, which shows that one can convert a power of a linear form to a sum of products of univariate polynomials.

\begin{theoremwp}[Saxena's Duality Trick~\cite{ShpilkaWigderson01,Saxena08,ForbesGS13}]\label{res:sumpowsum:duality}
	Let $n\ge 1$, and $d\ge 0$. If $|\F|\ge nd+1$, then there are $\poly(n,d)$-explicit univariates $f_{i,j}\in\F[x_i]$ such that
	\[
		(x_1+\cdots+x_n)^d=\sum_{i=1}^s f_{i,1}(x_1)\cdots f_{i,n}(x_n)
		\;,
	\]
	where $\deg f_{i,j}\le d$ and $s=(nd+1)(d+1)$.
\end{theoremwp}

The original proof of Saxena~\cite{Saxena08} only worked over fields of large enough characteristic, and gave $s=nd+1$. A similar version of this trick also appeared in Shpilka-Wigderson~\cite{ShpilkaWigderson01}. The parameters we use here are from the proof of Forbes, Gupta, and Shpilka~\cite{ForbesGS13}, which has the advantage of working over any large enough field.

Noting that the product $f_{i,1}(x_1)\cdots f_{i,n}(x_n)$ trivially has a width-1 roABP (in any variable order), it follows that $(x_1+\cdots+x_n)^d$ has a $\poly(n,d)$-width roABP over a large enough field.  Thus, size-$s$ $\sumpowsum$ formulas have $\poly(s)$-size roABPs over large enough fields by appealing to closure properties of roABPs (\autoref{fact:roABP:closure}).  As it turns out, this result also holds over any field as Forbes-Shpilka~\cite{ForbesShpilka13b} adapted Saxena's~\cite{Saxena08} duality to work over any field.  Their version works over any field, but loses the above clean form (sum of product of univariates).  

\begin{theoremwp}[Forbes-Shpilka~\cite{ForbesShpilka13b}]\label{res:sumpowsum:roABP}
	Let $f\in\F[\vx]$ be expressed as $f(\vx)=\sum_{i=1}^s (\alpha_{i,0}+\alpha_{i,1}x_i+\cdots+\alpha_{i,n}x_n)^{d_i}$.  Then $f$ is computable by a $\poly(r,n)$-explicit width-$r$ roABP of degree $\max_i\{d_i\}$, in any variable order, where $r=\sum_i (d_i+1)$.
\end{theoremwp}

One way to see this claim is to observe that for any variable partition, a linear function can be expressed as the sum of two variable-disjoint linear functions $\ell(\vx_1,\vx_2)=\ell_1(\vx_1)+\ell_2(\vx_2)$.  By the binomial theorem, the $d$-th power of this expression is a summation of $d+1$ variable-disjoint products, which implies a coefficient dimension upper bound of $d+1$ (\autoref{res:nisan_dim-eq-width}) and thus also an roABP-width upper bound (\autoref{res:roABP-width_eq_dim-coeffs}).  One can then sum over the linear forms.

While this simulation suffices for obtaining roABP upper bounds, we will also want the clean form obtained via duality for application to multilinear-formula IPS proofs of the subset-sum axiom (\autoref{res:ips-ubs:subset:mult-form}).

\subsection{Monomial Orders}\label{sec:mon-ord}

We recall here the definition and properties of a \emph{monomial order}, following Cox, Little and O'Shea~\cite{CoxLittleOShea07}.  We first fix the definition of a \emph{monomial} in our context.

\begin{definition}
	A \demph{monomial} in $\F[x_1,\ldots,x_n]$ is a polynomial of the form $\vx^\va=x_1^{a_1}\cdots x_n^{a_n}$ for $\va\in\N^n$.
\end{definition}

We will sometimes abuse notation and associate a monomial $\vx^\va$ with its exponent vector $\va$, so that we can extend this order to the exponent vectors. Note that in this definition ``$1$'' is a monomial, and that scalar multiples of monomials such as $2x$ are not considered monomials. We now define a monomial order, which will be total order on monomials with certain natural properties.

\begin{definition}
	A \demph{monomial ordering} is a total order $\prec$ on the monomials in $\F[\vx]$ such that
	\begin{itemize}
		\item For all $\va\in\N^n\setminus\{\vnz\}$, $1\prec \vx^{\va}$.
		\item For all $\va,\vb,\vc\in\N^n$, $\vx^\va\prec\vx^\vb$ implies $\vx^{\va+\vc}\prec\vx^{\vb+\vc}$.
	\end{itemize}

	For nonzero $f\in\F[\vx]$, the \demph{leading monomial of $f$ (with respect to a monomial order $\prec$)}, denoted $\LM(f)$, is the largest monomial in $\supp(f)\eqdef\{\vx^\va:\coeff{\vx^\va}(f)\ne 0\}$ with respect to the monomial order $\prec$. The \demph{trailing monomial of $f$}, denoted $\TM(f)$, is defined analogously to be the smallest monomial in $\supp(f)$. The zero polynomial has neither leading nor trailing monomial.

	For nonzero $f\in\F[\vx]$, the \demph{leading (resp.\ trailing) coefficient of $f$}, denoted $\LC(f)$ (resp.\ $\TC(f)$), is $\coeff{\vx^\va}(f)$ where $\vx^\va=\LM(f)$ (resp.\ $\vx^\va=\TM(f)$).
\end{definition}

Henceforth in this paper we will assume $\F[\vx]$ is equipped with some monomial order $\prec$.  The results in this paper will hold for \emph{any} monomial order.  However, for concreteness, one can consider the lexicographic ordering on monomials, which is easily seen to be a monomial ordering (see also Cox, Little and O'Shea~\cite{CoxLittleOShea07}).  

We begin with a simple lemma about how taking leading or trailing monomials (or coefficients) is homomorphic with respect to multiplication.

\begin{lemma}\label{res:hom_LM-TM_mult}
	Let $f,g\in\F[\vx]$ be nonzero polynomials. Then the leading monomial and trailing monomials and coefficients are homomorphic with respect to multiplication, that is, $\LM(fg)=\LM(f)\LM(g)$ and $\TM(fg)=\TM(f)\TM(g)$, as well as $\LC(fg)=\LC(f)\LC(g)$ and $\TC(fg)=\TC(f)\TC(g)$.
\end{lemma}
\begin{proof}
	We do the proof for leading monomials and coefficients, the claim for trailing monomials and coefficients is symmetric.

	Let $f(\vx)=\sum_\va \alpha_\va \vx^\va$ and $g(\vx)=\sum_\vb \beta_\vb \vx^\vb$.  Isolating the leading monomials,
	\begin{align*}
		f(\vx)=&\LC(f)\cdot \LM(f)+\sum_{\vx^\va\prec\LM(f)} \alpha_\va \vx^\va,&
		g(\vx)=&\LC(g)\cdot \LM(g)+\sum_{\vx^\vb\prec\LM(g)} \beta_\vb \vx^\vb,
	\end{align*}
	with $\LC(f)=\alpha_{\LM(f)}$ and $\LC(g)=\beta_{\LM(g)}$ being nonzero.  Thus,
	\begin{multline*}
		f(\vx)g(\vx)=\LC(f)\LC(g) \cdot \LM(f)\LM(g)
		+\LC(f)\LM(f)\left(\sum_{\vx^\vb\prec\LM(g)} \beta_\vb \vx^\vb\right)\\
		+\LC(g)\LM(g)\left(\sum_{\vx^\va\prec\LM(f)} \alpha_\va \vx^\va\right)
		+\left(\sum_{\vx^\va\prec\LM(f)} \alpha_\va \vx^\va\right)\left(\sum_{\vx^\vb\prec\LM(g)} \beta_\vb \vx^\vb\right).
	\end{multline*}
	Using that $\vx^\va\vx^\vb\prec \LM(f) \LM(g)$ whenever $\vx^\va\prec\LM(f)$ or $\vx^\vb\prec\LM(g)$ due to the definition of a monomial order, we have that $\LM(f)\LM(g)$ is indeed the maximal monomial in the above expression with nonzero coefficient, and as its coefficient is $\LC(f)\LC(g)$.
\end{proof}

We now recall the well-known fact that for any set of polynomials the dimension of their span in $\F[\vx]$ is equal to the number of distinct leading or trailing monomials in their span. 

\begin{lemmawp}\label{res:dim-eq-num-TM-spn}
	Let $S\subseteq\F[\vx]$ be a set of polynomials.  Then $\dim \spn S=\nLM{\spn S}=\nTM{\spn S}$.  In particular, $\dim \spn S\ge \nLM{S},\nTM{S}$.
\end{lemmawp}

\section{Upper Bounds for Linear-IPS}\label{sec:h-ips:ubs}

While the primary focus of this work is on \emph{lower bounds} for restricted classes of the IPS proof system, we begin by discussing \emph{upper bounds} to demonstrate that these restricted classes can prove the unsatisfiability of non-trivial systems of polynomials equations. In particular we go beyond existing work on upper bounds (\cite{GrigorievHirsch03,RazTzameret08a,RazTzameret08b,GrochowPitassi14,LiTW15}) and place interesting refutations in IPS subsystems where we will also prove lower bounds, as such upper bounds demonstrate the non-triviality of our lower bounds.

We begin by discussing the power of the linear-IPS proof system. While one of the most novel features of IPS proofs is their consideration of non-linear certificates, we show that in powerful enough models of algebraic computation, linear-IPS proofs can efficiently simulate general IPS proofs, essentially answering an open question of Grochow and Pitassi~\cite{GrochowPitassi14}. A special case of this result was obtained by Grochow and Pitassi~\cite{GrochowPitassi14}, where they showed that \lIPS can simulate $\sumprod$-IPS\@.  We then consider the \emph{subset-sum} axioms, previously considered by Impagliazzo, \Pudlak, and Sgall~\cite{IPS99}, and show that they can be refuted in polynomial size by the $\cC$-\lIPS proof system where $\cC$ is either the class of roABPs, or the class of multilinear formulas.

\subsection{Simulating IPS Proofs with Linear-IPS}\label{sec:h-ips:ips=h-ips}

We show here that general IPS proofs can be efficiently simulated by linear-IPS, assuming that the axioms to be refuted are described by small algebraic circuits.  Grochow and Pitassi~\cite{GrochowPitassi14} showed that whenever the IPS proof computes \emph{sparse} polynomials, one can simulate it by linear-IPS using (possibly non-sparse) algebraic circuits.  We give here a simulation of IPS when the proofs use general algebraic circuits.

To give our simulation, we will need to show that if a small circuit $f(\vx,y)$ is divisible by $y$, then the quotient $\nicefrac{f(\vx,y)}{y}$ also has a small circuit. Such a result clearly follows from Strassen's~\cite{Strassen73} elimination of divisions in general, but we give two constructions for the quotient which tailor Strassen's~\cite{Strassen73} technique to optimize certain parameters.  

The first construction assumes that $f$ has degree bounded by $d$, and produces a circuit for the quotient whose size depends polynomially on $d$.  This construction is efficient when $f$ is computed by a formula or branching program (so that $d$ is bounded by the size of $f$).  In particular, this construction will preserve the depth of $f$ in computing the quotient, and as such we only present it for formulas.  The construction proceeds via interpolation to decompose $f(\vx,y)=\sum_i f_i(\vx)y^i$ into its constituent parts $\{f_i(\vx)\}_i$ and then directly constructs $\nicefrac{f(\vx,y)}{y}=\sum_i f_i(\vx)y^{i-1}$.

\begin{lemma}\label{res:divide-by-y:black-box}
	Let $\F$ be a field with $|\F|\ge d+1$.  Let $f(\vx,y)\in\F[x_1,\ldots,x_n,y]$ be a degree $\le d$ polynomial expressible as $f(\vx,y)=\sum_{0\le i\le d} f_i(\vx)y^i$ for $f_i\in\F[\vx]$.  Assume $f$ is computable by a size-$s$ depth-$D$ formula. Then for $a\ge 1$ one can compute 
	\[
		\sum_{i=a}^d f_i(\vx) y^{i-a}
		\;,
	\]
	by a $\poly(s,a,d)$-size depth-$(D+2)$ formula.  Further, given $d$ and the formula for $f$, the resulting formula is $\poly(s,a,d)$-explicit. In particular, if $y^a|f(\vx,y)$ then the quotient $\nicefrac{f(\vx,y)}{y^a}$ has a formula of these parameters.
\end{lemma}
\begin{proof}
	Express $f(\vx,y)\in\F[\vx][y]$ by $f(\vx,y)=\sum_{0\le i\le d} f_i(\vx)y^i$. As $|\F|\ge 1+\deg_y f$, by interpolation there are $\poly(d)$-explicit constants $\alpha_{i,j},\beta_j\in\F$ such that
	\[
		f_i(\vx)=\sum_{j=0}^d \alpha_{i,j}f(\vx,\beta_j)
		\;.
	\]
	It then follows that
	\[
		\sum_{i=a}^d f_i(\vx) y^{i-a}
		=\sum_{i=a}^d \left(\sum_{j=0}^d \alpha_{i,j}f(\vx,\beta_j)\right) y^{i-a}
		=\sum_{i=a}^d \sum_{j=0}^d \alpha_{i,j}f(\vx,\beta_j) y^{i-a}
		\;,
	\]
	which is clearly a formula of the appropriate size, depth, and explicitness. The claim about the quotient $\nicefrac{f(\vx,y)}{y^a}$ follows from seeing that if the quotient is a polynomial then $\nicefrac{f(\vx,y)}{y^a}=\sum_{i=a}^d f_i(\vx) y^{i-a}$.
\end{proof}

The above construction suffices in the typical regime of algebraic complexity where the circuits compute polynomials whose degree is polynomially-related to their circuit size.  However, the simulation of Extended Frege by general IPS proved by Grochow-Pitassi~\cite{GrochowPitassi14} (\autoref{thm:GrochowPitassi14}) yields IPS refutations with circuits of possibly exponential degree (see also \autoref{rmk:EF-degree}).  As such, this motivates the search for an efficient division lemma in this regime.  We now provide such a lemma, which is a variant of Strassen's~\cite{Strassen73} homogenization technique for efficiently computing the low-degree homogeneous components of an unbounded degree circuit.  As weaker models of computation (such as formulas and branching programs) cannot compute polynomials of degree exponential in their size, we only present this lemma for circuits.

\begin{lemma}\label{res:divide-by-y:white-box}
	Let $f(\vx,y)\in\F[x_1,\ldots,x_n,y]$ be a polynomial expressible as $f(\vx,y)=\sum_i f_i(\vx)y^i$ for $f_i\in\F[\vx]$, and assume $f$ is computable by a size-$s$ circuit. Then for $a\ge 1$ there is an $O(a^2s)$-size circuit with outputs gates computing
	\[
		f_0(\vx),\ldots,f_{a-1}(\vx),\sum_{i\ge a} f_i(\vx) y^{i-a}
		\;.
	\]
	Further, given $a$ and the circuit for $f$, the resulting circuit is $\poly(s,a)$-explicit.  In particular, if $y^a|f(\vx,y)$ then the quotient $\nicefrac{f(\vx,y)}{y^a}$ has a circuit of these parameters.
\end{lemma}
\begin{proof}
	The proof proceeds by viewing the computation in the ring $\F[\vx][y]$, and splitting each gate in the circuit for $f$ into its coefficients in terms of $y$.  However, to avoid a dependence on the degree, we only split out the coefficients of $y^0,y^1,\ldots,y^{a-1}$, and then group together the rest of the coefficients together.  That is, for a polynomial $g(\vx,y)=\sum_{i\ge 0} g_i(\vx)y^i$, we can split this into $g=\sum_{0\le i<a} g_i(\vx)y^i+\left(\sum_{i\ge a} g_i(\vx)y^{i-a}\right)y^a$ to obtain the constituent parts $g_0(\vx),\ldots,g_{a-1}(\vx),\sum_{i\ge a}g_i(\vx)y^{i-a}$. We can then locally update this split by appropriately keeping track of how addition and multiplication affects this grouping of coefficients.  We note that we can assume without loss of generality that the circuit for $f$ has fan-in 2, as this only increases the size of the circuit by a constant factor (measuring the size of the circuit in number of edges) and simplifies the construction.

	\uline{construction:} Let $\Phi$ denote the circuit for $f$. For a gate $v$ in $\Phi$, denote $\Phi_v$ to be the configuration of $v$ in $\Phi$ and let $f_v$ to be the polynomial computed by the gate $v$.  We will define the new circuit $\Psi$, which will be defined by the gates $\{(v,i): v\in\Phi, 0\le i\le a\}$ and the wiring between them, as follows.
	\begin{itemize}
		\item $\Phi_v\in\F$: $\Psi_{(v,0)}\eqdef \Phi_v$, $\Psi_{(v,i)}\eqdef 0$ for $i\ge 1$.
		\item $\Phi_v=x_i$: $\Psi_{(v,0)}\eqdef x_i$, $\Psi_{(v,i)}\eqdef 0$ for $i\ge 1$.
		\item $\Phi_v=y$: $\Psi_{(v,1)}\eqdef 1$, $\Psi_{(v,i)}\eqdef 0$ for $i\ne 1$.
		\item $\Phi_v=\Phi_u+\Phi_w$: $\Psi_{(v,i)}\eqdef \Psi_{(u,i)}+\Psi_{(w,i)}$, all $i$.
		\item $\Phi_v=\Phi_u\times \Phi_w$, $0\le i<a$:
			\[
				\Psi_{(v,i)}\eqdef \sum_{0\le j\le i}\Psi_{(u,j)}\times\Psi_{(w,i-j)}
				\;.
			\]
		\item $\Phi_v=\Phi_u\times \Phi_w$, $i=a$:
			\begin{align*}
				\Psi_{(v,a)}
				&\eqdef \sum_{\ell=a}^{2(a-1)}y^{\ell-a}\sum_{\substack{i+j=\ell\\0\le i,j< a}}\Psi_{(u,i)}\times\Psi_{(w,j)}
					+ \sum_{0\le i<a} \Psi_{(u,i)}\times\Psi_{(w,a)}\times y^i\\
				&\hspace{1in}+ \sum_{0\le j<a} \Psi_{(u,a)}\times\Psi_{(w,j)}\times y^j
					+\Psi_{(u,a)}\times\Psi_{(w,a)}\times y^a
				\;.
			\end{align*}
	\end{itemize}

	\uline{complexity:} Split the gates in $\Psi$ into two types, those gates $(v,i)$ where $i=a$ and $v$ is a multiplication gate in $\Phi$, and then the rest.  For the former type, $\Psi_{(v,a)}$ is computable by a size-$O(a^2)$ circuit in its children, and there are at most $s$ such gates.  For the latter type, $\Psi_{(v,i)}$ is computable by a size-$O(a)$ circuit in its children, and there are at most $O(as)$ such gates.  As such, the total size is $O(a^2s)$.

	\uline{correctness:} We now establish correctness as a subclaim.  For a gate $(v,i)$ in $\Psi$, let $g_{(v,i)}$ denote the polynomial that it computes.  

	\begin{subclaim}
		For each gate $v$ in $\Phi$, for $0\le i<a$ we have that $g_{(v,i)}=\coeff{\vx|y^i}(f_v)$ and for $i=a$ we have that $g_{(v,a)}=\sum_{i\ge a} \coeff{\vx|y^i}(f_v)y^{i-a}$.  In particular, $f_v=\sum_{i=0}^a g_{(v,i)}y^i$.
	\end{subclaim}
	\begin{subproof}
		Note that the second part of the claim follows from the first. We now establish the first part by induction on the gates of the circuit.

		\begin{itemize}
			\item $\Phi_v\in\F$: By construction, $g_{(v,0)}=f_v=\coeff{\vx|y^0}(f_v)$, and for $i\ge1$, $g_{(v,i)}=0=\coeff{\vx|y^i}(f_v)$.
			\item $\Phi_v=x_i$: By construction, $g_{(v,0)}=f_v=\coeff{\vx|y^0}(f_v)$, and for $i\ge1$, $g_{(v,i)}=0=\coeff{\vx|y^i}(f_v)$.
			\item $\Phi_v=y$: By construction, $g_{(v,1)}=1=\coeff{\vx|y^1}(f_v)$, and for $i\ne 1$, $g_{(v,i)}=0=\coeff{\vx|y^i}(f_v)$.
			\item $\Phi_v=\Phi_u+\Phi_w$: 
				\begin{align*}
					g_{(v,i)}
					&=g_{(u,i)}+g_{(w,i)}\\
					&=\coeff{\vx|y^i}(f_u)+\coeff{\vx|y^i}(f_w)\\
					&=\coeff{\vx|y^i}(f_u+f_w)
					=\coeff{\vx|y^i}(f_v)
					\;.
				\end{align*}
			\item $\Phi_v=\Phi_u\times\Phi_w$, $0\le i<a$: 
				\begin{align*}
					g_{(v,i)}
					&=\sum_{0\le j\le i} g_{(u,j)}\cdot g_{(w,i-j)}\\
					&=\sum_{0\le j\le i} \coeff{\vx|y^{j}}(f_u) \cdot \coeff{\vx|y^{i-j}}(f_w)\\
					&=\coeff{\vx|y^i}(f_u\cdot f_w)
					=\coeff{\vx|y^i}(f_v)
					\;.
				\end{align*}
			\item $\Phi_v=\Phi_u\times\Phi_w$, $i=a$: 
				\begin{align*}
					g_{(v,a)}
					&=\sum_{\ell=a}^{2(a-1)}y^{\ell-a}\sum_{\substack{i+j=\ell\\0\le i,j< a}}g_{(u,i)}\cdot g_{(w,j)}
						+ \sum_{0\le i<a} g_{(u,i)}\cdot g_{(w,a)}\cdot y^i\\
					&\hspace{1in}+ \sum_{0\le j<a} g_{(u,a)}\cdot g_{(w,j)}\cdot y^j
						+g_{(u,a)}\cdot g_{(w,a)}\cdot y^a\\
					&=\sum_{\ell=a}^{2(a-1)}y^{-a}\sum_{\substack{i+j=\ell\\0\le i,j< a}}\coeff{\vx|y^i}(f_u)y^i\cdot \coeff{\vx|y^j}(f_w)y^j\\
					&\hspace{1in}+ \sum_{0\le i<a} \coeff{\vx|y^i}(f_u)\cdot \left (\sum_{j\ge a}\coeff{\vx|y^j}(f_w)y^{j-a}\right)\cdot y^i\\
					&\hspace{1in}+ \sum_{0\le j<a} \left(\sum_{i\ge a}\coeff{\vx|y^i}(f_u)y^{i-a}\right)\cdot \coeff{\vx|y^j}(f_w)\cdot y^j\\
					&\hspace{1in}+\left(\sum_{i\ge a}\coeff{\vx|y^i}(f_u)y^{i-a}\right)\cdot \left (\sum_{j\ge a}\coeff{\vx|y^j}(f_w)y^{j-a}\right)\cdot y^a\\
					&=\sum_{\substack{i+j\ge a\\0\le i,j<a}}\coeff{\vx|y^i}(f_u)\coeff{\vx|y^j}(f_w) y^{i+j-a}
						+\sum_{\substack{0\le i<a\\ j\ge a}}\coeff{\vx|y^i}(f_u)\coeff{\vx|y^j}(f_w) y^{i+j-a}\\
					&\hspace{.2in}+\sum_{\substack{i\ge a\\ 0\le j<a}}\coeff{\vx|y^i}(f_u)\coeff{\vx|y^j}(f_w) y^{i+j-a}
						+\sum_{i,j\ge a}\coeff{\vx|y^i}(f_u)\coeff{\vx|y^j}(f_w) y^{i+j-a}\\
					&=\sum_{i+j\ge a}\coeff{\vx|y^i}(f_u)\coeff{\vx|y^j}(f_w)\cdot y^{i+j-a}\\
					&=\sum_{\ell\ge a}\coeff{\vx|y^\ell}(f_u \cdot f_w)\cdot y^{\ell-a}\\
					&=\sum_{\ell\ge a}\coeff{\vx|y^\ell}(f_v)\cdot y^{\ell-a}
					\;.
					\qedhere
				\end{align*}
		\end{itemize}
	\end{subproof}

	The correctness then follows by examining $v_\mathrm{out}$, the output gate of $\Phi$, so that $f_{v_\mathrm{out}}=f$.  The gates $(v_\mathrm{out},0),\ldots,(v_\mathrm{out},a)$ are then outputs of $\Psi$ and by the above subclaim have the desired functionality.

	\uline{quotient:} The claim about the quotient $\nicefrac{f(\vx,y)}{y^a}$ follows from seeing that if the quotient is a polynomial then $\nicefrac{f(\vx,y)}{y^a}=\sum_{i\ge a} f_i(\vx) y^{i-a}$ which is one of the outputs of the constructed circuit.
\end{proof}

We now give our simulation of general IPS by linear-IPS\@. In the below set of axioms we do not separate out the boolean axioms from the rest, as this simplifies notation.

\begin{proposition}\label{res:h-ips_v_ips}
	Let $f_1,\ldots,f_m\in\F[x_1,\ldots,x_n]$ be unsatisfiable polynomials with an IPS refutation $C\in\F[\vx,y_1,\ldots,y_m]$.  Then $f_1,\ldots,f_m$ have a linear-IPS refutation $C'\in\F[\vx,\vy]$ under the following conditions.
	\begin{enumerate}
		\item Suppose $f_1,\ldots,f_m,C$ are computed by size-$s$ formulas, have degree at most $d$, and $|\F|\ge d+1$. Then $C'$ is computable by a $\poly(s,d,m)$-size formula of depth-$O(D)$, and $C'$ is $\poly(s,d,m)$-explicit given $d$ and the formulas for $f_1,\ldots,f_m,C$.
			\label{res:h-ips_v_ips:formula}
		\item Suppose $f_1,\ldots,f_m,C$ are computed by size-$s$ circuits. Then $C'$ is computable by a $\poly(s,m)$-size circuit, and $C'$ is $\poly(s,m)$-explicit given the circuits for $f_1,\ldots,f_m,C$.
			\label{res:h-ips_v_ips:circuit}
	\end{enumerate}
\end{proposition}
\begin{proof}
	Express $C(\vx,\vy)$ as a polynomial in $\F[\vx][\vy]$, so that $C(\vx,\vy)=\sum_{\va>\vnz} C_\va(\vx)\vy^\va$, where we use that $C(\vx,\vnz)=0$ to see that we can restrict $\va$ to $\va>\vnz$.  Partitioning the $\va\in\N^n$ based on the index of their first nonzero value, and denoting $\va_{<i}$ for the first $i-1$ coordinates of $\va$, we obtain
	\begin{align*}
		C(\vx,\vy)
		&=\sum_{\va>\vnz} C_\va(\vx)\vy^\va\\
		&=\sum_{i=1}^n\; \sum_{\substack{\va: \va_{<i}=\vnz,\\ a_i>0}} C_\va(\vx)\vy^\va
		\;.
		\intertext{Now define $C_i(\vx,\vy)\eqdef\sum_{\substack{\va: \va_{<i}=\vnz,\\ a_i>0}} C_\va(\vx)\vy^{\va-\ve_i}$, where $\ve_i$ is the $i$-th standard basis vector. Note that this is a valid polynomial as in this summation we assume $a_i>0$ so that $\va-\ve_i\ge \vnz$. Thus,}
		C(\vx,\vy)
		&=\sum_{i=1}^n C_i(\vx,\vy)y_i
		\;.
	\end{align*}

	We now define $C'(\vx,\vy)\eqdef\sum_{i=1}^n C_i(\vx,\vf(\vx))y_i$ and claim it is the desired linear-IPS refutation, where note that we have only \emph{partially} substituted in the $f_i$ for the $y_i$.  First, observe that it is a valid refutation, as $C'(\vx,\vnz)=\sum_{i=1}^n C_i(\vx,\vf(\vx))\cdot 0=0$, and $C'(\vx,\vf(\vx))=\sum_{i=1}^n C_i(\vx,\vf(\vx))f_i(\vx)=C(\vx,\vf(\vx))=1$ via the above equation and using that $C$ is a valid IPS refutation.

	We now argue that $C'$ can be efficiently computed in the two above regimes.
	
	\uline{\eqref{res:h-ips_v_ips:formula}:} Up to constant-loss in the depth and polynomial-loss in the size, for bounding the complexity of $C'$ it suffices to bound the complexity of each $C_i(\vx,\vf(\vx))$.  First, note that
	\begin{align*}
		C_i(\vx,\vy)y_i
		=\sum_{\substack{\va: \va_{<i}=\vnz,\\ a_i>0}} C_\va(\vx)\vy^\va
		=C(\vx,\vnz,y_i,\vy_{> i})-C(\vx,\vnz,0,\vy_{> i})
		\;,
	\end{align*}
	where each ``$\vnz$'' here is a vector of $i-1$ zeros.  Clearly each of $C(\vx,\vnz,y_i,\vy_{> i})$ and $C(\vx,\vnz,0,\vy_{> i})$ have formula size and depth bounded by that of $C$. From our division lemma for formulas (\autoref{res:divide-by-y:black-box}) it follows that $C_i(\vx,\vy)=\frac{1}{y_i}(C(\vx,\vnz,y_i,\vy_{> i})-C(\vx,\vnz,0,\vy_{> i}))$ has a $\poly(s,d)$-size depth-$O(D)$ formula of the desired explicitness (as $\deg_y \left(C(\vx,\vnz,y_i,\vy_{> i})-C(\vx,\vnz,0,\vy_{> i})\right)\le \deg_y C\le d\le |\F|-1$, so that $\F$ is large enough).  We replace $\vy\leftarrow \vf(\vx)$ to obtain $C_i(\vx,\vf(\vx))$ of the desired size and explicitness, using that the $f_j$ themselves have small-depth formulas.

	\uline{\eqref{res:h-ips_v_ips:circuit}:} This follows as in \eqref{res:h-ips_v_ips:formula}, using now the division lemma for circuits (\autoref{res:divide-by-y:white-box}).
\end{proof}

Grochow and Pitassi~\cite[Open Question 1.13]{GrochowPitassi14} asked whether one can relate the complexity of IPS and linear-IPS, as they only established such relations for simulating $\sumprod$-IPS by (general) linear-IPS\@.  Our above result essentially answers this question for general formulas and circuits, at least under the assumption that the unsatisfiable polynomial system $f_1=\cdots=f_m=0$ can be written using small algebraic formulas or circuits.  This is a reasonable assumption as it is the most common regime for proof complexity.  However, the above result does not fully close the question of Grochow-Pitassi~\cite{GrochowPitassi14} with respect to simulating $\cC$-IPS by $\cD$-\lIPS for various restricted subclasses $\cC,\cD$ of algebraic computation.  That is, for such a simulation our result requires $\cD$ to at the very least contain $\cC$ composed with the axioms $f_1,\ldots,f_m$, and the when applying this to the models considered in this paper (sparse polynomials, depth-3 powering formulas, roABPs, multilinear formulas) this seems to non-negligibly increase the complexity of the algebraic reasoning.

\subsection{Multilinearizing roABP-\lIPS}\label{sec:multilinearization:roABP}

We now exhibit instances where one can efficiently \emph{prove} that a polynomial equals its multilinearization modulo the boolean axioms. That is, for a polynomial $f$ computed by a small circuit we wish to prove that $f\equiv \ml(f)\mod \vx^2-\vx$ by expressing $f-\ml(f)=\sum_i h_i \cdot (x_i^2-x_i)$ so that the $h_i$ also have small circuits.  Such a result will simplify the search for linear-IPS refutations by allowing us to focus on the non-boolean axioms.  That is, if we are able to find a refutation of $\vf,\vx^2-\vx$ given by 
\[
	\sum_j g_j f_j \equiv 1 \mod \vx^2-\vx
	\;,
\]
where the $g_j$ have small circuits, multilinearization results of the above form guarantee that there are $h_i$ so that
\[
	\sum_j g_j f_j +\sum_i h_i\cdot (x_i^2-x_i)=1
	\;,
\]
which is a proper linear-IPS refutation.

We establish such a multilinearization result when $f$ is an roABP in this section, and consider when $f$ the product of a low-degree multilinear polynomial and a multilinear formula in the next section.  We will use these multilinearization results in our construction of IPS refutations of the subset-sum axiom (\autoref{sec:subset:ub}).

We begin by noting that multilinearization for these two circuit classes is rather special, as these classes straddle the conflicting requirements of neither being \emph{too weak} nor \emph{too strong}.  That is, some circuit classes are simply too weak to compute their multilinearizations.  An example is the class of depth-3 powering formulas, where $(x_1+\cdots+x_n)^n$ has a small $\sumpowsum$ formula, but its multilinearization has the leading term $n!x_1\cdots x_n$ and thus requires exponential size as a $\sumpowsum$ formula (by appealing to \autoref{res:lbs-mult:LM:sumpowsum}). On the other hand, some circuit classes are too strong to admit efficient multilinearization (under plausible complexity assumptions).  That is, consider an $n \times n$ symbolic matrix $X$ where $(X)_{i,j}=x_{i,j}$ and the polynomial $f(X,\vy)\eqdef (x_{1,1} y_1+\cdots +x_{1,n}y_n)\cdots (x_{n,1} y_1+\cdots +x_{n,n}y_n)$, which is clearly a simple depth-3 ($\prod\sum\prod$) circuit.  Viewing this polynomial in $\F[X][\vy]$, one sees that $\coeff{X|y_1\cdots y_n} f=\perm(X)$, where $\perm(X)$ is the $n\times n$ permanent.  Viewing $\ml(f)$, the multilinearization of $f$, in $\F[X][\vy]$ one sees that $\ml(f)$ is of degree $n$ and its degree $n$ component is the coefficient of $y_1\cdots y_n$ in $\ml(f)$, which is still $\perm(X)$.  Hence, by interpolation, one can extract this degree $n$ part and thus can compute a circuit for $\perm(X)$ given a circuit for $\ml(f)$.  Since we believe $\perm(X)$ does not have small algebraic circuits it follows that the multilinearization of $f$ does not have small circuits. These examples show that efficient multilinearization is a somewhat special phenomenon.

We now give our result for multilinearizing roABPs, where we multilinearize variable by variable via telescoping.  

\begin{proposition}\label{res:multilin:roABP}
	Let $f\in\F[x_1,\ldots,x_n]$ be computable by a width-$r$ roABP in variable order $x_1<\cdots<x_n$, so that $f(\vx)=\left(\prod_{i=1}^n A_i(x_i)\right)_{1,1}$ where $A_i\in\F[x_i]^{r\times r}$ have $\deg A_i\le d$.  Then $\ml(f)$ has a $\poly(r,n,d)$-explicit width-$r$ roABP in variable order $x_1<\cdots<x_n$, and there are $\poly(r,n,d)$-explicit width-$r$ roABPs $h_1,\ldots,h_n\in\F[\vx]$ in variable order $x_1<\cdots<x_n$ such that
	\[
		f(\vx)=\ml(f)+\sum_{j=1}^n h_j\cdot (x_j^2-x_j)
		\;.
	\]
	Further, $\ideg h_j\le \ideg f$ and the individual degree of the roABP for $\ml(f)$ is $\le 1$.
\end{proposition}
\begin{proof}
	We apply the multilinearization map $\ml:\F[\vx]\to\F[\vx]$ to matrices $\ml:\F[\vx]^{r\times r}\to\F[\vx]^{r\times r}$ by applying the map entry-wise (\autoref{fact:multilinearization}).  It follows then that $A_i(x_i)-\ml(A_i(x_i))\equiv 0\mod x_i^2-x_i$, so that $A_i(x_i)-\ml(A_i(x_i))=B(x_i)\cdot (x_i^2-x_i)$ for some $B_i(x_i)\in\F[x_i]^{r\times r}$ where $\ideg B_i(x_i)\le \ideg A_i(x_i)$. Now define $\ml_{\le i}$ be the map which multilinearizes the first $i$ variables and leaves the others intact, so that $\ml_{\le 0}$ is the identity map and $\ml_{\le n}=\ml$.  Telescoping,
	\begin{align*}
		\prod_{i=1}^n A_i(x_i)
		&=\ml_{<1}\left(\prod_{i=1}^n A_i(x_i)\right)\\
		&=\ml_{\le n}\left(\prod_{i=1}^n A_i(x_i)\right)+\sum_{j=1}^{n} \left[\ml_{<j}\left(\prod_{i=1}^n A_i(x_i)\right)-\ml_{\le j}\left(\prod_{i=1}^n A_i(x_i)\right)\right]\\
		\intertext{using that the identity $\ml(gh)=\ml(\ml(g)\ml(h))$ (\autoref{fact:multilinearization}) naturally extends from scalars to matrices, and also to partial-multilinearization (by viewing $\ml_{\le i}$ as multilinearization in $\F[\vx_{>i}][\vx_{\le i}]$),}
		&=\ml_{\le n}\left(\prod_{i=1}^n \ml_{\le n} A_i(x_i)\right)+\sum_{j=1}^{n} \left[\ml_{<j}\left(\prod_{i< j} \ml_{<j}(A_i(x_i))\prod_{i\ge j} A_i(x_i)\right)\right.\\
		&\hspace{2in}\left.-\ml_{\le j}\left(\prod_{i\le j} \ml_{\le j}(A_i(x_i))\prod_{i>j} A_i(x_i)\right)\right]\\
		\intertext{dropping the outside $\ml_{<j}$ and $\ml_{\le j}$ as the inside polynomials are now multilinear in the appropriate variables, and replacing them with $\ml$ as appropriate,}
		&=\prod_{i=1}^n \ml(A_i(x_i))+\sum_{j=1}^{n} \left[\prod_{i<j} \ml(A_i(x_i))\prod_{i\ge j} A_i(x_i)\right.\\
		&\hspace{2in}\left.-\prod_{i\le j} \ml(A_i(x_i))\prod_{i>j} A_i(x_i)\right]\\
		&=\prod_{i=1}^n \ml(A_i(x_i))+\sum_{j=1}^{n} \left[\prod_{i<j} \ml(A_i(x_i))\cdot \Big(A_j(x_j)-\ml(A_j(x_j))\Big)\cdot\prod_{i>j} A_i(x_i)\right]\\
		&=\prod_{i=1}^n \ml(A_i(x_i))+\sum_{j=1}^{n} \left[\prod_{i<j} \ml(A_i(x_i))\cdot B_j(x_j)\cdot \prod_{i>j} A_i(x_i)\cdot (x_j^2-x_j)\right]
		\;.
	\end{align*}
	Taking the $(1,1)$-entry in the above yields that
	\begin{align*}
		f(\vx)
		&=\left(\prod_{i=1}^n A_i(x_i)\right)_{1,1}\\
		&=\left(\prod_{i=1}^n \ml(A_i(x_i))\right)_{1,1}+\sum_{j=1}^{n} \left(\prod_{i<j} \ml(A_i(x_i))\cdot B_j(x_j)\cdot \prod_{i>j} A_i(x_i)\right)_{1,1}\cdot (x_j^2-x_j)
		\;.
	\end{align*}
	Thus, define 
	\[
		\hat{f}\eqdef \left(\prod_{i=1}^n \ml(A_i(x_i))\right)_{1,1}
		\;,
		\qquad
		h_j \eqdef \left(\prod_{i<j} \ml(A_i(x_i))\cdot B_j(x_j)\cdot\prod_{i>j} A_i(x_i)\right)_{1,1}
		\;.
	\]
	It follows by construction that $\hat{f}$ and each $h_j$ are computable by width-$r$ roABPs of the desired explicitness in the correct variable order. Further, $\ideg h_j\le \ideg f$ and $\hat{f}$ has individual degree $\le 1$.  Thus, the above yields that $f=\hat{f}+\sum_j h_j\cdot (x_j^2-x_j)$, from which it follows that $\ml(f)=\hat{f}$, as desired.
\end{proof}

We now conclude that in designing an roABP-\lIPS refutation $\sum_j g_j\cdot f_j+\sum_i h_i\cdot (x_i^2-x_i)$ of $f_1(\vx),\ldots,f_m(\vx),\vx^2-\vx$, it suffices to bound the complexity of the $g_j$'s.

\begin{proposition}\label{res:multilin:roABP-lIPS}
	Let $f_1,\ldots,f_m\in\F[x_1,\ldots,x_n]$ be unsatisfiable polynomials over $\vx\in\bits^n$ computable by width-$s$ roABPs in variable order $x_1<\cdots<x_n$.  Suppose that there are $g_j\in\F[\vx]$ such that
	\[
		\sum_{i=1}^m g_j(\vx)f_j(\vx)\equiv 1\mod\vx^2-\vx
		\;,
	\]
	where the $g_j$ have width-$r$ roABPs in the variable order $x_1<\cdots<x_n$.  Then there is an roABP-\lIPS refutation $C(\vx,\vy,\vz)$ of individual degree at most $1+\ideg \vf$ and computable in width-$\poly(s,r,n,m)$ in any variable order where $x_1<\cdots<x_n$.  Furthermore, this refutation is $\poly(s,r,\ideg \vg,\ideg \vf,n,m)$-explicit given the roABPs for $f_j$ and $g_j$.
\end{proposition}
\begin{proof}
	We begin by noting that we can multilinearize the $g_j$, so that $\sum_{i=1}^m \ml(g_j(\vx))f_j(\vx)\equiv 1\mod\vx^2-\vx$, and that $\ml(g_j)$ are $\poly(r,\ideg \vg,n,m)$-explicit multilinear roABPs of width-$r$ by \autoref{res:multilin:roABP}.  Thus, we assume going forward that the $g_j$ are multilinear.

	As $g_j,f_j$ are computable by roABPs, their product $g_jf_j$ is computable by a width-$rs$ roABP in the variable order $x_1<\cdots<x_n$ (\autoref{fact:roABP:closure}) with individual degree at most $1+\ideg f_j$ (\autoref{res:roABP-normal-form}). Thus, by the above multilinearization (\autoref{res:multilin:roABP}), there are $h_{j,i}\in\F[\vx]$ such that
	\[
		g_j(\vx)f_j(\vx)=\ml(g_jf_j)+\sum_{i=1}^n h_{j,i}(\vx)\cdot (x_i^2-x_i)
		\;.
	\]
	where the $h_{j,i}$ are computable by width-$rs$ roABPs of individual degree at most $1+\ideg f_j$. We now define
	\[
		C(\vx,\vy,\vz)\eqdef \sum_{j=1}^m g_j(\vx)y_j - \sum_{i=1}^n\left(\sum_{j=1}^m h_{j,i}(\vx)\right) z_i
		\;.
	\]
	By the closure operations of roABPs (\autoref{fact:roABP:closure}) it follows that $C$ is computable the appropriately-sized roABPs in the desired variable orders, has the desired individual degree, and that $C$ has the desired explicitness.

	We now show that this is a valid IPS refutation. Observe that $C(\vx,\vnz,\vnz)=0$ and that
	\begin{align}
		C(\vx,\vf,\vx^2-\vx)
		&=\sum_{j=1}^m g_j(\vx)f_j(\vx) - \sum_{i=1}^n\left(\sum_{j=1}^m h_{j,i}(\vx)\right) (x_i^2-x_i)\\
		&=\sum_{j=1}^m \left( g_j(\vx)f_j(\vx) - \sum_{i=1}^n h_{j,i}(\vx) (x_i^2-x_i)\right)\\
		&=\sum_{j=1}^m \ml(g_jf_j)\\
		\intertext{as $\sum_{i=1}^m g_j(\vx)f_j(\vx)\equiv 1\mod \vx^2-\vx$ we have that 
			\[
				\sum_{j=1}^m \ml(g_jf_j)
				=
				\ml\left(\sum_{i=1}^m g_j(\vx)f_j(\vx)\right)=1
				\;,
			\]
			where we appealed to linearity of multilinearization (\autoref{fact:multilinearization}), so that
		}
		C(\vx,\vf,\vx^2-\vx)
		&=\sum_{j=1}^m \ml(g_jf_j)
		=1
		\;,
	\end{align}
	as desired.
\end{proof}

\subsection{Multilinear-Formula-\lbIPS}\label{sec:multilinearization:mult-form}

We now turn to proving that $g\cdot f\equiv \ml(g\cdot f) \mod \vx^2-\vx$ when $f$ is low-degree and $g$ is a multilinear formula.  This multilinearization can be used to complete our construction of multilinear-formula-IPS refutations of subset-sum axiom (\autoref{sec:subset:ub}), though our actual construction will multilinearize more directly (\autoref{res:ips-ubs:subset:mult-form}).  

More importantly, the multilinearization we establish here shows that multilinear-formula-IPS can efficiently simulate sparse-\lIPS (when the axioms are low-degree and multilinear).  Such a simulation holds intuitively, as multilinear formulas can efficiently compute any sparse (multilinear) polynomial, and as we work over the boolean cube we are morally working with multilinear polynomials.  While this intuition suggests that such a simulation follows immediately, this intuition is false. Specifically, while sparse-\lIPS is a complete refutation system for any system of unsatisfiable polynomials over the boolean cube, multilinear-formula-\lIPS is \emph{incomplete} as seen by the following example (though, by \autoref{thm:GrochowPitassi14}, multilinear-formula-\lIPS \emph{is} complete for refuting unsatisfiable CNFs).  

\begin{example}\label{ex:multi-form:incomplete}
	Consider the unsatisfiable system of equations $xy+1,x^2-x,y^2-y$.  A multilinear linear-IPS proof is a tuple of multilinear polynomials $(f,g,h)\in\F[x,y]$ such that $f\cdot (xy+1)+g\cdot (x^2-x)+h\cdot (y^2-y)=1$.  In particular, $f(x,y)=\frac{1}{xy+1}$ for $x,y\in\bits$, which implies by interpolation over the boolean cube that $f(x,y)=1\cdot (1-x)(1-y)+1\cdot (1-x)y+1\cdot x(1-y)+\frac{1}{2} \cdot xy=1-\frac{1}{2}\cdot xy$.  Thus $f\cdot (xy+1)$ contains the monomial $x^2y^2$.  However, as $g,h$ are multilinear we see that $x^2y^2$ cannot appear in $g\cdot (x^2-x)+h\cdot (y^2-y)-1$, so that the equality $f\cdot (xy+1)+g\cdot (x^2-x)+h\cdot (y^2-y)=1$ cannot hold. Thus, $xy+1,x^2-x,y^2-y$ has no linear-IPS refutation only using multilinear polynomials.
\end{example}

Put another way, the above example shows that in a linear-IPS refutation $\sum_j g_j f_j+\sum_i h_i\cdot (x_i^2-x_i)=1$, while one can multilinearize the $g_j$ (with a possible increase in complexity) and still retain a refutation, one cannot multilinearize the $h_i$ in general.

As such, to simulate sparse-\lIPS (a complete proof system) we must resort to using the more general \lbIPS over multilinear formulas, where recall that the \lbIPS refutation system considers refutes $\vf,\vx^2-\vx$ with a polynomial $C(\vx,\vy,\vz)$ where $C(\vx,\vnz,\vnz)=0$, $C(\vx,\vf,\vx^2-\vx)=1$, with the added restriction that when viewing $C\in\F[\vx,\vz][\vy]$ that the degree of $C$ with respect to the $\vy$-variables is at most 1, that is, $\deg_\vy C\le 1$. In fact, we in fact establish such a simulation using the subclass of refutations of the form $C(\vx,\vy,\vz)=\sum_j g_j y_j+C'(\vx,\vz)$ where $C'(\vx,\vnz)=0$.  Note that such refutations are only linear in the \emph{non}-boolean axioms, which allows us to circumvent \autoref{ex:multi-form:incomplete}.

We now show how to prove $g\cdot f\equiv \ml(g\cdot f) \mod \vx^2-\vx$ when $f$ is low-degree and $g$ is a multilinear formula, where the proof is supplied by the equality $g\cdot f-\ml(g\cdot f)=C(\vx,\vx^2-\vx)$ where $C(\vx,\vnz)=0$, and where we seek to show that $C$ has a small multilinear formula.  We begin with the special case of $f$ and $g$ being the same monomial.

\begin{lemmawp}\label{res:multilin:mon}
	Let $\vx^\vno=\prod_{i=1}^n x_i$. Then,
	\[
		(\vx^{\vno})^2-\vx^{\vno}
		=C(\vx,\vx^2-\vx)
		\;,
	\]
	where $C(\vx,\vz)\in\F[\vx,z_1,\ldots,z_n]$ is defined by
	\[
		C(\vx,\vz)
		\eqdef 
		(\vz+\vx)^\vno-\vx^\vno
		=\sum_{\vnz<\va\le \vno} \vz^\va\vx^{\vno-\va}
		\;,
	\]
	so that $C(\vx,\vnz)=0$.
\end{lemmawp}

Note that the first expression for $C$ is a $\poly(n)$-sized depth-3 expression, while the second is an $\exp(n)$-sized depth-2 expression.  This difference will, going forward, show that we can multilinearize efficiently in depth-3, but can only efficiently multilinearize \emph{low-degree} monomials in depth-2. We now give an IPS proof for showing how a monomial times a multilinear formula equals its multilinearization.

\begin{lemma}\label{res:multilin:mon_formula}
	Let $g\in\F[x_1,\ldots,x_n,y_1,\ldots,y_d]$ be a multilinear polynomial. Then there is a $C\in\F[\vx,\vy,z_1,\ldots,z_d]$ such that
	\[
		g(\vx,\vy)\cdot \vy^\vno-\ml(g(\vx,\vy)\cdot \vy^\vno)
		=C(\vx,\vy,\vy^2-\vy)
		\;,
	\]
	and $C(\vx,\vy,\vnz)=0$. 
	
	If $g$ is $t$-sparse, then $C$ is computable as a $\poly(t,n,2^d)$-size depth-2 multilinear formula (which is $\poly(t,n,2^d)$-explicit given the computation for $g$), as well as being computable by a $\poly(t,n,d)$-size depth-3 multilinear formula with a $+$-output-gate (which is $\poly(t,n,d)$-explicit given the computation for $g$).  If $g$ is computable by a size-$t$ depth-$D$ multilinear formula, then $C$ is computable by a $\poly(t,2^d)$-size depth-$(D+2)$ multilinear formula with a $+$-output-gate (which is $\poly(t,2^d)$-explicit given the computation for $g$).
\end{lemma}
\begin{proof}
	\uline{defining $C$:} Express $g$ as $g=\sum_{\vnz\le\va\le\vno} g_\va(\vx)\vy^\va$ in the ring $\F[\vx][\vy]$, so that each $g_\va$ is multilinear. Then
	\begin{align*}
		g(\vx,\vy)\cdot \vy^\vno-\ml(g(\vx,\vy)\cdot \vy^\vno)
		&=\sum_{\va} g_\va(\vx)\vy^\va \cdot \vy^\vno-\ml\left(\sum_{\va} g_\va(\vx)\vy^\va\cdot \vy^\vno\right)\\
		&=\sum_{\va} g_\va(\vx)\left(\vy^{\va+\vno}-\vy^{\vno}\right)
		=\sum_{\va} g_\va(\vx)\vy^{\vno-\va}\left((\vy^\va)^2-\vy^{\va}\right)
		\;,
	\intertext{and appealing to \autoref{res:multilin:mon},}
		&=\sum_{\va} g_\va(\vx)\vy^{\vno-\va}\left(((\vy^2-\vy)+\vy)^\va-\vy^\va\right)\\
		&=C(\vx,\vy,\vy^2-\vy)
		\;,
	\end{align*}
	where we define
	\[
		C(\vx,\vy,\vz)\eqdef \sum_{\va} g_\va(\vx)\vy^{\vno-\va}\left((\vz+\vy)^\va-\vy^{\va}\right)
		\;.
	\]
	As $(\vz+\vy)^\va=\vy^{\va}$ under $\vz\leftarrow\vnz$ we have that $C(\vx,\vy,\vnz)=0$.
	
	\uline{$g$ is $t$-sparse:} As $g$ is $t$-sparse and multilinear, so are each $g_\va$, so that $g_\va=\sum_{i=1}^t \alpha_{\va,i}\vx^{\vc_{\va,i}}$, and thus
	\begin{align*}
		C(\vx,\vy,\vz)
		&=\sum_\va \sum_{i=1}^t \alpha_{\va,i}\vx^{\vc_{\va,i}}\vy^{\vno-\va}\left((\vz+\vy)^\va-\vy^{\va}\right)\\
		&=\sum_\va \sum_{i=1}^t \alpha_{\va,i}\vx^{\vc_{\va,i}}\vy^{\vno-\va}(\vz+\vy)^\va
			-\sum_\va \sum_{i=1}^t \alpha_{\va,i}\vx^{\vc_{\va,i}}\vy^\vno\;,
		\intertext{where this is clearly an explicit depth-3 ($\sum\prod\sum$) multilinear formula (as $\vy^{\vno-\va}$ is variable-disjoint from $(\vz+\vy)^\va$), and the size of this formula is $\poly(n,t,d)$ as there are at most $t$ such $\va$ where $g_\va\ne 0$ as $g$ is $t$-sparse. Continuing the expansion, appealing to \autoref{res:multilin:mon},}
		&=\sum_\va \sum_{i=1}^t \alpha_{\va,i}\vx^{\vc_{\va,i}}\vy^{\vno-\va}\sum_{0\le \vb\le \va}\vz^{\vb}\vy^{\va-\vb}
			-\sum_\va \sum_{i=1}^t \alpha_{\va,i}\vx^{\vc_{\va,i}}\vy^\vno\\
		&=\sum_\va \sum_{i=1}^t \alpha_{\va,i}\vx^{\vc_{\va,i}}\sum_{\vnz<\vb\le \va}\vz^{\vb}\vy^{\va-\vb}\\
		&=\sum_\va \sum_{i=1}^t \sum_{\vnz<\vb\le \va}\alpha_{\va,i}\vx^{\vc_{\va,i}}\vz^{\vb}\vy^{\va-\vb}
		\;,
	\end{align*}
	which is then easily seen to be explicit and $\poly(n,t,2^d)$-sparse appealing to the above logic.

	\uline{$g$ a multilinear formula:} By interpolation, it follows that for each $\va$ there are $\poly(2^d)$-explicit constants $\vaa_{\va,\vbb}$ such that $g_\va(\vx)=\sum_{\vbb\in\bits^d} \alpha_{\va,\vbb}g(\vx,\vbb)$.  From this it follows that $g_\va$ is computable by a depth $D+1$ multilinear formula of size $\poly(t,2^d)$.  Expanding the definition of $C$ we get that
	\begin{align*}
		C(\vx,\vy,\vz)
		&=\sum_{\va} g_\va(\vx)\vy^{\vno-\va}\left((\vz+\vy)^\va-\vy^{\va}\right)\\
		&=\sum_{\va} g_\va(\vx)\vy^{\vno-\va}(\vz+\vy)^\va
			-\sum_{\va} g_\va(\vx)\vy^\vno\\
		&=\sum_{\va} \sum_{\vbb\in\bits^d} \alpha_{\va,\vbb}g(\vx,\vbb)\vy^{\vno-\va}(\vz+\vy)^\va
			-\sum_{\va} \sum_{\vbb\in\bits^d} \alpha_{\va,\vbb}g(\vx,\vbb)\vy^\vno
		\;,
	\end{align*}
	which is clearly an explicit depth $D+2$ multilinear formula of size $\poly(t,2^d)$, as $D\ge 1$ so that the computation of the $z_i+y_i$ is parallelized with computing $g(\vx,\vbb)$, and we absorb the subtraction into the overall top-gate of addition.
\end{proof}

We can then straightforwardly extend this to multilinearizing the product of a low-degree sparse multilinear polynomial and a multilinear formula, as we can use that multilinearization is linear.

\begin{corollarywp}\label{res:multilin:sparse_formula}
	Let $g\in\F[x_1,\ldots,x_n]$ be a multilinear polynomial and $f\in\F[\vx]$ a $s$-sparse multilinear polynomial of degree $\le d$. Then there is a $C\in\F[\vx,z_1,\ldots,z_n]$ such that
	\[
		g\cdot f-\ml(g\cdot f)
		=C(\vx,\vx^2-\vx)
		\;,
	\]
	and $C(\vx,\vnz)=0$. 
	
	If $g$ is $t$-sparse, then $C$ is computable as a $\poly(s,t,n,2^d)$-size depth-2 multilinear formula (which is $\poly(s,t,n,2^d)$-explicit given the computations for $f,g$), as well as being computable by a $\poly(s,t,n,d)$-size depth-3 multilinear formula with a $+$-output-gate (which is $\poly(s,t,n,d)$-explicit given the computations for $f,g$).  If $g$ is computable by a size-$t$ depth-$D$ multilinear formula, then $C$ is computable by a $\poly(s,t,2^d)$-size depth-$(D+2)$ multilinear formula with a $+$-output-gate (which is $\poly(s,t,2^d)$-explicit given the computations for $f,g$).
\end{corollarywp}

We now show how to derive multilinear-formula-\lbIPS refutations for $\vf,\vx^2-\vx$ from equations of the form $\sum_j g_jf_j\equiv 1\mod\vx^2-\vx$.

\begin{corollary}\label{res:multilin:multilin-lbIPS}
	Let $f_1,\ldots,f_m\in\F[x_1,\ldots,x_n]$ be degree at most $d$ multilinear $s$-sparse polynomials which are unsatisfiable over $\vx\in\bits^n$.  Suppose that there are multilinear $g_j\in\F[\vx]$ such that
	\[
		\sum_{i=1}^m g_j(\vx)f_j(\vx)\equiv 1\mod\vx^2-\vx
		\;.
	\]
	Then there is a multilinear-formula-\lbIPS refutation $C(\vx,\vy,\vz)$ such that
	\begin{itemize}
		\item If the $g_j$ are $t$-sparse, then $C$ is computable by a $\poly(s,t,n,m,2^d)$-size depth-2 multilinear formula (which is $\poly(s,t,n,m,2^d)$-explicit given the computations for the $f_j,g_j$), as well as being computable by a $\poly(s,t,n,m,d)$-size depth-3 multilinear formula (which is $\poly(s,t,n,m,d)$-explicit given the computations for $f_j,g_j$).
		\item If the $g_j$ are computable by size-$t$ depth-$D$ multilinear formula, then $C$ is computable by a $\poly(s,t,m,2^d)$-size depth-$(D+2)$ multilinear formula (which is $\poly(s,t,m,2^d)$-explicit given the computations for $f_j,g_j$).
	\end{itemize}
\end{corollary}
\begin{proof}
	\uline{construction:}
	By the above multilinearization (\autoref{res:multilin:sparse_formula}), there are $C_j\in\F[\vx,\vz]$ such that
	\[
		g_j(\vx)f_j(\vx)=\ml(g_jf_j)+C_j(\vx,\vx^2-\vx)
		\;.
	\]
	where $C_j(\vx,\vnz)=0$. We now define
	\[
		C(\vx,\vy,\vz)\eqdef \sum_{j=1}^m \left (g_j(\vx)y_j - C_j(\vx,\vz)\right)
		\;.
	\]
	We now show that this is a valid IPS refutation. Observe that $C(\vx,\vnz,\vnz)=0$ and that
	\begin{align*}
		C(\vx,\vf,\vx^2-\vx)
		&=\sum_{j=1}^m \left (g_j(\vx)f_j(\vx) - C_j(\vx,\vx^2-\vx)\right)\\\
		&=\sum_{j=1}^m \ml(g_jf_j)\\
		\intertext{as $\sum_{i=1}^m g_j(\vx)f_j(\vx)\equiv 1\mod \vx^2-\vx$ we have that 
			\[
				\sum_{j=1}^m \ml(g_jf_j)
				=
				\ml\left(\sum_{i=1}^m g_j(\vx)f_j(\vx)\right)=1
				\;,
			\]
			where we appealed to linearity of multilinearization (\autoref{fact:multilinearization}), so that
		}
		C(\vx,\vf,\vx^2-\vx)
		&=\sum_{j=1}^m \ml(g_jf_j)
		=1
		\;.
	\end{align*}

	\uline{complexity:}
	The claim now follows from appealing to \autoref{res:multilin:sparse_formula} for bounding the complexity of the $C_j$.  That is, if the $g_j$ are $t$-sparse then $\sum_{j=1}^m g_j(\vx)y_j$ is $tm$-sparse and thus computable by a $\poly(t,m)$-size depth-2 multilinear formula with $+$-output-gate. As each $C_j$ is computable by a $\poly(s,t,n,2^d)$-size depth-2 or $\poly(s,t,n,d)$-size depth-3 multilinear formula (each having a $+$-output-gate), it follows that $C(\vx,\vy,\vz)\eqdef \sum_{j=1}^m \left (g_j(\vx)y_j - C_j(\vx,\vz)\right)$ can be explicitly computed by a $\poly(s,t,n,m,2^d)$-size depth-2 or $\poly(s,t,n,m,d)$-size depth-3 multilinear formula (collapsing addition gates into a single level).

	If the $g_j$ are computable by size-$t$ depth-$D$ multilinear formula then $\sum_{j=1}^m g_j(\vx)y_j$ is computable by size $\poly(m,t)$-size depth-$(D+2)$ multilinear formula (with a $+$-output-gate), and each $C_j$ is computable by a $\poly(s,t,2^d)$-size depth-$(D+2)$ multilinear formula with a $+$-output-gate, from which it follows that $C(\vx,\vy,\vz)\eqdef \sum_{j=1}^m \left (g_j(\vx)y_j - C_j(\vx,\vz)\right)$ can be explicitly computed by a $\poly(s,t,m,2^d)$-size depth-$(D+2)$ multilinear formula by collapsing addition gates.
\end{proof}

We now conclude by showing that multilinear-formula-\lbIPS can efficiently simulate sparse-\lIPS when the axioms are low-degree. As this latter system is complete, this shows the former is as well. That is, we allow the refutation to depend non-linearly on the boolean axioms, but only linearly on the other axioms.

\begin{corollary}\label{res:multilin:simulate-sparse}
	Let $f_1,\ldots,f_m\in\F[x_1,\ldots,x_n]$ be degree at most $d$ $s$-sparse polynomials unsatisfiable over $\vx\in\bits^n$.  Suppose they have an $\sumprod$-\lIPS refutation, that is, that there are $t$-sparse polynomials $g_1,\ldots,g_m,h_1,\ldots,h_n\in\F[\vx]$ such that $\sum_{j=1}^m g_j f_j+\sum_{i=1}^n h_i \cdot (x_i^2-x_i)=1$.  Then $\vf,\vx^2-\vx$ can be refuted by a depth-2 multilinear-formula-\lbIPS proof of $\poly(s,t,n,m,2^d)$-size, or by a depth-3 multilinear-formula-\lbIPS proof of $\poly(s,t,n,m,d)$-size.
\end{corollary}
\begin{proof}
	This follows from \autoref{res:multilin:multilin-lbIPS} by noting that $\sum_j \ml(g_j)f_j\equiv 1\mod \vx^2-\vx$, and that the $\ml(g_j)$ are multilinear and $t$-sparse.
\end{proof}

\subsection{Refutations of the Subset-Sum Axiom}\label{sec:subset:ub}

We now give efficient IPS refutations of the subset-sum axiom, where these IPS refutations can be even placed in the restricted roABP-\lIPS or multilinear-formula-\lIPS subclasses.  That is, we give such refutations for whenever the polynomial $\sum_i \alpha_ix_i-\beta$ is unsatisfiable over the boolean cube $\bits^n$, where the size of the refutation is polynomial in the size of the set $A\eqdef\{\sum_i \alpha_ix_i: \vx\in\bits^n\}$.  A motivating example is when $\vaa=\vno$ so that $A=\{0,\ldots,n\}$.

To construct our refutations, we first show that there is an efficiently computable polynomial $f$ such that $f(\vx)\cdot (\sum_i \alpha_ix_i-\beta)\equiv 1\mod\vx^2-\vx$.  This will be done by considering the univariate polynomial $p(t)\eqdef\prod_{\alpha\in A}(t-\alpha)$. Using that for any univariate $p(x)$ that $x-y$ divides $p(x)-p(y)$, we see that $p(\sum_i\alpha_ix_i)-p(\beta)$ is a multiple of $\sum_i \alpha_i x_i-\beta$.  As $\sum_i \alpha_ix_i-\beta$ is unsatisfiable it must be that $\beta\notin A$. This implies that $p(\sum_i\alpha_ix_i)\equiv 0\mod\vx^2-\vx$ while $p(\beta)\ne 0$.  Consequently, $p(\sum_i\alpha_ix_i)-p(\beta)$ is equivalent to a nonzero constant modulo $\vx^2-\vx$, yielding the Nullstellensatz refutation
\[
	\frac{1}{-p(\beta)}\cdot \frac{p(\sum_i\alpha_ix_i)-p(\beta)}{\sum_i \alpha_i x_i-\beta}\cdot (\tsum_i \alpha_i x_i-\beta)\equiv 1 \mod\vx^2-\vx
	\;.
\]
By understanding the quotient $\frac{p(\sum_i\alpha_ix_i)-p(\beta)}{\sum_i \alpha_i x_i-\beta}$ we see that it can be efficiently computed by a small $\sumpowsum$ formula and thus an roABP, so that using our multilinearization result for roABPs (\autoref{res:multilin:roABP}) this yields the desired roABP-\lIPS refutation.  However, this does not yield the desired multilinear-formula-\lIPS refutation.  For this, we need to (over a large field) convert the above quotient to a sum of products of univariates using duality (\autoref{res:sumpowsum:duality}).  We can then multilinearize this to a sum of products of \emph{linear} univariates, which is a depth-3 multilinear formula.  By appealing to our proof-of-multilinearization result for multilinear formula (\autoref{res:multilin:multilin-lbIPS}) one obtains a multilinear-formula-\lbIPS refutation, and we give a direct proof which actually yields the desired multilinear-formula-\lIPS refutation.

We briefly remark that for the special case of $\vaa=\vno$, one can explicitly describe the unique multilinear polynomial $f$ such that $f(\vx)(\sum_i x_i-\beta)\equiv 1\mod\vx^2-\vx$.  This description (\autoref{res:subsetsum:multlin}) shows that $f$ is a linear combination of elementary symmetric polynomials, which implies the desired complexity upper bounds for this case via known bounds on the complexity of elementary symmetric polynomials (\cite{NisanWigderson96}).  However, this proof strategy is more technical and thus we pursue the more conceptual outline given above to bound the complexity of $f$ for general $A$.

\begin{proposition}\label{res:ips-ubs:subset}
	Let $\vaa\in\F^n$, $\beta\in\F$ and $A\eqdef\{\sum_{i=1}^n \alpha_i x_i : \vx\in\bits^n\}$ be so that $\beta\notin A$. Then there is a multilinear polynomial $f\in\F[\vx]$ such that
	\[
		f(\vx)\cdot \left(\tsum_i \alpha_i x_i-\beta\right)\equiv 1\mod  \vx^2-\vx
		\;.
	\]
	For any $|\F|$, $f$ is computable by a $\poly(|A|,n)$-explicit $\poly(|A|,n)$-width roABP of individual degree $\le 1$.

	If $|\F|\ge\poly(|A|,n)$, then $f$ is computable as a sum of product of linear univariates (and hence a depth-3 multilinear formula)
	\[
		f(\vx)=\sum_{i=1}^s f_{i,1}(x_1)\cdots f_{i,n}(x_n)
		\;,
	\]
	where each $f_{i,j}\in\F[x_i]$ has $\deg f_{i,j}\le 1$, $s\le \poly(|A|,n)$, and this expression is $\poly(|A|,n)$-explicit.
\end{proposition}
\begin{proof}
	\uline{computing $A$:} We first note that $A$ can be computed from $\vaa$ in $\poly(|A|,n)$-steps (as opposed to the naive $2^n$ steps).  That is, define $A_j\eqdef\{\sum_{i=1}^j \alpha_i x_i : x_1,\ldots,x_j\in\bits\}$, so that $A_0=\emptyset$ and $A_n=A$.  It follows that for all $j$, we have that $A_j\subseteq A$ and thus $|A_j|\le |A|$, and that $A_{j+1}\subseteq A_j\cup (A_j+\alpha_j)$.  It follows that we can compute $A_{j+1}$ from $A_j$ in $\poly(|A|)$ time, so that $A=A_n$ can be computed in $\poly(|A|,n)$-time.
	
	\uline{defining $f$:} 
	Define $p(t)\in\F[t]$ by $p\eqdef\prod_{\alpha\in A}(t-\alpha)$, so that $p(A)=0$ and $p(\beta)\ne 0$.  Express $p(t)$ in its monomial representation as $p(t)=\sum_{k=0}^{|A|} \gamma_k t^k$, where the $\gamma_k$ can be computed in $\poly(|A|)$ time from $\vaa$ by computing $A$ as above.  Then observe that
	\begin{align*}
		p(t)-p(\beta)
		&=\left(\sum_{k=0}^{|A|} \gamma_k\frac{t^k-\beta^k}{t-\beta}\right)(t-\beta)\\
		&=\left(\sum_{k=0}^{|A|} \gamma_k\sum_{j=0}^{k-1} t^{j}\beta^{(k-1)-j}\right)(t-\beta)\\
		&=\left(\sum_{j=0}^{|A|-1}\left(\sum_{k=j+1}^{|A|} \gamma_k\beta^{(k-1)-j} \right)t^j\right)(t-\beta)
		\;.
	\end{align*}
	Thus, plugging in $t\leftarrow \sum_i \alpha_i x_i$, we can define the polynomial $g(\vx)\in\F[\vx]$ by
	\begin{align}
		g(\vx)
		&\eqdef\frac{p(\sum_i \alpha_i x_i)-p(\beta)}{\sum_i \alpha_i x_i -\beta}\nonumber\\
		&=\sum_{j=0}^{|A|-1}\left(\sum_{k=j+1}^{|A|} \gamma_k\beta^{(k-1)-j} \right) \left(\sum_i \alpha_i x_i\right)^j\label{eq:subset}
		\;.
	\end{align}
	Hence,
	\begin{align*}
		g(\vx) (\tsum_i \alpha_ix_i-\beta)
		&=p(\tsum_i \alpha_i x_i)-p(\beta)
		\;.
	\end{align*}
	For any $\vx\in\bits^n$ we have that $\sum_i \alpha_i x_i\in A$.  As $p(A)=0$ it follows that $p(\sum_i \alpha_i x_i)=0$ for all $\vx\in\bits^n$. This implies that $p(\sum_i \alpha_i x_i)\equiv 0 \mod \vx^2-\vx$, yielding
	\begin{align*}
		g(\vx) (\tsum_i \alpha_ix_i-\beta)
		&\equiv -p(\beta)\mod\vx^2-\vx
		\;.
	\end{align*}
	As $-p(\beta)\in\F\setminus\{0\}$, we have that
	\[
		\frac{1}{-p(\beta)}\cdot g(\vx) \cdot \left(\tsum_i \alpha_ix_i-\beta\right)\equiv 1 \mod \vx^2-\vx
		\;.
	\]
	We now simply multilinearize, and thus define the multilinear polynomial 
	\[
		f(\vx)\eqdef \ml\left(\frac{1}{-p(\beta)}\cdot g(\vx)\right)
		\;.
	\]
	First, we see that this has the desired form, using the interaction of multilinearization and multiplication (\autoref{fact:multilinearization}).
	\begin{align*}
		1
		&=\ml\left(\frac{1}{-p(\beta)} g(\vx) \cdot (\tsum_i \alpha_ix_i-\beta)\right)\\
		&=\ml\left(\ml\left(\frac{1}{-p(\beta)}\cdot g(\vx)\right)\ml(\tsum_i \alpha_ix_i-\beta)\right)\\
		&=\ml\Big(f \cdot \ml(\tsum_i \alpha_ix_i-\beta)\Big)\\
		&=\ml\Big(f\cdot (\tsum_i \alpha_ix_i-\beta)\Big)
	\end{align*}
	Thus, $f\cdot (\tsum_i \alpha_ix_i-\beta)\equiv 1\mod\vx^2-\vx$ as desired.

	\uline{computing $f$ as an roABP:}  By \autoref{eq:subset} we see that $g(\vx)$ is computable by a $\poly(|A|,n)$-size $\sumpowsum$-formula, and by the efficient simulation of $\sumpowsum$-formula by roABPs (\autoref{res:sumpowsum:roABP}) $g(\vx)$ and thus $\frac{1}{-p(\beta)}\cdot g(\vx)$ are computable by $\poly(|A|,n)$-width roABPs of $\poly(|A|,n)$-degree. Noting that roABPs can be efficiently multilinearized (\autoref{res:multilin:roABP}) we see that $f=\ml(\frac{1}{-p(\beta)}\cdot g(\vx))$ can thus be computed by such an roABP also, where the individual degree of this roABP is at most 1. Finally, note that each of these steps has the required explicitness.

	\uline{computing $f$ via duality:} We apply duality (\autoref{res:sumpowsum:duality}) to see that over large enough fields there are univariates $g_{j,\ell,i}$ of degree at most $|A|$, where
	\begin{align*}
		g(\vx)
		&=\sum_{j=0}^{|A|-1}\left(\sum_{k=j+1}^{|A|} \gamma_k\beta^{(k-1)-j} \right) \left(\sum_i \alpha_i x_i\right)^j\\
		&=\sum_{j=0}^{|A|-1}\left(\sum_{k=j+1}^{|A|} \gamma_k\beta^{(k-1)-j} \right) \sum_{\ell=1}^{(nj+1)(j+1)}g_{j,\ell,1}(x_1)\cdots g_{j,\ell,n}(x_n)\\
		\intertext{Absorbing the constant $\left(\sum_{k=j+1}^{|A|} \gamma_k\beta^{(k-1)-j} \right)$ into these univariates and re-indexing,}
		&=\sum_{i=1}^{s} g_{i,1}(x_1)\cdots g_{i,n}(x_n)
	\end{align*}
	for some univariates $g_{i,j}$, where $s\le |A|(n|A|+1)(|A|+1)=\poly(|A|,n)$. 

	We now obtain $f$ by multilinearizing the above expression, again appealing to multilinearization (\autoref{fact:multilinearization}).
	\begin{align*}
		f
		&=\ml\left(\frac{1}{-p(\beta)} g(\vx)\right)\\
		&=\ml\left(\frac{1}{-p(\beta)} \sum_{i=1}^{s} g_{i,1}(x_1)\cdots g_{i,n}(x_n)\right)\\
		\intertext{absorbing the constant $\nicefrac{1}{-p(\beta)}$ and renaming,}
		&=\ml\left(\sum_{i=1}^{s} g'_{i,1}(x_1)\cdots g'_{i,n}(x_n)\right)\\
		&=\ml\left(\sum_{i=1}^{s} \ml(g'_{i,1}(x_1))\cdots \ml(g'_{i,n}(x_n))\right)\\
		\intertext{defining $f_{i,j}(x_j)\eqdef \ml(g_{i,j}(x_j))$, so that $\deg f_{i,j}\le 1$,}
		&=\ml\left(\sum_{i=1}^{s} f_{i,1}(x_1)\cdots f_{i,n}(x_n)\right)\\
		\intertext{and we can drop the outside $\ml$ as this expression is now multilinear,}
		&=\sum_{i=1}^{s} f_{i,1}(x_1)\cdots f_{i,n}(x_n)
		\;,
	\end{align*}
	showing that $f$ is computable as desired, noting that this expression has the desired explicitness.
\end{proof}

Note that computing $f$ via duality also implies an roABP for $f$, but only in large enough fields $|\F|\ge\poly(|A|,n)$.  Of course, the field must at least have $|\F|\ge |A|$, but by using the field-independent conversion of $\sumpowsum$ to roABP (\autoref{res:sumpowsum:roABP}) this shows that $\F$ need not be any larger than $A$ for the refutation to be efficient.

The above shows that one can give an ``IPS proof'' $g(\vx)(\sum_i \alpha_ix_i-\beta)+\sum_i h_i(\vx)(x_i^2-x_i)=1$, where $g$ is efficiently computable.  However, this is not yet an IPS proof as it does not bound the complexity of the $h_i$. We now extend this to an actual IPS proof by using the above multilinearization results for roABPs (\autoref{res:multilin:roABP-lIPS}), leveraging that $\sum_i \alpha_ix_i-\beta$ is computable by an roABP in any variable order (and that the above result works in any variable order).

\begin{corollarywp}\label{res:ips-ubs:subset:roABP}
	Let $\vaa\in\F^n$, $\beta\in\F$ and $A\eqdef\{\sum_{i=1}^n \alpha_i x_i : \vx\in\bits^n\}$ be so that $\beta\notin A$.  Then $\sum_i \alpha_ix_i-\beta,\vx^2-\vx$ has a $\poly(|A|,n)$-explicit roABP-\lIPS refutation of individual degree $2$ computable in width-$\poly(|A|,n)$ in any variable order.
\end{corollarywp}

Note that while the above results give a small $\sumpowsum$ formula $g$ such that $g\cdot (\sum_i\alpha_i x_i-\beta)\equiv -p(\beta)\mod\vx^2-\vx$ for nonzero scalar $-p(\beta)$, this does not translate to a $\sumpowsum$-IPS refutation as $\sumpowsum$ formulas cannot be multilinearized efficiently (see the discussion in \autoref{sec:multilinearization:roABP}).

We now turn to refuting the subset-sum axioms by multilinear-formula-\lIPS (which is not itself a complete proof system, but will suffice here).  While one can use the multilinearization techniques for multilinear-formula-\lbIPS of \autoref{res:multilin:multilin-lbIPS} it gives slightly worse results due to its generality, so we directly multilinearize the refutations we built above using that the subset-sum axiom is linear.

\begin{proposition}\label{res:ips-ubs:subset:mult-form}
	Let $\vaa\in\F^n$, $\beta\in\F$ and $A\eqdef\{\sum_{i=1}^n \alpha_i x_i : \vx\in\bits^n\}$ be so that $\beta\notin A$.  If $|\F|\ge \poly(|A|,n)$, then $\sum_i \alpha_ix_i-\beta,\vx^2-\vx$ has a $\poly(|A|,n)$-explicit $\poly(|A|,n)$-size depth-3 multilinear-formula-\lIPS refutation.
\end{proposition}
\begin{proof}
	By \autoref{res:ips-ubs:subset}, there is a multilinear polynomial $f\in\F[\vx]$ such that $f(\vx)\cdot \left(\tsum_i \alpha_i x_i-\beta\right)\equiv 1\mod  \vx^2-\vx$, and $f$ is explicitly given as
	\[
		f(\vx)=\sum_{i=1}^s f_{i,1}(x_1)\cdots f_{i,n}(x_n)
		\;,
	\]
	where each $f_{i,j}\in\F[x_i]$ has $\deg f_{i,j}\le 1$ and $s\le \poly(|A|,n)$. 

	We now efficiently prove that $f(\vx)\cdot (\sum_{i=1}^n \alpha_i x_i-\beta)$ is equal to its multilinearization (which is $1$) modulo the boolean cube. The key step is that for a linear function $p(x)=\gamma x+\delta$ we have that $(\gamma x+\delta)x=(\gamma+\delta)x+\gamma(x^2-x)=p(1)x+(p(1)-p(0))(x^2-x)$.

	Thus,
	\begin{align*}
		f(\vx)&\cdot \left(\tsum_i\alpha_i x_i-\beta\right)\\
		&=\left( \sum_{i=1}^s f_{i,1}(x_1)\cdots f_{i,n}(x_n) \right)\cdot (\tsum_i\alpha_i x_i-\beta)\\
		&=\sum_{i=1}^s -\beta f_{i,1}(x_1)\cdots f_{i,n}(x_n)\\
		&\hspace{.5in}+\sum_{i=1}^s\sum_{j=1}^n \alpha_j \prod_{k\ne j} f_{i,k}(x_k) \cdot \Big(f_{i,j}(1)x_j + (f_{i,j}(1)-f_{i,j}(0))(x_j^2-x_j)\Big)\\
		&=\sum_{i=1}^s -\beta f_{i,1}(x_1)\cdots f_{i,n}(x_n)
		+\sum_{i=1}^s\sum_{j=1}^n \alpha_j \prod_{k\ne j} f_{i,k}(x_k) \cdot f_{i,j}(1)x_j \\
		&\hspace{.5in}+\sum_{i=1}^s\sum_{j=1}^n \alpha_j \prod_{k\ne j} f_{i,k}(x_k) \cdot (f_{i,j}(1)-f_{i,j}(0))\cdot (x_j^2-x_j)
		\intertext{absorbing constants and renaming, using $j=0$ to incorporate the above term involving $\beta$,}
		&=\sum_{i=1}^s\sum_{j=0}^n \prod_{k=1}^n f_{i,j,k}(x_k)+\sum_{j=1}^n \left(\sum_{i=1}^s\prod_{k=1}^n h_{i,j,k}(x_k)\right)(x_j^2-x_j)
		\intertext{where each $f_{i,j,k}$ and $h_{i,j,k}$ are linear univariates. As $f(\vx)\cdot (\sum_{i=1}^n \alpha_i x_i-\beta)\equiv 1\mod \vx^2-\vx$ it follows that $\sum_i \sum_j \prod_k f_{i,j,k}(x_k)\equiv 1\mod\vx^2-\vx$, but as each  $f_{i,j,k}$ is linear it must actually be that $\sum_i \sum_j\prod_k f_{i,j,k}(x_k)=1$, so that,}
		&=1+\sum_{j=1}^n \left(\sum_{i=1}^s\prod_{k=1}^n h_{i,j,k}(x_k)\right)(x_j^2-x_j)
		\;.
	\end{align*}
	Define $C(\vx,y,\vz)\eqdef f(\vx)y-\sum_{j=1}^n h_j(\vx) z_j$, where $h_j(\vx)\eqdef \sum_{i=1}^s\prod_{k=1}^n h_{i,j,k}(x_k)$.  It follows that $C(\vx,0,\vnz)=0$ and that $C(\vx,\sum_i\alpha_ix_i-\beta,\vx^2-\vx)=1$, so that $C$ is a linear-IPS refutation.  Further, as each $f,h_j$ is computable as a sum of products of linear univariates, these are depth-3 multilinear formulas. By distributing the multiplication of the variables $y,z_1,\ldots,z_n$ to the bottom multiplication gates, we see that $C$ itself has a depth-3 multilinear formula of the desired complexity.
\end{proof}

\section{Lower Bounds for Linear-IPS via Functional Lower Bounds}\label{sec:lbs-fn}

In this section we give \emph{functional} circuit lower bounds for various measures of algebraic complexity, such as degree, sparsity, roABPs and multilinear formulas.  That is, while algebraic complexity typically treats a polynomial $f\in\F[x_1,\ldots,x_n]$ as a \emph{syntactic} object, one can also see that it defines a function on the boolean cube $\hat{f}:\bits^n\to\F$.  If this function is particularly complicated then one would expect that this implies that the polynomial $f$ requires large algebraic circuits. In this section we obtain such lower bounds, showing that for \emph{any} polynomial $f$ (not necessarily multilinear) that agrees with a certain function on the boolean cube must in fact have large algebraic complexity.  

Our lower bounds will proceed by first showing that the complexity of $f$ is an upper bound for the complexity of its multilinearization $\ml(f)$.  While such a statement is known to be false for general circuits (under plausible assumptions, see \autoref{sec:multilinearization:roABP}), such efficient multilinearization can be shown for the particular restricted models of computation we consider.  In particular, this multilinearization is easy for degree and sparsity, for multilinear formulas $f$ is already multilinear, and for roABPs this is seen in \autoref{sec:multilinearization:roABP}.  As then $\ml(f)$ is uniquely defined by the function $\hat{f}$ (\autoref{fact:multilinearization}), we then only need to lower bound the complexity of $\ml(f)$ using standard techniques.  We remark that the actual presentation will not follow the above exactly, as for roABPs and multilinear formulas it is just as easy to just work directly with the underlying lower bound technique.

We then observe that by deriving such lower bounds for carefully crafted functions $\hat{f}:\bits^n\to\F$ one can easily obtain linear-IPS lower bounds for the above circuit classes. That is, consider the system of equations $f(\vx),\vx^2-\vx$, where $f(\vx)$ is chosen so this system is unsatisfiable, hence $f(\vx)\ne0$ for all $\vx\in\bits^n$.  Any linear-IPS refutation yields an equation $g(\vx)\cdot f(\vx)+\sum_i h_i(\vx)(x_i^2-x_i)=1$, which implies that $g(\vx)=\nicefrac{1}{f(\vx)}$ for all $\vx\in\bits^n$ (that this system is unsatisfiable allows us to avoid division by zero). It follows that the polynomial $g(\vx)$ agrees with the function $\hat{f}(\vx)\eqdef \nicefrac{1}{f(\vx)}$ on the boolean cube. If the function $\hat{f}$ has a functional lower bound then this implies $g$ must have large complexity, giving the desired lower bound for the linear-IPS refutation.

The section proceeds as follows.  We begin by detailing the above strategy for converting functional lower bounds into lower bounds for linear-IPS\@. We then derive a tight functional lower bound of $n$ for the degree of $\nicefrac{1}{\left(\sum_i x_i+1\right)}$. We then extend this via random restrictions to a functional lower bound of $\exp(\Omega(n))$ on the sparsity of $\nicefrac{1}{\left(\sum_i x_i+1\right)}$.  We can then lift this degree bound to a functional lower bound of $2^n$ on the evaluation dimension of $\nicefrac{1}{\left(\sum_i x_iy_i+1\right)}$ in the $\vx|\vy$ partition, which we then symmetrize to obtain a functional lower bound on the evaluation dimension in any partition of the related function $\nicefrac{1}{\big(\sum_{i<j} z_{i,j}x_ix_j+1\big)}$.  In each case, the resulting linear-IPS lower bounds are immediate via the known relations of these measures to circuit complexity classes (\autoref{sec:background}).

\subsection{The Strategy}\label{sec:lbs-fn:strategy}

We give here the key lemma detailing the general reduction from linear-IPS lower bounds to functional lower bounds.

\begin{lemma}\label{res:lbs-fn_lbs-ips:lIPS}
	Let $\cC\subseteq\F[x_1,\ldots,x_n]$ be a set of polynomials closed under partial evaluation. Let $f\in\F[\vx]$, where the system $f(\vx),\vx^2-\vx$ is unsatisfiable.  Suppose that for all $g\in\F[\vx]$ with
	\[
		g(\vx)=\frac{1}{f(\vx)},\qquad \forall \vx\in\bits^n\;,
	\]
	that $g\notin \cC$.  Then $f(\vx),\vx^2-\vx$ does not have $\cC$-\lIPS refutations.
\end{lemma}
\begin{proof}
	Suppose for contradiction that $f(\vx),\vx^2-\vx$ has the $\cC$-\lIPS refutation $C(\vx,y,\vz)=g(\vx)\cdot y+\sum_i h_i(\vx) \cdot z_i$ where $C(\vx,f,\vx^2-\vx)=1$ (and clearly $C(\vx,0,\vnz)=0$). As $g=C(\vx,1,\vnz)$, it follows that $g\in\cC$ from the closure properties we assumed of $\cC$. Thus,
	\begin{align*}
		1&=C(\vx,f,\vx^2-\vx)\\
		&=g(\vx)\cdot f(\vx)+\sum_i h_i(\vx) (x_i^2-x_i)\\
		&\equiv g(\vx)\cdot f(\vx)\mod\vx^2-\vx
		\;.
	\end{align*}
	Thus, for any $\vx\in\bits^n$, as $f(\vx)\ne 0$,
	\[
		g(\vx)=\nicefrac{1}{f(\vx)}
		\;.
	\]
	However, this yields the desired contradiction, as this contradicts the assumed functional lower bound for $\nicefrac{1}{f}$.
\end{proof}

We now note that the lower bound strategy of using functional lower bounds actually produces lower bounds for \lbIPS (and even for the full IPS system if we have multilinear polynomials), and not just \lIPS.  This is because we work modulo the boolean axioms, so that any non-linear dependence on these axioms vanishes, only leaving a linear dependence on the remaining axioms.  This slightly stronger lower bound is most interesting for multilinear-formulas, where the \lIPS version is not complete in general (\autoref{ex:multi-form:incomplete}) (though it is still interesting due to its short refutations of the subset-sum axiom (\autoref{res:ips-ubs:subset:mult-form})), while the \lbIPS version is complete (\autoref{res:multilin:simulate-sparse}).

\begin{lemma}\label{res:lbs-fn_lbs-ips:lbIPS}
	Let $\cC\subseteq\F[x_1,\ldots,x_n]$ be a set of polynomials closed under evaluation, and let $\cD$ be the set of differences of $\cC$, that is, $\cD\eqdef \{p(\vx)-q(\vx): p,q\in\cC\}$. Let $f\in\F[\vx]$, where the system $f(\vx),\vx^2-\vx$ is unsatisfiable.  Suppose that for all $g\in\F[\vx]$ with
	\[
		g(\vx)=\frac{1}{f(\vx)},\qquad \forall \vx\in\bits^n\;,
	\]
	that $g\notin \cD$.  Then $f(\vx),\vx^2-\vx$ does not have $\cC$-\lbIPS refutations.

	Furthermore, if $\cC$ (and thus $\cD$) are a set of multilinear polynomials, then $f(\vx),\vx^2-\vx$ does not have $\cC$-IPS refutations.
\end{lemma}
\begin{proof}
	Suppose for contradiction that $f(\vx),\vx^2-\vx$ has the $\cC$-\lbIPS refutation $C(\vx,y,\vz)$.  That $\deg_{y} C(\vx,y,\vz)\le 1$ implies there is the decomposition $C(\vx,y,\vz)=C_1(\vx,\vz)y+C_0(\vx,\vz)$.  As $C_1(\vx,\vnz)=C(\vx,1,\vnz)-C(\vx,0,\vnz)$, the assumed closure properties imply that $C_1(\vx,\vnz)\in\cD$. By the definition of an IPS refutation, we have that $0=C(\vx,0,\vnz)=C_1(\vx,\vnz)\cdot 0+C_0(\vx,\vnz)$, so that $C_0(\vx,\vnz)=0$. By using the definition of an IPS refutation again, we have that $1=C(\vx,f,\vx^2-\vx)=C_1(\vx,\vx^2-\vx)\cdot f+C_0(\vx,\vx^2-\vx)$, so that modulo the boolean axioms,
	\begin{align*}
		1&=\textstyle C_1(\vx,\vx^2-\vx)\cdot f+C_0(\vx,\vx^2-\vx) \\
		&\equiv \textstyle C_1(\vx,\vnz)\cdot f+C_0(\vx,\vnz) \mod \vx^2-\vx\\
		\intertext{using that $C_0(\vx,\vnz)=0$,}
		&\equiv \textstyle C_1(\vx,\vnz)\cdot f \mod \vx^2-\vx
		\;.
	\end{align*}
	Thus, for every $\vx\in\bits^n$ we have that $C_1(\vx,\vnz)=\nicefrac{1}{f(\vx)}$ so that by the assumed functional lower bound $C_1(\vx,\vnz)\notin \cD$, yielding the desired contradiction to the above $C_1(\vx,\vnz)\in\cD$.

	Now suppose that $\cC$ is a set of multilinear polynomials.  Any $\cC$-IPS refutation $C(\vx,y,\vz)$ of $f(\vx),\vx^2-\vx$ thus must have $\deg_y C\le 1$ as $C$ is multilinear, so that $C$ is actually a $\cC$-\lbIPS refutation, thus the above lower bound applies.
\end{proof}

\subsection{Degree of a Polynomial}\label{sec:lbs-fn:deg}

We now turn to obtaining functional lower bounds, and deriving the corresponding linear-IPS lower bounds. We begin with a particularly weak form of algebraic complexity, the degree of a polynomial.  While it is trivial to obtain such bounds in some cases (as any polynomial that agrees with the AND function on the boolean cube $\bits^n$ must have degree $\ge n$), for our applications to proof complexity we will need such degree bounds for functions defined by $\hat{f}(\vx)=\nicefrac{1}{f(\vx)}$ for simple polynomials $f(\vx)$.

We show that any multilinear polynomial agreeing with $\nicefrac{1}{f(\vx)}$, where $f(\vx)$ is the subset-sum axiom $\sum_i x_i-\beta$, must have the maximal degree $n$.  We note that a degree lower bound of $\ceil{\nicefrac{n}{2}}+1$ was established by Impagliazzo, \Pudlak, and Sgall~\cite{IPS99} (\autoref{thm:IPS99}).  They actually established this degree bound\,\footnote{The degree lower bound of Impagliazzo, \Pudlak, and Sgall~\cite{IPS99} (\autoref{thm:IPS99}) actually holds for the (dynamic) polynomial calculus proof system (see section \autoref{sec:alg-proofs}), while we only consider the (static) Nullstellensatz proof system here.  Note that for polynomial calculus Impagliazzo, \Pudlak, and Sgall~\cite{IPS99} also showed a matching upper bound of $\ceil{\nicefrac{n}{2}}+1$ for $\vaa=\vno$.}  when $f(\vx)=\sum_i \alpha_i x_i-\beta$ for \emph{any} $\vaa$, while we only consider $\vaa=\vno$ here.  However, we need the tight bound of $n$ here as it will be used crucially in the proof of \autoref{res:lbs-fn:dim-eval}. 

\begin{proposition}\label{res:subsetsum:deg}
	Let $n\ge 1$ and $\F$ be a field with $\chara(\F)>n$.  Suppose that $\beta\in \F\setminus\{0,\ldots,n\}$. Let $f\in\F[x_1,\ldots,x_n]$ be a multilinear polynomial such that
	\[
		f(\vx)\left(\sum_i x_i-\beta\right)=1 \mod \vx^2-\vx
		\;.
	\]
	Then $\deg f=n$.
\end{proposition}
\begin{proof}
	\uline{$\le n$:} This is clear as $f$ is multilinear.

	\uline{$\ge n$:} Begin by observing that as $\beta\notin\{0,\ldots,n\}$, this implies that $\sum_i x_i-\beta$ is never zero on the boolean cube, so that the above functional equation implies that for $\vx\in\bits^n$ the expression
	\begin{align*}
		f(\vx)= \frac{1}{\sum_i x_i-\beta}
		\;,
	\end{align*}
	is well defined.

	Now observe that this implies that $f$ is a symmetric polynomial.  That is, define the multilinear polynomial $g$ by symmetrizing $f$,
	\begin{align*}
		g(x_1,\ldots,x_n)\eqdef \frac{1}{n!}\sum_{\sigma\in\Sn_n} f(x_{\sigma(1)},\ldots,x_{\sigma(n)})
		\;,
	\end{align*}
	where $\Sn_n$ is the symmetric group on $n$ symbols.  Then we see that $f$ and $g$ agree on $\vx\in\bits^n$, as
	\begin{align*}
		g(\vx)
		&=\frac{1}{n!}\sum_{\sigma\in\Sn_n} f(x_{\sigma(1)},\ldots,x_{\sigma(n)})\\
		&=\frac{1}{n!}\sum_{\sigma\in\Sn_n} \frac{1}{\sum_i x_{\sigma(i)}-\beta}
		=\frac{1}{n!}\sum_{\sigma\in\Sn_n} \frac{1}{\sum_i x_i-\beta}\\
		&=\frac{1}{n!}\cdot n!\cdot \frac{1}{\sum_i x_i-\beta}
		=\frac{1}{\sum_i x_i-\beta}=f(\vx)
		\;.
	\end{align*}
	It follows then that $g=f$ as polynomials, since they are multilinear and agree on the boolean cube (\autoref{fact:multilinearization}).  As $g$ is clearly symmetric, so is $f$. Thus $f$ can be expressed as $f=\sum_{k=0}^d \gamma_k S_{n,k}(\vx)$, where $d\eqdef\deg f$, $S_{n,k}\eqdef \sum_{S\subseteq\binom{[n]}{k}}\prod_{i\in S} x_i$ is the $k$-th elementary symmetric polynomial, and $\gamma_k\in\F$ are scalars with $\gamma_d\ne 0$.

	Now observe that for $k<n$, we can understand the action of multiplying $S_{n,k}$ by $\sum_i x_i-\beta$.
	\begin{align*}
		\textstyle(\sum_i x_i-\beta) S_{n,k}(\vx)
		&=\sum_{S\in\binom{[n]}{k}} \textstyle(\sum_i x_i-\beta) \prod_{j\in S} x_j\\
		&=\sum_{S\in\binom{[n]}{k}} \textstyle\left(\sum_{i\notin S} x_i \prod_{j\in S} x_j+\sum_{i\in S} x_i \prod_{j\in S} x_j-\beta \prod_{j\in S} x_j\right)\\
		&=\sum_{S\in\binom{[n]}{k}} \left(\sum_{\substack{|T|=k+1\\T\supseteq S}} \prod_{j\in T}x_j +(k-\beta) \prod_{j\in S} x_j\right) \mod \vx^2-\vx\\
		&=(k+1)S_{n,k+1}+(k-\beta)S_{n,k}\mod \vx^2-\vx
		\;.
	\end{align*}
	Note that we used that each subset of $[n]$ of size $k+1$ contains exactly $k+1$ subsets of size $k$.

	Putting the above together, suppose for contradiction that $d<n$.  Then,
	\begin{align*}
		1
		&=f(\vx)\left(\sum_i x_i-\beta\right) \mod \vx^2-\vx\\
		&=\left(\sum_{k=0}^d \gamma_k S_{n,k}\right)\left(\sum_i x_i-\beta\right) \mod \vx^2-\vx\\
		&=\left(\sum_{k=0}^d \gamma_k \Big((k+1)S_{n,k+1}+(k-\beta)S_{n,k}\Big)\right) \mod \vx^2-\vx\\
		&=\gamma_d(d+1)S_{n,d+1} + (\text{degree $\le d$}) \mod \vx^2-\vx
	\end{align*}
	However, as $\gamma_d\ne 0$, $d+1\le n$ (so that $d+1\ne 0$ in $\F$ and $S_{n,d+1}$ is defined) this shows that $1$ (a multilinear degree 0 polynomial) equals $\gamma_d(d+1)S_{n,d+1} + (\text{degree $\le d$})$ (a multilinear degree $d+1$ polynomial) modulo $\vx^2-\vx$, which is a contradiction to the uniqueness of representation of multilinear polynomials modulo $\vx^2-\vx$.  Thus, we must have $d=n$.
\end{proof}

To paraphrase the above argument, it shows that for multilinear $f$ of $\deg f<n$ with $\ml(f(\vx)\cdot (\sum_i x_i-\beta))=1$ it holds that $\deg \ml(f(\vx)\cdot (\sum_i x_i-\beta))=\deg f+1$. This contradicts the fact that $\deg 1=0$, so that $\deg f=n$.  It is tempting to attempt to argue this claim without using that $\ml(f(\vx)\cdot (\sum_i x_i-\beta))=1$ in some way. That is, one could hope to argue that $\deg(\ml(f(\vx)\cdot(\sum_i x_i-\beta)))= \deg f+1$ directly.  Unfortunately this is false, as seen by the example $\ml((x+y)(x-y))=\ml(x^2-y^2)=x-y$.  However, one can make this approach work to obtain a degree lower bound of $\ceil{\nicefrac{n}{2}}+1$, as shown by Impagliazzo, \Pudlak, and Sgall~\cite{IPS99}.

Putting the above together with the fact that multilinearization is degree non-increasing (\autoref{fact:multilinearization}) we obtain that any polynomial agreeing with $\frac{1}{\sum_i x_i-\beta}$ on the boolean cube must be of degree $\ge n$.

\begin{corollary}\label{res:subsetsum:deg:ge}
	Let $n\ge 1$ and $\F$ be a field with $\chara(\F)>n$.  Suppose that $\beta\in \F\setminus\{0,\ldots,n\}$. Let $f\in\F[x_1,\ldots,x_n]$ be a polynomial such that
	\[
		f(\vx)\left(\sum_i x_i-\beta\right)=1 \mod \vx^2-\vx
		\;.
	\]
	Then $\deg f\ge n$.
\end{corollary}
\begin{proof}
	Multilinearizing (\autoref{fact:multilinearization}) we see that $1=\ml\big(f(\vx)\cdot\left(\sum_i x_i-\beta\right)\big)=\ml\big(\ml(f)\cdot\left(\sum_i x_i-\beta\right)\big)$, so that $\ml(f)\cdot \left(\sum_i x_i-\beta\right)=1\mod \vx^2-\vx$.  Thus $\deg f\ge \deg \ml(f)$ (\autoref{fact:multilinearization}) and $\deg \ml(f)=n$ by the above \autoref{res:subsetsum:deg}, yielding the claim.
\end{proof}

The above proof shows that the unique multilinear polynomial $f$ agreeing with $\nicefrac{1}{\left(\sum_i x_i-\beta\right)}$ on the hypercube has degree $n$, but does so without actually specifying the coefficients of $f$.  In \autoref{res:subsetsum:multlin} we compute the coefficients of this polynomial, giving an alternate proof that it has degree $n$ (\autoref{res:subsetsum:multlin:deg-sparse}).  In particular, this computation yields a small algebraic circuit for $f$, expressing it as an explicit linear combination of elementary symmetric polynomials (which have small algebraic circuits).

\subsection{Sparse polynomials}\label{sec:h-ips:lbs:sparse}

We now use the above functional lower bounds for degree, along with random restrictions, to obtain functional lower bounds for sparsity.  We then apply this to obtain exponential lower bounds for sparse-\lIPS refutations of the subset-sum axiom.  Recall that sparse-\lIPS is equivalent to the Nullstellensatz proof system when we measure the size of the proof in terms of the number of monomials.  While we provide the proof here for completeness, we note that this result has already been obtained by Impagliazzo-\Pudlak-Sgall~\cite{IPS99}, who also gave such a lower bound for the stronger polynomial calculus proof system.

We first recall the random restrictions lemma.  This lemma shows that by randomly setting half of the variables to zero, sparse polynomials become sums of monomials involving few variables, which after multilinearization is a low-degree polynomial.

\begin{lemma}\label{res:random-restriction}
	Let $f\in\F[x_1,\ldots,x_n]$ be an $s$-sparse polynomial.  Let $\rho:\F[\vx]\to\F[\vx]$ be the homomorphism induced by randomly and independently setting each variable $x_i$ to $0$ with probability $\nicefrac{1}{2}$ and leaving $x_i$ intact with probability $\nicefrac{1}{2}$.  Then with probability $\ge \nicefrac{1}{2}$, each monomial in $\rho(f(\vx))$ involves $\le \lg s+1$ variables.  Thus, with probability $\ge \nicefrac{1}{2}$, $\deg\ml(\rho(f))\le \lg s+1$.
\end{lemma}
\begin{proof}
	Consider a monomial $\vx^\va$ involving $\ge t$ variables, $t\in\R$.  Then the probability that $\rho(\vx^\va)$ is nonzero is at most $2^{-t}$.  Now consider $f(\vx)=\sum_{j=1}^s \alpha_j \vx^{\va_j}$. By a union bound, the probability that any monomial $\vx^{\va_j}$ involving at least $t$ variables survives the random restriction is at most $s2^{-t}$. For $t=\lg s+1$ this is at most $\nicefrac{1}{2}$.  The claim about the multilinearization of $\rho(f(\vx))$ follows by observing that for a monomial $\vx^\va$ involving $\le \lg s+1$ variables it must be that $\deg\ml(\rho(\vx^\va))\le \lg s+1$ (\autoref{fact:multilinearization}).
\end{proof}

We now give our functional lower bound for sparsity.  This follows from taking any refutation of the subset-sum axiom and applying a random restriction.  The subset-sum axiom will be relatively unchanged, but any sparse polynomial will become (after multilinearization) low-degree, to which our degree lower bounds (\autoref{sec:lbs-fn:deg}) can then be applied.

\begin{proposition}
	Let $n\ge 8$ and $\F$ be a field with $\chara(\F)>n$.  Suppose that $\beta\in \F\setminus\{0,\ldots,n\}$.  Let $f\in\F[x_1,\ldots,x_n]$ be a polynomial such that
	\[
		f(\vx)=\frac{1}{\sum_i x_i-\beta}
		\;,
	\]
	for $\vx\in\bits^n$. Then $f$ requires $\ge 2^{\nicefrac{n}{4}-1}$ monomials.
\end{proposition}
\begin{proof}
	Suppose that $f$ is $s$-sparse so that $f(\vx)=\sum_{j=1}^s \alpha_j\vx^{\va_j}$. Take a random restriction $\rho$ as in \autoref{res:random-restriction}, so that with probability at least $\nicefrac{1}{2}$ we have that $\deg \ml(\rho(f))\le \lg s+1$.  By the Chernoff bound,\footnote{For independent $[0,1]$-valued random variables $\rX_1,\ldots,\rX_n$, $\Pr\left[\sum_i \rX_i -\sum_i\E[\rX_i]\le-\eps n\right]\le \e^{-2\eps^2n}$.} we see that $\rho$ keeps alive at least $\nicefrac{n}{4}$ variables with probability at least $1-\e^{-2\cdot (\nicefrac{1}{4})^2\cdot n}$, which is $\ge 1-\e^{-1}$ for $n\ge 8$.
	Thus, by a union bound the probability that $\rho$ fails to have either that $\deg\ml(\rho(f))\le \lg s+1$ or that it keeps at least $\nicefrac{n}{4}$ variables alive is at most $\nicefrac{1}{2}+\e^{-1}<1$.  Thus a $\rho$ exists obeying both properties.

	Thus, the functional equation for $f$ implies that
	\[
		f(\vx)\left(\sum_i x_i-\beta\right)=1+\sum_i h_i(\vx) (x_i^2-x_i)
		\;,
	\]
	for some $h_i\in\F[\vx]$.  Applying the random restriction and multilinearization to both sizes of this equation, we obtain that
	\[
		\ml(\rho(f))\cdot \left(\sum_{\rho(x_i)\ne 0} x_i-\beta\right)\equiv 1 \mod \{x_i^2-x_i\}_{\rho(x_i)\ne 0}
		\;.
	\]
	Thus, by appealing to the degree lower bound for this functional equation (\autoref{res:subsetsum:deg}) we obtain that $\lg s+1\ge \deg\ml(\rho(f))$ is at least the number of variables which is $\ge \nicefrac{n}{4}$, so that $s\ge 2^{\nicefrac{n}{4}-1}$ as desired.
\end{proof}

We remark that one can actually improve the sparsity lower bound to the optimal ``$\ge 2^n$'' by computing the sparsity of the unique multilinear polynomial satisfying the above functional equation (\autoref{res:subsetsum:multlin:deg-sparse}). We now apply these functional lower bounds to obtain lower bounds for sparse-\lIPS refutations of $\sum_i x_i-\beta,\vx^2-\vx$ via our reduction (\autoref{res:lbs-fn_lbs-ips:lIPS}).

\begin{corollarywp}
	Let $n\ge 1$ and $\F$ be a field with $\chara(\F)>n$.  Suppose that $\beta\in \F\setminus\{0,\ldots,n\}$.  Then $\sum_{i=1}^n x_i-\beta,\vx^2-\vx$ is unsatisfiable and any sparse-\lIPS refutation requires size $\exp(\Omega(n))$.
\end{corollarywp}

\subsection{Coefficient Dimension in a Fixed Partition}\label{sec:evaluation}

We now seek to prove functional circuit lower bounds for more powerful models of computation such as roABPs and multilinear formulas.  As recalled in \autoref{sec:background}, the coefficient dimension complexity measure can give lower bounds for such models.  However, by definition it is a \emph{syntactic} measure as it speaks about the coefficients of a polynomial.  Unfortunately, knowing that a polynomial $f\in\F[\vx]$ agrees with a function $\hat{f}:\bits^n\to\F$ on the boolean cube $\bits^n$ does not in general give enough information to determine its coefficients.  In contrast, the \emph{evaluation} dimension measure is concerned with evaluations of a polynomial (which is functional).  Obtaining lower bounds for evaluation dimension, and leveraging the fact that the evaluation dimension lower bounds coefficient dimension (\autoref{res:evals_eq-coeffs}) we can obtain the desired lower bounds for this complexity measure.

We now proceed to the lower bound.  It will follow from the degree lower bound for the subset-sum axiom (\autoref{res:subsetsum:deg:ge}).  That is, this degree bound shows that if $f(\vz)\cdot(\sum_i z_i-\beta)\equiv 1\mod \vz^2-\vz$ then $f$ must have degree $\ge n$.  We can then ``lift'' this lower bound by the use of a gadget, in particular by replacing $\vz\leftarrow \vx\circ \vy$, where `$\circ$' is the Hadamard (entry-wise) product.  Because the degree of $f$ is maximal, this gadget forces $\vx$ and $\vy$ to maximally ``interact'', and hence the evaluation dimension is large in the $\vx$ versus $\vy$ partition.

\begin{proposition}\label{res:lbs-fn:dim-eval}
	Let $n\ge 1$ and $\F$ be a field with $\chara(\F)>n$.  Suppose that $\beta\in\F\setminus\{0,\ldots,n\}$.  Let $f\in\F[x_1,\ldots,x_n,y_1,\ldots,y_n]$ be a polynomial such that
	\[
		f(\vx,\vy)=\frac{1}{\sum_i x_iy_i-\beta}
		\;,
	\]
	for $\vx,\vy\in\bits^n$. Then $\dim\coeffs{\vx|\vy} f\ge 2^n$.
\end{proposition}
\begin{proof}
	By lower bounding coefficient dimension by the evaluation dimension over the boolean cube (\autoref{res:evals_eq-coeffs}), 
	\begin{align*}
		\dim\coeffs{\vx|\vy} f
		&\ge \dim\evals{\vx|\vy,\bits} f\\
		&=\dim \{f(\vx,\ind{S}) : S\subseteq[n]\}\\
		&\ge\dim \{\ml(f(\vx,\ind{S})) : S\subseteq[n]\}
		\;,
	\end{align*}
	where $\ind{S}\in\bits^n$ is the indicator vector for a set $S$, and $\ml$ is the multilinearization operator.  Note that we used that multilinearization is linear (\autoref{fact:multilinearization}) and that dimension is non-increasing under linear maps. Now note that for $\vx\in\bits^n$,
	\[
		f(\vx,\ind{S})=\frac{1}{\sum_{i\in S} x_i-\beta}
		\;,
	\]
	It follows then $\ml(f(\vx,\ind{S}))$ is a multilinear polynomial only depending on $\vx|_S$ (\autoref{fact:multilinearization}), and by its functional behavior it follows from \autoref{res:subsetsum:deg} that $\deg \ml(f(\vx,\ind{S}))=|S|$.  As $\ml(f(\vx,\ind{S}))$ is multilinear it thus follows that the leading monomial of $\ml(f(\vx,\ind{S}))$ is $\prod_{i\in S} x_i$, which is distinct for each distinct $S$.  This is also readily seen from the explicit description of $\ml(f(\vx,\ind{S}))$ given by \autoref{res:subsetsum:multlin}.  Thus, we can lower bound the dimension of this space by the number of leading monomials (\autoref{res:dim-eq-num-TM-spn}),
	\begin{align*}
		\dim\coeffs{\vx|\vy} f
		&\ge\dim \{\ml(f(\vx,\ind{S})) : S\subseteq[n]\}\\
		&\ge\left|\LM \Big(\{\ml(f(\vx,\ind{S})) : S\subseteq[n]\}\Big)\right|\\
		&=\left|\left\{\prod_{i\in S} x_i : S\subseteq[n]\right\}\right|\\
		&=2^n
		\;.
		\qedhere
	\end{align*}
\end{proof}

Note that in the above proof we crucially leveraged that the degree bound of \autoref{res:subsetsum:deg} is \emph{exactly} $n$, not just $\Omega(n)$.  This exact bound allows us to determine the leading monomials of these polynomials, which seems not to follow from degree lower bounds of $\Omega(n)$.

As coefficient dimension lower bounds roABP-width (\autoref{res:roABP-width_eq_dim-coeffs}) and depth-3 powering formulas can be computed by roABPs in any variable order (\autoref{res:sumpowsum:roABP}), we obtain as a corollary our functional lower bound for these models.

\begin{corollarywp}\label{res:lbs-fn:roABP}
	Let $n\ge 1$ and $\F$ be a field with $\chara(\F)>n$.  Suppose that $\beta\in\F\setminus\{0,\ldots,n\}$.  Let $f\in\F[x_1,\ldots,x_n,y_1,\ldots,y_n]$ be a polynomial such that
	\[
		f(\vx,\vy)=\frac{1}{\sum_i x_iy_i-\beta}
		\;,
	\]
	for $\vx,\vy\in\bits^n$.  Then $f$ requires width $\ge 2^n$ to be computed as an roABP in any variable order where $\vx$ precedes $\vy$.  In particular, $f$ requires $\exp(\Omega(n))$-size as a depth-$3$ powering formula.
\end{corollarywp}

We now conclude with a lower bound for linear-IPS over roABPs in certain variable orders, and thus also for depth-$3$ powering formulas, by appealing to our reduction to functional lower bounds (\autoref{res:lbs-fn_lbs-ips:lIPS}).

\begin{corollary}\label{res:lbs-fn:lbs-ips:fixed-order}
	Let $n\ge 1$ and $\F$ be a field with $\chara(\F)>n$.  Suppose that $\beta\in \F\setminus\{0,\ldots,n\}$.  Then $\sum_{i=1}^n x_iy_i-\beta,\vx^2-\vx,\vy^2-\vy$ is unsatisfiable and any roABP-\lIPS refutation, where the roABP reads $\vx$ before $\vy$, requires width $\ge \exp(\Omega(n))$.  In particular, any $\sumpowsum$-\lIPS refutation requires size $\ge\exp(\Omega(n))$.
\end{corollary}
\begin{proof}
	That this system is unsatisfiable is clear from construction.  The proof then follows from applying our functional lower bound (\autoref{res:lbs-fn:roABP}) to our reduction strategy (\autoref{res:lbs-fn_lbs-ips:lIPS}), where we use that partial evaluations of small roABPs yield small roABPs in the induced variable order (\autoref{fact:roABP:closure}), and that depth-3 powering formulas are a subclass of roABPs (in any order) (\autoref{res:sumpowsum:roABP}).
\end{proof}

The above result shows an roABP-\lIPS lower bound for variable orders where $\vx$ precedes $\vy$, and we complement this by giving an upper bound showing there \emph{are} small roABP-\lIPS upper bounds for variable orders where $\vx$ and $\vy$ are tightly interleaved.  This is achieved by taking the roABP-\lIPS upper bound of \autoref{res:ips-ubs:subset:roABP} for $\sum_i z_i-\beta,\vz^2-\vz$ under the substitution $z_i\leftarrow x_iy_i$, and observing that such substitutions preserve roABP width in the $x_1<y_1<\cdots<x_n<y_n$ order (\autoref{fact:roABP:closure}). In particular, as $\sumpowsum$ formulas are small roABPs in \emph{every} variable order, this allows us to achieve an exponential separation between $\sumpowsum$-\lIPS and roABP-\lIPS.

\begin{corollary}\label{res:lbs-fn:lbs-ips:order-sep}
	Let $n\ge 1$ and $\F$ be a field with $\chara(\F)>n$.  Suppose that $\beta\in \F\setminus\{0,\ldots,n\}$.  Then $\sum_{i=1}^n x_iy_i-\beta,\vx^2-\vx,\vy^2-\vy$ is unsatisfiable, has a $\poly(n)$-explicit $\poly(n)$-size roABP-\lIPS refutation in the variable order $x_1<y_1<\cdots<x_n<y_n$, and every $\sumpowsum$-\lIPS refutation requires size $\ge\exp(\Omega(n))$.
\end{corollary}
\begin{proof}
	\autoref{res:lbs-fn:lbs-ips:fixed-order} showed that this system is unsatisfiable and has the desired $\sumpowsum$-\lIPS lower bound, so that it remains to prove the roABP upper bound.

	By \autoref{res:ips-ubs:subset} the unique multilinear polynomial $f\in\F[\vz]$ such that $f(\vz)\cdot (\sum_{i=1}^nz_i-\beta)\equiv 1\mod \vz^2-\vz$ has a multilinear $\poly(n)$-size roABP in the variable order $z_1<\cdots<z_n$. Applying the variable substitution $z_i\leftarrow x_iy_i$, it follows that $f'(\vx,\vy)\eqdef f(x_1y_1,\ldots,x_ny_n)$ obeys $f'\cdot (\sum_{i=1}^nx_iy_i-\beta)\equiv 1 \mod \vx^2-\vx,\vy^2-\vy$ (as $z_i^2-z_i\equiv 0 \mod \vx^2-\vx,\vy^2-\vy$ under the substitution $z_i\leftarrow x_iy_i$) and that $f'$ is computable by a $\poly(n)$-size roABP in the variable order $x_1<y_1<\cdots<x_n<y_n$ (\autoref{fact:roABP:closure}, using that $f$ has individual degree 1).  Appealing to the efficient multilinearization of roABPs (\autoref{res:multilin:roABP-lIPS}) completes the claim as $\sum_i x_iy_i -\beta$ is computable by a $\poly(n)$-size roABP (in any order).
\end{proof}

\subsection{Coefficient Dimension in any Variable Partition}\label{sec:lbs-fn:every-partition}

The previous section gave functional lower bounds for coefficient dimension, and thus roABP width, in the $\vx|\vy$ variable partition. However, this lower bound fails for other variable orderings where $\vx$ and $\vy$ are interleaved because of corresponding upper bounds (\autoref{res:lbs-fn:lbs-ips:order-sep}). In this section we extend the lower bound to \emph{any} variable ordering by using suitable auxiliary variables to plant the previous lower bound into any partition we desire by suitably evaluating the auxiliary variables.

We begin by developing some preliminaries for how coefficient dimension works in the presence of auxiliary indicator variables.  That is, consider a polynomial $f(\vx,\vy,\vz)$ where we wish to study the coefficient dimension of $f$ in the $\vx|\vy$ partition. We can view this polynomial as lying in $\F[\vz][\vx,\vy]$ so that its coefficients are polynomials in $\vz$ and one studies the dimension of the coefficient space in the field of rational functions $\F(\vz)$.  Alternatively one can evaluate $\vz$ at some point $\vz\leftarrow\vaa$ so that $f(\vx,\vy,\vaa)\in\F[\vx,\vy]$ and study its coefficient dimension over $\F$.  The following straightforward lemma shows the first dimension over $\F(\vz)$ is lower-bounded by the second dimension over $\F$.

\begin{lemma}\label{res:coeff-dim:fraction-field}
	Let $f\in\F[\vx,\vy,\vz]$. Let $f_\vz$ denote $f$ as a polynomial in $\F[\vz][\vx,\vy]$, so that for any $\vaa\in\F^{|\vz|}$ we have that $f_\vaa\in\F[\vx,\vy]$.  Then for any such $\vaa$,
	\[
		\dim_{\F(\vz)}\coeffs{\vx|\vy} f_\vz(\vx,\vy) \ge \dim_\F \coeffs{\vx|\vy}f_\vaa(\vx,\vy)
		\;.
	\]
\end{lemma}
\begin{proof}
	Let $f(\vx,\vy,\vz)$ be written in $\F[\vx,\vy,\vz]$ as $f=\sum_{\va,\vb} f_{\va,\vb} (\vz)\vx^\va\vy^\vb$. By \autoref{res:y-dim_eq-x-dim} we see that $\dim_{\F(\vz)}\coeffs{\vx|\vy} f_\vz(\vx,\vy)$ is equal to the rank (over $\F(\vz)$) of the coefficient matrix $C_{f_\vz}$, so that its entries $(C_{f_\vz})_{\va,\vb}=f_{\va,\vb}(\vz)$ are in $\F[\vz]$.  Similarly, $\dim_{\F}\coeffs{\vx|\vy} f_\vaa(\vx,\vy)$ is equal to the rank (over $\F$) of the coefficient matrix $C_{f_\vaa}$, so that as $f(\vx,\vy,\vaa)= \sum_{\va,\vb} f_{\va,\vb}(\vaa)\vx^\va\vy^\vb$ we have that $(C_{f_\vaa})_{\va,\vb}=f_{\va,\vb}(\vaa)$, which is in $\F$.  Thus, it follows that $C_{f_\vz}|_{\vz\leftarrow\vaa}=C_{f_\vaa}$.  
	
	The claim then follows by noting that for a matrix $M(\vw)\in\F[\vw]^{r\times r}$ it holds that $\rank_{\F(\vw)} M(\vw)\ge \rank_{\F} M(\vbb)$ for any $\vbb\in\F^{|\vw|}$. This follows as the rank of $M(\vw)$ is equal to the maximum size of a minor with a non-vanishing determinant.  As such determinants are polynomials in $\vw$, they can only further vanish when $\vw\leftarrow\vbb$.  
\end{proof}

We now use auxiliary variables to embed the coefficient dimension lower bound from \autoref{res:lbs-fn:dim-eval} into any variable order.  We do this by viewing the polynomial $\sum_i u_iv_i-\beta$ as using a matching between variables in $\vu$ and $\vv$.  We then wish to embed this matching graph-theoretically into a complete graph, where nodes are labelled with the variables $\vx$.  Any equipartition of this graph will induce many edges across this cut, and we can drop edges to find a large matching between the $\vx$ variables which we then identify as instance of $\sum_i u_iv_i-\beta$.  We introduce one new auxiliary variable $z_{i,j}$ per edge which, upon setting it to 0 or 1, allows us to have this edge (respectively) dropped from or kept in the desired matching.  This leads to the new (symmetrized) equation $\sum_{i<j} z_{i,j} x_ix_j-\beta$, for which we now give the desired lower bound.

\begin{proposition}\label{res:lbs-fn:any-order:coeff-dim}
	Let $n\ge 1$ and $\F$ be a field with $\chara(\F)>\binom{2n}{2}$.  Suppose that $\beta\in\F\setminus\{0,\ldots,\binom{2n}{2}\}$.  Let $f\in\F[x_1,\ldots,x_{2n},z_1,\ldots,z_{\binom{2n}{2}}]$ be a polynomial such that
	\[
		f(\vx,\vz)=\frac{1}{\sum_{i<j} z_{i,j}x_ix_j-\beta}
		\;,
	\]
	for $\vx\in\bits^{2n}$, $\vz\in\bits^{\binom{2n}{2}}$. Let $f_\vz$ denote $f$ as a polynomial in $\F[\vz][\vx]$.  Then for any partition $\vx=(\vu,\vv)$ with $|\vu|=|\vv|=n$,
	\[
		\dim_{\F(\vz)}\coeffs{\vu|\vv} f_\vz \ge 2^n
		\;.
	\]
\end{proposition}
\begin{proof}
	We wish to embed $\sum_i u_iv_i-\beta$ in this instance via a restriction of $\vz$. Define the $\vz$-evaluation $\vaa\in\bits^{\binom{2n}{2}}$ to restrict $f$ to sum over those $x_ix_j$ in the natural matching between $\vu$ an $\vv$, so that
	\[
		\alpha_{i.j}
		=\begin{cases}
			1	&	x_i=u_k, x_j=v_k\\
			0	&	\text{else}
		\end{cases}
		\;.
	\]
	It follows then that $f(\vu,\vv,\vaa)=\frac{1}{\sum_{k=1}^n u_kv_k-\beta}$ for $\vu,\vv\in\bits^{n}$. Thus, by appealing to our lower bound for a fixed partition (\autoref{res:lbs-fn:dim-eval}) and the relation between the coefficient dimension in $f_\vz$ versus $f_\vaa$ (\autoref{res:coeff-dim:fraction-field}),
	\begin{align*}
		\dim_{\F(\vz)}\coeffs{\vu|\vv} f_\vz(\vu,\vv)
		&\ge \dim_\F \coeffs{\vu|\vv}f_\vaa(\vu,\vv)\\
		&\ge 2^n
		\;.
		\qedhere
	\end{align*}
\end{proof}

We remark that this lower bound is only $\exp(\Omega(\sqrt{m}))$ where $m=2n+\binom{2n}{2}$ is the number of total variables, while one could hope for an $\exp(\Omega(m))$ lower bound as $2^m$ is the trivial upper bound for multilinear polynomials.  One can achieve such a lower bound by replacing the above auxiliary variable scheme (which corresponds to a complete graph) with one derived from a constant-degree expander graph. That is because in such graphs any large partition of the vertices induces a large matching across that partition, where one can then embed the fixed-partition lower bounds of the previous section (\autoref{sec:evaluation}).  However, we omit the details as this would not qualitatively change the results.

We now obtain our desired functional lower bounds for roABPs and multilinear formulas.

\begin{corollary}\label{res:lbs-fn:any-order}
	Let $n\ge 1$ and $\F$ be a field with $\chara(\F)>\binom{2n}{2}$.  Suppose that $\beta\in\F\setminus\{0,\ldots,\binom{2n}{2}\}$.  Let $f\in\F[x_1,\ldots,x_{2n},z_1,\ldots,z_{\binom{2n}{2}}]$ be a polynomial such that
	\[
		f(\vx,\vz)=\frac{1}{\sum_{i<j} z_{i,j}x_ix_j-\beta}
		\;,
	\]
	for $\vx\in\bits^{2n}$, $\vz\in\bits^{\binom{2n}{2}}$. Then $f$ requires width $\ge 2^n$ to be computed by an roABP in any variable order.  Also, $f$ requires $n^{\Omega(\log n)}$-size to be computed as a multilinear formula.  For $d=o(\nicefrac{\log n}{\log\log n})$, $f$ requires $n^{\Omega((\nicefrac{n}{\log n})^{\nicefrac{1}{d}}/d^2)}$-size multilinear formulas of product-depth-$d$.
\end{corollary}
\begin{proof}
	\uline{roABPs:}  Suppose that $f(\vx,\vz)$ is computable by a width-$r$ roABP in some variable order.  By pushing the $\vz$ variables into the fraction field, it follows that $f_\vz$ ($f$ as a polynomial in $\F[\vz][\vx]$) is also computable by a width-$r$ roABP over $\F(\vz)$ in the induced variable order on $\vx$ (\autoref{fact:roABP:closure}).  By splitting $\vx$ in half along its variable order one obtains the lower bound by combining the coefficient dimension lower bound of \autoref{res:lbs-fn:any-order:coeff-dim} with its relation to roABPs (\autoref{res:roABP-width_eq_dim-coeffs}).

	\uline{multilinear formulas:} This follows immediately from our coefficient dimension lower bound (\autoref{res:lbs-fn:any-order:coeff-dim}) and the Raz~\cite{Raz09} and Raz-Yehudayoff~\cite{RazYehudayoff09} results (\autoref{thm:full-rank-lb}).
\end{proof}

As before, this immediately yields the desired roABP-\lIPS and multilinear-formula-IPS lower bounds.

\begin{corollary}\label{res:lbs-fn:lbs-ips:vary-order}
	Let $n\ge 1$ and $\F$ be a field with $\chara(\F)>\binom{2n}{2}$.  Suppose that $\beta\in\F\setminus\{0,\ldots,\binom{2n}{2}\}$.  Then $\sum_{i<j} z_{i,j}x_ix_j-\beta,\vx^2-\vx,\vz^2-\vz\in\F[x_1,\ldots,x_{2n},z_1,\ldots,z_{\binom{2n}{2}}]$ is unsatisfiable, and any roABP-\lIPS refutation (in any variable order) requires $\exp(\Omega(n))$-size. Further, any multilinear-formula-IPS refutation requires $n^{\Omega(\log n)}$-size, and any product-depth-$d$ multilinear-formula-IPS refutation requires $n^{\Omega((\nicefrac{n}{\log n})^{\nicefrac{1}{d}}/d^2)}$-size.
\end{corollary}
\begin{proof}
	The system is unsatisfiable as any setting of $\vx\in\bits^n$ yields a sum over at most $\binom{2n}{2}$ $z$-variables, which must be in $\{0,\ldots,\binom{2n}{2}\}$ which by hypothesis does not contain $\beta$.  
	
	The roABP-\lIPS lower bound follows immediately from the above functional lower bound (\autoref{res:lbs-fn:any-order}) along with our reduction (\autoref{res:lbs-fn_lbs-ips:lIPS}), just as in \autoref{res:lbs-fn:lbs-ips:fixed-order}.  

	The multilinear-formula-IPS lower bound also follows immediately from the above functional lower bound (\autoref{res:lbs-fn:any-order}) along with our reduction from IPS lower bounds to functional lower bounds for multilinear polynomials (\autoref{res:lbs-fn_lbs-ips:lbIPS}).  In particular, this application uses that multilinear formulas are closed under partial evaluations, and that taking the difference of two formulas will only double its size and does not change the product depth.
\end{proof}

\section{Lower Bounds for Multiples of Polynomials}\label{sec:lbs-mult}

In this section we consider the problem of finding explicit polynomials whose nonzero multiples are all hard. Such polynomials are natural to search for, as intuitively if $f$ is hard to compute then so should small modifications such as $x_1f^2+4f^3$. This intuition is buttressed by Kaltofen's~\cite{Kaltofen89} result that if a polynomial has a small algebraic circuit then so do all of its factors (up to some pathologies in small characteristic). Taken in a contrapositive, this says that if a polynomial $f$ requires super-polynomial size algebraic circuits, then so must all of its nonzero multiples.  Thus, for general circuits the question of lower bounds for multiples reduces to the standard lower bounds question.

Unfortunately, for many restricted classes of circuits where lower bounds are known (depth-3 powering formulas, sparse polynomials, roABPs) Kaltofen's~\cite{Kaltofen89} result produces circuits for the factors which do not fall into (possibly stronger) restricted classes of circuits where lower bounds are still known.\footnote{While some results (\cite{DvirSY09,Oliveira15}) can bound the depth of the factors in terms of the depth of the input circuit, there are only very weak lower bounds known for constant-depth algebraic circuits.}  As such, developing lower bounds for multiples against these restricted classes seems to require further work beyond the standard lower bound question.

We will begin by discussing the applications of this problem to the hardness versus randomness paradigm in algebraic complexity.  We then use existing derandomization results to show that multiples of the determinant are hard for certain restricted classes. However, this method is very rigidly tied to the determinant.  Thus, we also directly study existing lower bound techniques for restricted models of computation (depth-$3$ powering formulas, sparse polynomials, and roABPs) and extend these results to also apply to multiples.  We will show the applications of such polynomials to proof complexity in \autoref{sec:ips-mult}.

\subsection{Connections to Hardness versus Randomness and Factoring Circuits}\label{sec:lbs-mult:hard-v-rand}

To motivate the problem of finding polynomials with hard multiples, we begin by discussing the hardness versus randomness approach to derandomizing polynomial identity testing.  That is, Kabanets and Impagliazzo~\cite{KabanetsImpagliazzo04} extended the hardness versus randomness paradigm of Nisan and Wigderson~\cite{NisanWigderson94} to the algebraic setting, showing that sufficiently good algebraic circuit lower bounds for an explicit polynomial would qualitatively derandomize PIT\@.  While much of the construction is similar (using combinatorial designs, hybrid arguments, etc.) to the approach of Nisan and Wigderson~\cite{NisanWigderson94} for boolean derandomization, there is a key difference.  In the boolean setting, obtaining a hardness versus randomness connection requires converting \emph{worst-case} hardness (no small computation can compute the function everywhere) to \emph{average-case} hardness (no small computation can compute the function on most inputs).  Such a reduction (obtained by Impagliazzo and Wigderson~\cite{ImpagliazzoWigderson97}) can in fact be obtained using certain error-correcting codes based on multivariate polynomials (as shown by Sudan, Trevisan and Vadhan~\cite{SudanTV01}). 

Such a worst-case to average-case reduction is also needed in the algebraic setting, but as multivariate polynomials are one source of this reduction in the boolean regime, it is natural to expect it to be easier in the algebraic setting. Specifically, the notion of average-case hardness for a polynomial $f(\vx)$ used in Kabanets-Impagliazzo~\cite{KabanetsImpagliazzo04} is that for any $g(\vx,y)$ satisfying $g(\vx,f(\vx))=0$, it must be that $g$ then requires large algebraic circuits (by taking $g(\vx,y)\eqdef y-f(\vx)$ this implies $f$ itself requires large circuits).  This can be interpreted as average-case hardness because if such a $g$ existed with a small circuit, then for any value $\vaa$ we have that $g(\vaa,y)$ is a univariate polynomial that vanishes on $f(\vaa)$. By factoring this univariate (which can be done efficiently), we see that such $g$ give a small list (of size at most $\deg g$) of possible values for $f(\vaa)$. By picking a random element from this list, one can correctly compute $f(\vx)$ with noticeable probability, which by an averaging argument one can convert to a (non-uniform) deterministic procedure to compute $f(\vx)$ on most inputs (over any fixed finite set).  While this procedure (involving univariate factorization) is not an algebraic circuit, the above argument shows that the Kabanets-Impagliazzo~\cite{KabanetsImpagliazzo04} notion is a natural form of average case hardness.

To obtain this form of average-case hardness from worst-case hardness, Kabanets and Impagliazzo~\cite{KabanetsImpagliazzo04} used a result of Kaltofen~\cite{Kaltofen89}, who showed that (up to pathologies in low-characteristic fields), factors of small (general) circuits have small circuits.  As $g(\vx,f(\vx))=0$ iff $y-f(\vx)$ divides $g(\vx,y)$, it follows that if $g(\vx,y)$ has a small circuit then so does $y-f(\vx)$, and thus so does $f(\vx)$.  Taking the contrapositive, if $f$ requires large circuits (worst-case hardness) then any such $g(\vx,y)$ with $g(\vx,f(\vx))=0$ also requires large circuits (average-case hardness). Note that this says that \emph{any} worst-case hard polynomial is \emph{also} average-case hard.  In contrast, this is provably false for boolean functions, where such worst-case to average-case reductions thus necessarily modify the original function.

Unfortunately, Kaltofen's~\cite{Kaltofen89} factoring algorithm does not preserve structural restrictions (such as multilinearity, homogeneousness, low-depth, read-once-ness, etc.) of the original circuit, so that obtaining average-case hardness for restricted classes of circuits requires worst-case hardness for much stronger classes.  While follow-up work has reduced the complexity of the circuits resulting from Kaltofen's~\cite{Kaltofen89} algorithm (Dvir-Shpilka-Yehudayoff~\cite{DvirSY09} and Oliveira~\cite{Oliveira15} extended Kaltofen's~\cite{Kaltofen89} to roughly preserve the depth of the original computation) these works are limited to factoring polynomials of small individual degree and do not seem applicable to other types of computations such as roABPs. Indeed, it even remains an open question to show any non-trivial upper bounds on the complexity of the factors of sparse polynomials.  In fact, we actually have non-trivial \emph{lower} bounds. Specifically, von zur Gathen and Kaltofen~\cite{vzGathenKaltofen85} gave an explicit $s$-sparse polynomial (over any field) which has a factor with $s^{\Omega(\log s )}$ monomials, and Volkovich~\cite{Volkovich15} gave, for a prime $p$, an explicit $n$-variate $n$-sparse polynomial of degree-$p$ which in characteristic $p$ has a factor with $\binom{n+p-2}{n-1}$ monomials (an exponential separation for $p\ge\poly(n)$). We refer the reader to the survey of Forbes and Shpilka~\cite{ForbesShpilka15} for more on the challenges in factoring small algebraic circuits.

While showing the equivalence of worst-case and average-case hardness for restricted circuit classes seems difficult, to derandomize PIT via Kabanets-Impagliazzo~\cite{KabanetsImpagliazzo04} only requires a \emph{single} polynomial which is average case hard.  To facilitate obtaining such hard polynomials, we now record an easy lemma showing that polynomials with only hard multiples are average-case hard.

\begin{lemma}\label{res:lbs-mult_to_nonroot}
	Let $f(\vx)\in\F[\vx]$ and $g(\vx,y)\in\F[\vx,y]$ both be nonzero, where $g(\vx,0)\ne 0$ also. If $g(\vx,f(\vx))=0$ then $g(\vx,0)$ is a nonzero multiple of $f(\vx)$.
\end{lemma}
\begin{proof}
	Let $g(\vx,y)=\sum_i g_i(\vx)y^i$ and $g_0(\vx)\eqdef g(\vx,0)$.  That $g(\vx,f(\vx))=0$ implies that 
	\begin{align*}
		0
		=g(\vx,f(\vx))
		=\sum_i g_i(\vx)(f(\vx))^i
		=g_0(\vx)+\sum_{i\ge 1} g_i(\vx)(f(\vx))^i
	\end{align*}
	so that $g_0(\vx)=f(\vx)\cdot \left(-\sum_{i\ge 1} g_i(\vx)(f(\vx))^{i-1}\right)$ as desired.
\end{proof}

That is, saying that $f(\vx)$ is \emph{not} average-case hard means that $g(\vx,f(\vx))=0$ for a nonzero $g(\vx,y)$.  One can assume that $g(\vx,0)\ne 0$, as otherwise one can replace $g$ by $\nicefrac{g}{y^i}$ for some $i\le \deg g$, as this only mildly increases the size for most measures of circuit size (see for example \autoref{sec:h-ips:ips=h-ips}).  As then the complexity of $g(\vx,0)$ is bounded by that of $g(\vx,y)$ (for natural measures), the lemma shows then that $f$ has a nonzero multiple of low-complexity.  Taken contrapositively, if $f$ only has hard nonzero multiples then it is average-case hard in the sense needed for Kabanets-Impagliazzo~\cite{KabanetsImpagliazzo04}. This shows that lower bounds for multiples is essentially the lower bound needed for algebraic hardness versus randomness.\footnote{However, it is not an exact equivalence between lower bounds for multiples and average case hardness, as the converse to \autoref{res:lbs-mult_to_nonroot} is false, as seen by considering $g(x,y)\eqdef y-x(x+1)$, so that $x|g(x,0)$ but $g(x,x)\ne 0$.}

While in the below sections we are able to give explicit polynomials with hard multiples for various restricted classes of algebraic circuits, some of these classes (such as sparse polynomials and roABPs) still do not have the required closure properties to use Kabanets-Impagliazzo~\cite{KabanetsImpagliazzo04} to obtain deterministic PIT algorithms.  Even for classes with the needed closure properties (such as $\sumpowc$ formulas, where the hard polynomial is the monomial), the resulting PIT algorithms are only worse than existing results (which for $\sumpowc$ formulas is the result of Forbes~\cite{Forbes15}).  However, it seems likely that future results establishing polynomials with hard multiples would imply new PIT algorithms.

\subsection{Lower Bounds for Multiples via PIT}\label{sec:lb-via-pit}

This above discussion shows that obtaining lower bounds for multiples is sufficient for instantiating the hardness versus randomness paradigm.  We now observe the converse, showing that one can obtain such polynomials with hard multiples via derandomizing (black-box) PIT, or equivalently, producing generators with small seed-length.  That is, Heintz-Schnorr~\cite{HeintzSchnorr80} and Agrawal~\cite{Agrawal05} showed that one can use explicit generators for small circuits to obtain hard polynomials, and we observe here that the resulting polynomials also have only hard multiples.

Thus the below claim shows that obtaining black-box PIT yields the existence of a polynomial with hard multiples, which yields average-case hardness, which (for general enough classes) will allow the Kabanets-Impagliazzo~\cite{KabanetsImpagliazzo04} reduction to again yield black-box PIT\@.  Thus, we see that obtaining such polynomials with hard multiples is essentially what is needed for this hardness versus randomness approach.

Note that we give the construction based on a non-trivial \emph{generator} for a class of circuits. While one can analogously prove the \emph{hitting-set} version of this claim, it is weaker.  That is, it is possible to consider classes $\cC$ of unbounded degree and still have generators with small seed-length (see for example \autoref{res:SVb-gen} below), but for such classes one must have hitting sets with infinite size (as hitting univariate polynomials of unbounded degree requires an infinite number of points).

\begin{lemma}\label{res:generator_to_lbs-mult}
	Let $\cC\subseteq\F[\vx]$ be a class of polynomials and let $\cvG:\F^\ell\to\F^\vx$ be a generator for $\cC$.  Suppose $0\ne h\in\F[\vx]$ has $h\circ \cvG=0$.  Then for any nonzero $g\in\F[\vx]$ we have that $g\cdot h\notin \cC$.
\end{lemma}
\begin{proof}
	By definition of $\cvG$, for any $f\in\cC$, $f=0$ iff $f\circ \cvG=0$.  Then for any nonzero $g$, $g\cdot h\ne 0$ and $(g\cdot h)\circ \vG=(g\circ \cvG)\cdot (h\cdot \cvG)=(g\circ \cvG) \cdot 0=0$.  Thus, we must have that $g\cdot h\notin\cC$.
\end{proof}

That is, if $\ell<n$ then such an $h$ exists (as the coordinates of $\cvG$ are algebraically dependent) and such an $h$ can be found in exponential time by solving an exponentially-large linear system.  As such, $h$ is a weakly-explicit polynomial with only hard multiples, which is sufficient for instantiating hardness versus randomness.

While there are now a variety of restricted circuit classes with non-trivial (black-box) PIT results, it seems challenging to find for any given generator $\cG$ an \emph{explicit} nonzero polynomial $f$ with $f\circ \cG=0$.  Indeed, to the best of our knowledge no such examples have ever been furnished for interesting generators. Aside from the quest for polynomials with hard multiples, this question is independently interesting as it demonstrates the limits of the generator in question, especially for generators that are commonly used.  There is not even a consensus as to whether the generators currently constructed could suffice to derandomize PIT for general circuits. Agrawal~\cite{Agrawal05} has even conjectured that a certain generator for depth-2 circuits (sparse polynomials) would actually suffice for PIT of constant-depth circuits.  

We consider here the generator of Shpilka-Volkovich~\cite{ShpilkaVolkovich09}. This generator has a parameter $\ell$, and intuitively can be seen as an algebraic analogue of the boolean pseudorandomness notion of a (randomness efficient) $\ell$-wise independent hash function.  Just as $\ell$-wise independent hash functions are ubiquitous in boolean pseudorandomness, the Shpilka-Volkovich~\cite{ShpilkaVolkovich09} generator has likewise been used in a number of papers on black-box PIT (for example \cite{ShpilkaVolkovich09,AndersonvMV11,ForbesShpilka13a,ForbesSS14,Volkovich15,Forbes15} is a partial list).  As such, we believe it is important to understand the limits of this generator.  

However, $\ell$-wise independence is a \emph{property} of a hash function and likewise the Shpilka-Volkovich~\cite{ShpilkaVolkovich09} generator is really a family of generators that all share a certain property. Specifically, the map $\GSV{\ell,n}:\F^{r}\to\F^n$ has the property\,\footnote{The most obvious algebraic analogue of an $\ell$-wise independent hash function would require that for a generator $\vG:\F^r\to\F^n$ that any subset $S\subseteq[n]$ with $|S|\le \ell$ the output of $\vG$ restricted to $S$ is all of $\F^{S}$.  This property is implied by the Shpilka-Volkovich~\cite{ShpilkaVolkovich09} property, but is strictly weaker, and is in fact too weak to be useful for PIT\@.  That is, consider the map $(x_1,\ldots,x_n)\mapsto (x_1,\ldots,x_n,x_1+\cdots+x_n)$.  This map has this naive ``algebraic $n$-wise independence'' property, but does not even fool linear polynomials (which the Shpilka-Volkovich~\cite{ShpilkaVolkovich09} generator does).} that the image $\GSV{\ell,n}(\F^r)$ contains all $\ell$-sparse vectors in $\F^n$.  The most straightforward construction  of a randomness efficient generator with this property (via Lagrange interpolation, given by Shpilka-Volkovich~\cite{ShpilkaVolkovich09}) has that $r=2\ell$.  Even this construction is actually a family of possibly constructions, as there is significant freedom to choose the finite set of points where Lagrange interpolation will be performed.  As such, instead of studying a specific generator we seek to understand the power of the above \emph{property}, and thus we are free to construct another generator $\GSVb{\ell,n}$ with this property for which we can find an explicit nonzero $f$ where $f\circ \GSVb{\ell,n}=0$ for small $\ell$. We choose a variant of the original construction so that we can take $f$ as the determinant.

In the original Shpilka-Volkovich~\cite{ShpilkaVolkovich09} generator, one first obtains the $\ell=1$ construction by using two variables, the control variable $y$ and another variable $z$.  By using Lagrange polynomials to simulate indicator functions, the value of $y$ can be set to choose between the outputs $z\ve_1,\ldots,z\ve_n\in\F[z]^n$, where $\ve_i\in\F^n$ is the $i$-th standard basis vector.  By varying $z$ one obtains all $1$-sparse vectors in $\F^n$.  To obtain $\GSV{\ell,n}$ one can sum $\ell$ independent copies of $\GSV{1,n}$.  In contrast, our construction will simply use a different control mechanism, where instead of using univariate polynomials we use bivariates.

\begin{construction}\label{const:SVb}
	Let $n,\ell\ge 1$. Let $\F$ be a field of size $\ge n$. Let $\Omega\eqdef\{\omega_1,\ldots,\omega_n\}$ be distinct elements in $\F$. Define $\GSVb{1,n^2}:\F^3\to\F^{n\times n}$ by
	\[
		\left(\GSVb{1,n^2}(x,y,z)\right)_{i,j}=z\cdot \ind{\omega_i,\Omega}(x)\cdot\ind{\omega_j,\Omega}(y)
		\;.
	\]
	where $\ind{\omega_i,\Omega}(x)\in\F[x]$ is the unique univariate polynomial of degree $<n$ such that
	\[
		\ind{\omega_i,\Omega}(\omega_j)=
		\begin{cases}
			1	&	j=i\\
			0	&	\text{else}
		\end{cases}
		\;.
	\]
	Define $\GSVb{\ell,n^2}:\F^{3\ell}\to\F^{n\times n}$ by the polynomial map
	\[
		\GSVb{\ell,n^2}(x_1,y_1,z_1,\ldots,x_\ell,y_\ell,z_\ell)\eqdef\GSVb{1,n^2}(x_1,y_1,z_1)+\cdots+\GSVb{1,n^2}(x_\ell,y_\ell,z_\ell)
		\;,
	\]
	working in the ring $\F[\vx,\vy,\vz]$.
\end{construction}

Note that this map has $n^2$ outputs. Now observe that it is straightforward to see that this map has the desired property.

\begin{lemmawp}
	Assume the setup of \autoref{const:SVb}.  Then the image of the generator, $\GSVb{\ell,n^2}(\F^{3\ell})$, contains all $\ell$-sparse vectors in $\F^{n\times n}$.
\end{lemmawp}

To the best of the authors knowledge, existing works using the Shpilka-Volkovich~\cite{ShpilkaVolkovich09} generator\,\footnote{Note that for black-box PIT it is important that we use a \emph{generator} that contains all sparse vectors, instead of just the \emph{set} of sparse vectors. As an example, the monomial $x_1\cdots x_n$ is zero on all $k$-sparse vectors for $k<n$, but is nonzero when evaluated on the Shpilka-Volkovich~\cite{ShpilkaVolkovich09} generator for any $\ell\ge 1$.} only use the above property (and occasionally also the fact that a coordinate-wise sum of constantly-many such generators is a generator of the original form with similar parameters (\cite{AgrawalSS13,ForbesSS14,GurjarKST15,Forbes15}), which our alternate construction also satisfies).  As such, we can replace the standard construction with our variant in known black-box PIT results (such as \cite{ShpilkaVolkovich09,AgrawalSS13,ForbesShpilka13a,ForbesSS14,GurjarKST15,Forbes15}), some of which we now state.

\begin{corollary}\label{res:SVb-gen}
	Assume the setup of \autoref{const:SVb}. Then $\GSVb{O(\log s),n^2}$ is a generator for the following classes of $n$-variate polynomials, of arbitrary degree.
	\begin{itemize}
		\item Polynomials of sparsity $s$ (\cite{ShpilkaVolkovich09,GurjarKST15,Forbes15}).
		\item Polynomials computable as a depth-$3$ powering formula of top-fan-in $s$ (\cite{AgrawalSS13,ForbesShpilka13a}).
		\item Polynomials computable as a \sumpowc formula of top-fan-in $s$ (\cite{Forbes15}), in characteristic zero.
		\item Polynomials computable by width-$s$ roABPs in every variable order, also known as commutative roABPs (\cite{AgrawalSS13,ForbesSS14}).
	\end{itemize}
\end{corollary}

The above result shows the power of the $\GSVb{\ell,n^2}$ generator to hit restricted classes of circuits.  We now observe that it is also limited by its inability to hit the determinant.

\begin{proposition}
	Assume the setup of \autoref{const:SVb}.  The output of $\GSVb{\ell,n^2}$ is an $n\times n$ matrix of rank $\le \ell$, when viewed as a matrix over the ring $\F(\vx,\vy,\vz)$.  Thus, taking $\det_n\in\F[X]$ to be the $n\times n$ determinant, we have that $\det_n\circ \;\GSVb{\ell,n^2}=0$ for $\ell<n$.
\end{proposition}
\begin{proof}
	\uline{$\ell=1$:} For a field $\K$, a (nonzero) matrix $M\in \K^{n\times n}$ is rank-1 if it can be expressed as an outer-product, so that $(M)_{i,j}=\alpha_i\beta_j$ for $\vaa,\vbb\in \K^n$.  Taking $\vaa,\vbb\in\F(\vx,\vy,\vz)^n$ defined by $\alpha_i\eqdef z\ind{\omega_i,\Omega}(x)$ and $\beta_j\eqdef\ind{\omega_j,\Omega}(y)$ we immediately see that $\GSVb{1,n^2}(x,y,z)$ is rank-1.

	\uline{$\ell>1$:} This follows as rank is subadditive, and $\GSVb{\ell,n^2}$ is the sum of $\ell$ copies of $\GSVb{1,n^2}$.

	\uline{$\det_n$ vanishes:} This follows as the $n\times n$ determinant vanishes on matrices of rank $<n$.
\end{proof}

Note that in the above we could hope to find an $f$ such that $f\circ \GSVb{\ell,n^2}=0$ for all $\ell<n^2$, but we are only able to handle $\ell<n$. Given the above result, along with \autoref{res:generator_to_lbs-mult}, we obtain that the multiples of the determinant are hard.

\begin{corollary}\label{res:lbs-mult:pit:det}
	Let $\det_n\in\F[X]$ denote the $n\times n$ determinant.  Then any nonzero multiple $f\cdot \det_n$ of $\det_n$, for $0 \ne f\in\F[X]$, has the following lower bounds.
	\begin{itemize}
		\item $f\cdot \det_n$ involves $\exp(\Omega(n))$ monomials.
		\item $f\cdot \det_n$ requires size $\exp(\Omega(n))$ to be expressed as a depth-$3$ powering formula.
		\item $f\cdot \det_n$ requires size $\exp(\Omega(n))$ to be expressed as a \sumpowc formula, in characteristic zero.
		\item $f\cdot \det_n$ requires width-$\exp(\Omega(n))$ to be computed as an roABP in some variable order.
	\end{itemize}
\end{corollary}
\begin{proof}
	By \autoref{res:SVb-gen}, $\GSVb{O(\log s),n^2}$ is a generator for the above size-$s$ computations in the above classes.  However, following \autoref{res:generator_to_lbs-mult}, $(f\cdot\det_n)\circ \left(\GSVb{\ell,n^2}\right)=0$ for $\ell<n$.  Thus, if $f\cdot \det_n$ (which is nonzero) is computable in size-$s$ it must be that $O(\log s)\ge n$, so that $s\ge\exp(\Omega(n))$.
\end{proof}

Note that we crucially leveraged that the determinant vanishes on low-rank matrices.  As such, the above results do not seem to imply similar results for the permanent, despite the fact that the permanent is a harder polynomial. That is, recall that because of $\VNP$-completeness of the permanent the determinant $\det_n(X)$ can be written as a projection of the permanent, so that $\det_n(X)=\perm_m(A(X))$ for an affine matrix $A(X)\in\F[X]^{m\times m}$ with $m\le \poly(n)$.  Then, given a multiple $g(Y)\cdot\perm_m(Y)$ one would like to use this projection to obtain $g(A(X))\perm_m( A(X))=g(A(X))\det_n X$, which is a multiple of $\det_n$.  Unfortunately this multiple may not be a \emph{nonzero multiple}: it could be that $g(A(X))=0$, from which no lower bounds for $g(A(X))\det_n(X)$ (and hence $g(Y)\perm_m (Y)$) can be derived.

\subsection{Lower Bounds for Multiples via Leading/Trailing Monomials}\label{sec:lbs-mult:LM}

We now use the theory of leading (and trailing) monomials to obtain explicit polynomials with hard multiples. We aim at finding as simple polynomials as possible so they will give rise to simple ``axioms'' with no small refutations for our proof complexity applications in \autoref{sec:ips-mult}. These results will essentially be immediate corollaries of previous work.

\subsubsection{Depth-3 Powering Formulas}

Kayal~\cite{Kayal08} observed that using the partial derivative method of Nisan and Wigderson~\cite{NisanWigderson96} one can show that these formulas require $\exp(\Omega(n))$ size to compute the monomial $x_1\cdots x_n$.  Forbes and Shpilka~\cite{ForbesShpilka13a} later observed that this result can be made \emph{robust} by modifying the \emph{hardness of representation} technique of Shpilka and Volkovich~\cite{ShpilkaVolkovich09}, in that similar lower bounds apply when the leading monomial involves many variables, as we now quote.

\begin{theoremwp}[Forbes-Shpilka~\cite{ForbesShpilka13a}]\label{res:lbs-mult:LM:sumpowsum}
	Let $f(\vx)\in\F[\vx]$ be computed a \sumpowsum formula of size $\le s$.  Then the leading monomial $\vx^\va=\LM(f)$ involves $\ellzero{\va}\le \lg s$ variables.
\end{theoremwp}

We now note that as the leading monomial is multiplicative (\autoref{res:hom_LM-TM_mult}) this lower bound automatically extends to multiples of the monomial.

\begin{corollary}\label{res:lbs-mult:sumpowsum}
	Any nonzero multiple of $x_1\cdots x_n$ requires size $\ge 2^n$ to be computed as a \sumpowsum formula.
\end{corollary}
\begin{proof}
	Consider any $0\ne g(\vx)\in\F[x_1,\ldots,x_n]$.  Then as the leading monomial is multiplicative (\autoref{res:hom_LM-TM_mult}) we have that $\LM(g\cdot x_1\cdots x_n)=\LM(g)\cdot x_1\cdots x_n$, so that $\LM(g\cdot x_1\cdots x_n)$ involves $n$ variables.  By the robust lower bound (\autoref{res:lbs-mult:LM:sumpowsum}) this implies $g(\vx)\cdot x_1\cdots x_n$ requires size $\ge 2^n$ as a $\sumpowsum$ formula.
\end{proof}

\subsubsection{\texorpdfstring{$\sumpowc$ Formulas}{Sums of Powers of Low-Degree Polynomials}}

Kayal~\cite{Kayal12} introduced the method of shifted partial derivatives, and Gupta-Kamath-Kayal-Saptharishi~\cite{GuptaKKS14} refined it to give exponential lower bounds for various sub-models of depth-4 formulas.  In particular, it was shown that the monomial $x_1\cdots x_n$ requires $\exp(\Omega(n))$-size to be computed as a \sumpowc formula.  Applying the hardness of representation approach of Shpilka and Volkovich~\cite{ShpilkaVolkovich09}, Mahajan-Rao-Sreenivasaiah~\cite{MahajanRS14} showed how to develop a deterministic black-box PIT algorithm for \emph{multilinear} polynomials computed by $\sumpowc$ formulas.  Independently, Forbes~\cite{Forbes15} (following Forbes-Shpilka~\cite{ForbesShpilka13a}) showed that this lower bound can again be made to apply to leading monomials\,\footnote{The result there is stated for trailing monomials, but the argument equally applies to leading monomials.} (which implies a deterministic black-box PIT algorithm for \emph{all} $\sumpowc$ formulas, with the same complexity as Mahajan-Rao-Sreenivasaiah~\cite{MahajanRS14}). This leading monomial lower bound, which we now state, is important for its applications to polynomials with hard multiples.

\begin{theoremwp}[Forbes~\cite{Forbes15}]\label{res:sumpowsumprodt_TM-ubs}
	Let $f(\vx)\in\F[\vx]$ be computed as a $\sumpow{t}$ formula of size $\le s$.  If $\chara(\F)\ge \ideg(\vx^\va)$, then the leading monomial $\vx^\va=\LM(f)$ involves $\ellzero{\va}\le O(t\lg s)$ variables.
\end{theoremwp}

As for depth-$3$ powering formulas (\autoref{res:lbs-mult:sumpowsum}), this immediately yields that all multiples (of degree below the characteristic) of the monomial are hard. 

\begin{corollarywp}\label{res:lbs-mult:sumpowt}
	All nonzero multiples of $x_1\cdots x_n$ of degree $<\chara(\F)$ require size $\ge\exp(\Omega(\nicefrac{n}{t}))$ to be computed as $\sumpow{t}$ formula.
\end{corollarywp}

\subsubsection{Sparse Polynomials}

While the above approaches for $\sumpowsum$ and $\sumpowc$ formulas focus on leading monomials, one cannot show that the leading monomials of sparse polynomials involve few variables as sparse polynomials can easily compute the monomial $x_1\cdots x_n$.  However, following the \emph{translation} idea of Agrawal-Saha-Saxena~\cite{AgrawalSS13}, Gurjar-Korwar-Saxena-Thierauf~\cite{GurjarKST15} showed that sparse polynomials under full-support translations have \emph{some} monomial involving few variables, and Forbes~\cite{Forbes15} (using different techniques) showed that in fact the \emph{trailing} monomial involving few variables (translations do not affect the leading monomial, so the switch to trailing monomials is necessary here).  

\begin{theoremwp}[Forbes~\cite{Forbes15}]\label{res:transdiff_TM-ub_sparse}
	Let $f(\vx)\in\F[\vx]$ be $(\le s)$-sparse, and let $\vaa\in(\F\setminus\{0\})^n$ so that $\vaa$ has full-support.  Then the trailing monomial $\vx^\va=\TM(f(\vx+\vaa))$ involves $\ellzero{\va}\le \lg s$ variables.
\end{theoremwp}

This result thus allows one to construct polynomials whose multiples are all non-sparse.

\begin{corollary}\label{res:lbs-mult:sparse-LM}
	All nonzero multiples of $(x_1+1)\cdots (x_n+1)\in\F[\vx]$ require sparsity $\ge 2^n$.  Similarly, all nonzero multiples of $(x_1+y_1)\cdots (x_n+y_n)\in\F[\vx,\vy]$ require sparsity $\ge 2^n$.
\end{corollary}
\begin{proof}
	Define $f(\vx)=\prod_{i=1}^n (x_i+1)$. For any $0\ne g(\vx)\in\F[\vx]$ the multiple $g(\vx)f(\vx)$ under the translation $\vx\mapsto \vx-\vno$ is equal to $g(\vx-\vno)\prod_i x_i$.  Then all monomials (in particular the trailing monomial) involves $n$ variables (as $g(\vx)\ne 0$ implies $g(\vx-\vno)\ne 0$).  Thus, by \autoref{res:transdiff_TM-ub_sparse} it must be that $g(\vx)f(\vx)$ requires $\ge 2^n$ monomials.

	The second part of the claim follows as the first, where we now work over the fraction field $\F(\vy)[\vx]$, noting that this can only decrease sparsity.  Thus, using the translation $\vx\mapsto\vx-\vy$ the above trailing monomial argument implies that the sparsity of nonzero multiples $\prod_i (x_i+y_i)$ are $\ge 2^n$ over $\F(\vy)[\vx]$ and hence also over $\F[\vx,\vy]$.
\end{proof}

Note that it is tempting to derive the second part of the above corollary from the first, using that the substitution $\vy\leftarrow\vno$ does not increase sparsity.  However, this substitution can convert nonzero multiples of $\prod_i (x_i+y_i)$ to zero multiples of $\prod_i (x_i+1)$, which ruins such a reduction, as argued in the discussion after \autoref{res:lbs-mult:pit:det}.

\subsection{Lower Bounds for Multiples of Sparse Multilinear Polynomials}

While the previous section established that all multiples of $(x_1+1)\cdots (x_n+1)$ are non-sparse, the argument was somewhat specific to that polynomial and fails to obtain an analogous result for $(x_1+1)\cdots (x_n+1)+1$.  Further, while that argument can show for example that all multiples of the $n\times n$ determinant or permanent require sparsity $\ge \exp(\Omega(n))$, this is the best sparsity lower bound obtainable for these polynomials with this method.\footnote{Specifically, as the determinant and permanent are degree $n$ multilinear polynomials, and thus so are their translations, their monomials always involve $\le n$ variables so no sparsity bound better than $2^n$ can be obtained by using \autoref{res:transdiff_TM-ub_sparse}.} In particular, one cannot establish a sparsity lower bound of ``$n!$'' for the determinant or permanent (which would be tight) via this method.

We now give a different argument, due to Oliveira~\cite{Oliveira15b} that establishes a much more general result showing that multiples of \emph{any} multilinear polynomial have at least the sparsity of the original polynomial.  While Oliveira~\cite{Oliveira15b} gave a proof using Newton polytopes, we give a more compact proof here using induction on variables (loosely inspired by a similar result of Volkovich~\cite{Volkovich15b} on the sparsity of factors of multi-quadratic polynomials).

\begin{proposition}[Oliveira~\cite{Oliveira15b}]\label{res:lbs-mult:sparse:newton}
	Let $f(\vx)\in\F[x_1,\ldots,x_n]$ be a nonzero multilinear polynomial with sparsity exactly $s$.  Then any nonzero multiple of $f$ has sparsity $\ge s$.
\end{proposition}
\begin{proof}
	By induction on variables.

	\uline{$n=0$:} Then $f$ is a constant, so that $s=1$ as $f\ne 0$. All nonzero multiples are nonzero polynomials so have sparsity $\ge 1$.

	\uline{$n\ge 1$:} Partition the variables $\vx=(\vy,z)$, so that $f(\vy,z)=f_1(\vy)z+f_0(\vy)$, where $f_i(\vy)$ has sparsity $s_i$ and $s=s_1+s_0$.  Consider any nonzero $g(\vy,z)=\sum_{i=d_0}^{d_1} g_i(\vy)z^i$ with $g_{d_0}(\vy),g_{d_1}(\vy)\ne 0$ (possibly with $d_0=d_1$). Then
	\begin{align*}
		g(\vy,z) f(\vy,z)
		&=\Big(f_1(\vy)z+f_0(\vy)\Big)\cdot \left(\sum_{i=d_0}^{d_1} g_i(\vy)z^i\right)\\
		&=f_1(\vy)g_{d_1}(\vy)z^{d_1+1}+\left[\sum_{d_0<i\le d_1} \Big(f_1(\vy)g_{i-1}(\vy)+f_0(\vy)g_i(\vy)\Big)z^i\right]+f_0(\vy)g_{d_0}(\vy)z^{d_0}
		\;.
	\end{align*}
	By partitioning this sum by powers of $z$ so that there is no cancellation, and then discarding the middle terms,
	\begin{align*}
		\left|\supp\Big(g(\vy,z) f(\vy,z)\Big)\right|
		&\ge\left|\supp\Big(f_1(\vy)g_{d_1}(\vy)\Big)\right|+\left|\supp\Big(f_0(\vy)g_{d_0}(\vy)\Big)\right|
		\intertext{so that appealing to the induction hypothesis, as $f_0$ and $f_1$ are multilinear polynomials of sparsity $s_0$ and $s_1$ respectively,}
		&\ge s_1+s_0=s
		\;.
		\qedhere
	\end{align*}
\end{proof}

We note that multilinearity is essential in the above lemma, even for univariates.  This is seen by noting that the 2-sparse polynomial $x^n-1$ is a multiple of $x^{n-1}+\cdots+x+1$. 

Thus, the above not only gives a different proof of the non-sparsity of multiples of $\prod_i (x_i+1)$ (\autoref{res:lbs-mult:sparse-LM}), but also establishes that nonzero multiples of $\prod_i (x_i+1)+1$ are $\ge 2^n$ sparse, and nonzero multiples of the determinant or permanent are $n!$ sparse, which is tight.  Note further that this lower bound proof is ``monotone'' in that it applies to any polynomial with the same support, whereas the proof of \autoref{res:lbs-mult:sparse-LM} is seemingly not monotone as seen by contrasting $\prod_i (x_i+1)$ and $\prod_i (x_i+1)+1$.

\subsection{Lower Bounds for Multiples by Leading/Trailing Diagonals}

In the previous sections we obtained polynomials with hard multiples for various circuit classes by appealing to the fact that lower bounds for these classes can be reduced to studying the number of variables in leading or trailing monomials.  Unfortunately this approach is restricted to circuit classes where monomials (or translations of monomials) are hard to compute, which in particular rules out this approach for roABPs.  Thus, to develop polynomials with hard multiples for roABPs we need to develop a different notion of a ``leading part'' of a polynomial.  In this section, we define such a notion called a \emph{leading diagonal}, establish its basic properties, and obtain the desired polynomials with hard multiples.  The ideas of this section are a cleaner version of the techniques used in the PIT algorithm of Forbes and Shpilka~\cite{ForbesShpilka12} for commutative roABPs.

\subsubsection{Leading and Trailing Diagonals}

We begin with the definition of a leading diagonal, which is a generalization of a leading monomial.

\begin{definition}
	Let $f\in\F[x_1,\ldots,x_n,y_1,\ldots,y_n]$ be nonzero.  The \demph{leading diagonal of $f$}, denoted $\LD(f)$, is the leading coefficient of $f(\vx\circ\vz,\vy\circ\vz)$ when this polynomial is considered in the ring $\F[\vx,\vy][z_1,\ldots,z_n]$, and where $\vx\circ\vz$ denotes the Hadamard product $(x_1z_1,\ldots,x_nz_n)$.  The \demph{trailing diagonal of $f$} is defined analogously. The zero polynomial has no leading or trailing diagonal.
\end{definition}

As this notion has not explicitly appeared prior in the literature, we now establish several straightforward properties. The first is that extremal diagonals are homomorphic with respect to multiplication.

\begin{lemma}\label{res:diagonals:mult-hom}
	Let $f,g\in\F[x_1,\ldots,x_n,y_1,\ldots,y_n]$ be nonzero.  Then $\LD(fg)=\LD(f)\LD(g)$ and $\TD(fg)=\TD(f)\TD(g)$.
\end{lemma}

\begin{proof}
	As $\LD(f)=\LC_{\vx,\vy|\vz}(f(\vx\circ\vz,\vy\circ\vz))$, where this leading coefficient is taken in the ring $\F[\vx,\vy][\vz]$, this automatically follows from the fact that leading coefficients are homomorphic with respect to multiplication (\autoref{res:hom_LM-TM_mult}).  The result for trailing diagonals is symmetric.
\end{proof}

We now show how to relate the leading monomials of the coefficient space of $f$ to the respective monomials associated to the leading diagonal of $f$.

\begin{proposition}\label{res:diagonal:tm}
	Let $f\in\F[x_1,\ldots,x_n,y_1,\ldots,y_n]$.  For any $\vb$, if $\coeff{\vx|\vy^\vb}(\LD(f))\ne 0$, then 
	\[
		\LMp{\coeff{\vx|\vy^\vb}(\LD(f))}=\LMp{\coeff{\vx|\vy^\vb}(f)}
		\;.
	\]
	The respective trailing statement also holds.
\end{proposition}
\begin{proof}
	We prove the leading statement, the trailing version is symmetric.
	Let $f=\sum_{\va,\vb} \alpha_{\va,\vb}\vx^\va\vy^\vb$.  We can then expand $f(\vx\circ\vz,\vy\circ\vz)$ as follows.
	\begin{align*}
		f(\vx\circ\vz,\vy\circ\vz)
		&=\textstyle\sum_{\vc}\left(\sum_{\va+\vb=\vc} \alpha_{\va,\vb}\vx^\va\vy^\vb\right)\vz^\vc\\
		\intertext{choose $\vc_0$ so that $\LC_{\vx,\vy|\vz}(f)=\coeff{\vx,\vy|\vz^{\vc_0}}(f)$, we get that}
		&=\textstyle\left(\sum_{\va+\vb=\vc_0} \alpha_{\va,\vb}\vx^\va\vy^\vb\right)\vz^{\vc_0}+\sum_{\vc\prec \vc_0}\left(\sum_{\va+\vb=\vc} \alpha_{\va,\vb}\vx^\va\vy^\vb\right)\vz^\vc
		\;,
	\end{align*}
	where $\LD(f)=\sum_{\va+\vb=\vc_0} \alpha_{\va,\vb}\vx^\va\vy^\vb$ and $\sum_{\va+\vb=\vc} \alpha_{\va,\vb}\vx^\va\vy^\vb=0$ for $\vc\succ\vc_0$.  In particular, this means that for any $\vb$ we have that $\alpha_{\va,\vb}=0$ for $\va\succ\vc_0-\vb$.

	Thus we have that
	\begin{align*}
		\LMp{\coeff{\vx|\vy^\vb}(\LD(f))}
		&=\textstyle\LMp{\coeff{\vx|\vy^\vb}\left(\sum_{\va+\vb=\vc_0} \alpha_{\va,\vb}\vx^\va\vy^\vb\right)}\\
		&=\LMp{\alpha_{\vc_0-\vb,\vb}\vx^{\vc_0-\vb}}\\
		&=\vx^{\vc_0-\vb}
		\;,
	\end{align*}
	as we assume this leading monomial exists, which is equivalent here to $\alpha_{\vc_0-\vb,\vb}\ne 0$.

	In comparison,
	\begin{align*}
		\LMp{\coeff{\vx|\vy^\vb}(f)}
		&=\textstyle\LMp{\coeff{\vx|\vy^\vb}\left(\sum_{\va,\vb} \alpha_{\va,\vb}\vx^\va\vy^\vb\right)}\\
		&=\textstyle\LMp{\sum_{\va} \alpha_{\va,\vb}\vx^\va}\\
		&=\textstyle\LMp{\sum_{\va\succ \vc_0-\vb} \alpha_{\va,\vb}\vx^\va + \alpha_{\vc_0-\vb,\vb}\vx^{\vc_0-\vb}+\sum_{\va\prec \vc_0-\vb} \alpha_{\va,\vb}\vx^\va}
		\intertext{as $\alpha_{\va,\vb}=0$ for $\va\succ\vc_0-\vb$,}
		&=\textstyle\LMp{\alpha_{\vc_0-\vb,\vb}\vx^{\vc_0-\vb}+\sum_{\va\prec \vc_0-\vb} \alpha_{\va,\vb}\vx^\va}\\
		&=\vx^{\vc_0-\vb}
		\;,
	\end{align*}
	where in the last step we again used that $\alpha_{\vc_0-\vb,\vb}\ne0$.  This establishes the desired equality.
\end{proof}

We now relate the extremal monomials of the coefficient space of $f$ to the monomials of the coefficient space of the extremal diagonals of $f$.

\begin{corollary}\label{res:diagonal:lm-coeffs}
	Let $f\in\F[x_1,\ldots,x_n,y_1,\ldots,y_n]$.  Then
	\[
		\LM(\coeffs{\vx|\vy}(f))\supseteq \LM(\coeffs{\vx|\vy}(\LD(f)))
		\;.
	\]
	The respective trailing statement also holds.
\end{corollary}
\begin{proof}
	This follows as $\LM(\coeffs{\vx|\vy}(\LD(f)))$ is equal to 
	\[
		\left\{\left.\LM\left(\coeff{\vx|\vy^\vb}\big(\LD(f)\big)\right) \, \right|\,  \coeff{\vx|\vy^\vb}\big(\LD(f)\big)\ne 0\right\}
		\;,
	\]	
	but by \autoref{res:diagonal:tm} this set equals 
	\[
		\left\{\left.\LM\left(\coeff{\vx|\vy^\vb}(f)\right) \, \right|\,  \coeff{\vx|\vy^\vb}\big(\LD(f)\big)\ne 0\right\}
		\;,
	\]	
	which is clearly contained in $\LM(\coeffs{\vx|\vy}(f))$.
\end{proof}

We now observe that the number of leading monomials of the coefficient space of a leading diagonal is equal to its sparsity.

\begin{lemma}\label{lem:LM-by-LD}
	Let $f\in\F[x_1,\ldots,x_n,y_1,\ldots,y_n]$. For a polynomial $g$, let $\ellzero{g}$ denotes its sparsity. Then
	\[
		\left|\LM\left(\coeffs{\vx|\vy}\big(\LD(f)\big)\right)\right|
		=\ellzero{\LD(f)}
		\;.
	\]
	The respective trailing statement also holds.
\end{lemma}
\begin{proof}
	We prove the claim for the leading diagonal, the trailing statement is symmetric. Note that the claim is a vacuous ``0=0'' if $f$ is zero.  For nonzero $f$, express it as $f=\sum_{\va,\vb} \alpha_{\va,\vb}\vx^\va\vy^\vb$ so that $\LD(f)=\sum_{\va+\vb=\vc_0}\alpha_{\va,\vb}\vx^\va\vy^\vb=\sum_{\vb}\alpha_{\vc_0-\vb,\vb}\vx^{\vc_0-\vb}\vy^\vb$ for some $\vc_0\in\N^n$.  Then $\coeff{\vx|\vy^\vb}(\LD(f))=\alpha_{\vc_0-\vb,\vb}\vx^{\vc_0-\vb}$.  As the monomials $\vx^{\vc_0-\vb}$ are distinct and hence linearly independent for distinct $\vb$, it follows that $\dim\coeffs{\vx|\vy}(\LD(f))=|\{\vb| \alpha_{\vc_0-\vb,\vb}\ne 0\}|$, which is equal the sparsity $\ellzero{\LD(f)}$.
\end{proof}

Finally, we now lower bound the coefficient dimension of a polynomial by the sparsity of its extremal diagonals.

\begin{corollarywp}\label{res:diagonal:coeff-dim-lb}
	Let $f\in\F[x_1,\ldots,x_n,y_1,\ldots,y_n]$.  Then
	\[
		\dim\coeffs{\vx|\vy}(f)\ge \ellzero{\LD(f)},\ellzero{\TD(f)}
		\;,
	\]
	where for a polynomial $g$, $\ellzero{g}$ denotes its sparsity.
\end{corollarywp}
\begin{proof}
	We give the proof for the leading diagonal, the trailing diagonal is symmetric.  By the above,
	\[
		\dim\coeffs{\vx|\vy}(f)
		\ge \left|\LM\left(\coeffs{\vx|\vy}(f)\right)\right|
		\ge \left|\LM\left(\coeffs{\vx|\vy}\big(\LD(f)\big)\right)\right|
		=  \ellzero{\LD(f)}
		\;,
	\]
	where we passed from span to number of leading monomials (\autoref{res:dim-eq-num-TM-spn}), and then passed to the leading monomials of the leading diagonal (\autoref{res:diagonal:lm-coeffs}), and then passed to sparsity of the leading diagonal (\autoref{lem:LM-by-LD}).
\end{proof}

\subsubsection{Lower Bounds for Multiples for Read-Once and Read-Twice ABPs}

Having developed the theory of leading diagonals in the previous section, we now turn to using this theory to obtain explicit polynomials whose nonzero multiples all require large roABPs. We also generalize this to read-$O(1)$ oblivious ABPs, but only state the results for $k=2$ as this has a natural application to proof complexity (\autoref{sec:ips-mult}).  As the restricted computations considered above (\sumpowsum formulas and sparse polynomials) have small roABPs, the hard polynomials in this section will also have multiples requiring large complexity in these models as well and thus qualitatively reprove some of the above results. However, we included the previous sections as the hard polynomials there are simpler (being monomials or translations of monomials), and more importantly we will need those results for the proofs below.

The proofs will use the characterization of roABPs by their coefficient dimension (\autoref{res:roABP-width_eq_dim-coeffs}), the lower bound for coefficient dimension in terms of the sparsity of the extremal diagonals (\autoref{res:diagonal:coeff-dim-lb}), and polynomials whose multiples are all non-sparse (\autoref{res:lbs-mult:sparse-LM}).

\begin{proposition}\label{prop:good-f-roabp}
	Let $f(\vx,\vy)\eqdef\prod_{i=1}^n (x_i+y_i+\alpha_i)\in\F[x_1,\ldots,x_n,y_1,\ldots,y_n]$, for $\alpha_i\in\F$.  Then for any $0\ne g\in\F[\vx,\vy]$,
	\[
		\dim\coeffs{\vx|\vy}(g\cdot f)\ge 2^n
		\;.
	\]
	In particular, all nonzero multiples of $f$ require width at least $2^n$ to be computed by an roABP in any variable order where $\vx\prec \vy$.
\end{proposition}
\begin{proof}
	Observe that leading diagonal of $f$ is insensitive to the $\alpha_i$.  That is, $\LD(x_i+y_i+\alpha_i)=x_i+y_i$, so by multiplicativity of the leading diagonal (\autoref{res:diagonals:mult-hom}) we have that $\LD(f)=\prod_i (x_i+y_i)$.  Thus, appealing to \autoref{res:diagonal:coeff-dim-lb} and \autoref{res:lbs-mult:sparse-LM},
	\begin{align*}
		\dim\coeffs{\vx|\vy}(g\cdot f)
		&\ge \ellzero{\LD(g\cdot f)}\\
		&=\ellzero{\LD(g)\cdot \LD(f)}\\
		&=\textstyle\ellzero{\LD(g)\cdot \prod_i (x_i+y_i)}\\
		&\ge 2^n
		\;.
	\end{align*}

	The claim about roABP width follows from \autoref{res:roABP-width_eq_dim-coeffs}.
\end{proof}

Note that this lower bound actually works in the ``monotone'' setting (if we replace \autoref{res:lbs-mult:sparse-LM} with the monotone \autoref{res:lbs-mult:sparse:newton}), as the result only uses the zero/nonzero pattern of the coefficients.

The above result gives lower bounds for coefficient dimension in a \emph{fixed} variable partition. We now symmetrize this construction to get lower bounds for coefficient dimension in \emph{any} variable partition.  We proceed as in \autoref{sec:lbs-fn:every-partition}, where we plant the fixed-partition lower bound into an arbitrary partition.  Note that unlike that construction, we will not need auxiliary variables here.

\begin{corollary}\label{cor:lbs-mult:every-partition}
	Let $f\in\F[x_1,\ldots,x_n]$ be defined by $f(\vx)\eqdef\prod_{i<j}(x_i+x_j+\alpha_{i,j})$ for $\alpha_{i,j}\in\F$.  Then for any partition $\vx=(\vu,\vv,\vw)$ with $m\eqdef|\vu|=|\vv|$, and any $0\ne g\in\F[\vx]$,
	\[
		\dim_{\F(\vw)}\coeffs{\vu|\vv}(g\cdot f)\ge 2^m
		\;,
	\]
	where we treat $g\cdot f$ as a polynomial in $\F(\vw)[\vu,\vv]$.  In particular, all nonzero multiples of $f$ require width $\ge 2^n$ to be computed by an roABP in any variable order.
\end{corollary}
\begin{proof}
	We can factor $f$ into a copy of the hard polynomial from \autoref{prop:good-f-roabp}, and the rest.  That is, 
	\begin{align*}
		f(\vx)
		=\prod_{i<j} (x_i+x_j+\alpha_{i,j})
		=\prod_{i=1}^m(u_i+v_i+\beta_i)\cdot f'(\vu,\vv,\vw)
		\;,
	\end{align*}
	for some $\beta_i\in\F$ and nonzero $f'(\vu,\vv,\vw)\in\F[\vu,\vv,\vw]$.  Thus,
	\[
		g\cdot f
		=\left(g(\vu,\vv,\vw)\cdot f'(\vu,\vv,\vw)\right)\cdot \prod_{i=1}^m(u_i+v_i+\beta_i)
		\;.
	\]
	Noting that $g,f'$ are nonzero in $\F[\vu,\vv,\vw]$, they are also nonzero in $\F(\vw)[\vu,\vv]$, so that $g\cdot f$ is nonzero multiple of $\prod_{i=1}^m(u_i+v_i+\beta_i)$ in $\F(\vw)[\vu,\vv]$.  Appealing to our lower bound for (nonzero) multiples of coefficient dimension (\autoref{prop:good-f-roabp}), we have that
	\[
		\dim_{\F(\vw)}\coeffs{\vu|\vv}(g\cdot f)
		=\dim_{\F(\vw)}\coeffs{\vu|\vv}\left(g\cdot f'\cdot \prod_{i=1}^m(u_i+v_i+\beta_i)\right)
		\ge 2^m
		\;.
	\]
	The statement about roABPs follows from \autoref{res:roABP-width_eq_dim-coeffs}.
\end{proof}

We briefly remark that the above bound does not match the naive bound achieved by writing the polynomial $\prod_{i<j}(x_i+x_j+\alpha_{i,j})$ in its sparse representation, which has $2^{\Theta(n^2)}$ terms.  The gap between the lower bound ($2^{\Omega(n)}$) and the upper bound ($2^{O(n^2)}$) is explained by our use of a complete graph to embed the lower bounds of \autoref{prop:good-f-roabp} into an arbitrary partition.  As discussed after \autoref{res:lbs-fn:any-order:coeff-dim} one can use expander graphs to essentially close this gap.

We now observe that the above lower bounds for coefficient dimension suffices to obtain lower bounds for read-twice oblivious ABPs, as we can appeal to the structural result of Anderson, Forbes, Saptharishi, Shpilka and Volk~\cite{AndersonFSSV16} (\autoref{thm:read-k-eval-dim}).  This result shows that for any read-twice oblivious ABP that (after discarding some variables) there is a partition of the variables across which has small coefficient dimension, which is in contrast to the above lower bound.

\begin{corollary}\label{cor:r2abp-multiples}
	Let $f\in\F[x_1,\ldots,x_n]$ be defined by $f(\vx)\eqdef\prod_{i<j}(x_i+x_j+\alpha_{i,j})$ for $\alpha_{i,j}\in\F$.  Then for any $0\ne g\in\F[\vx]$, $g\cdot f$ requires width-$2^{\Omega(n)}$ as a read-twice oblivious ABP.
\end{corollary}
\begin{proof}
	Suppose that $g\cdot f$ has a read-twice oblivious ABP of width-$w$. By the lower-bound of Anderson, Forbes, Saptharishi, Shpilka and Volk~\cite{AndersonFSSV16} (\autoref{thm:read-k-eval-dim}), there exists a partition $\vx=(\vu,\vv,\vw)$ where $|\vu|,|\vv|\ge\Omega(n)$, and such that $\dim_{\F(\vw)}\coeffs{\vu|\vv}(g\cdot f)\le w^4$ (where we treat $g\cdot f$ as a polynomial in $\F(\vw)[\vu,\vv]$).  Note that we can take enforce that the partition obeys $m\eqdef|\vu|=|\vv|\ge\Omega(n)$, as we can balance $\vu$ and $\vv$ by pushing variables into $\vw$, as this cannot increase the coefficient dimension (\autoref{fact:roABP:closure}). However, appealing to our coefficient dimension bound (\autoref{cor:lbs-mult:every-partition})
	\[
		w^4\ge\dim_{\F(\vw)}\coeffs{\vu|\vv}(g\cdot f)\ge 2^m\ge 2^{\Omega(n)}
		\;,
	\]
	so that $w\ge2^{\Omega(n)}$ as desired.
\end{proof}

\section{IPS Lower Bounds via Lower Bounds for Multiples}\label{sec:ips-mult}

In this section we use the lower bounds for multiples of \autoref{sec:lbs-mult} to derive lower bounds for $\cC$-IPS proofs for various restricted algebraic circuit classes $\cC$.  The advantage of this approach over the functional lower bounds strategy of \autoref{sec:lbs-fn} is that we derive lower bounds for the general IPS system, not just its subclass linear-IPS\@.  While our equivalence (\autoref{res:h-ips_v_ips}) of $\cC$-IPS and $\cC$-\lIPS holds for any strong-enough class $\cC$, the restricted classes we consider here (depth-3 powering formulas and roABPs)\,\footnote{As in \autoref{sec:lbs-mult}, we will not treat multilinear formulas in this section as they are less natural for the techniques under consideration.  Further,  IPS lower bounds for multilinear formulas \emph{can} be obtained via functional lower bounds (\autoref{res:lbs-fn:lbs-ips:vary-order}).} are not strong enough to use \autoref{res:h-ips_v_ips} to lift the results of \autoref{sec:lbs-fn} to lower bounds for the full IPS system.  However, as discussed in the introduction, the techniques of this section can only yield lower bounds for $\cC$-IPS refutations of systems of equations which are hard to compute within $\cC$ (though our examples are computable by small (general) circuits).

We begin by first detailing the relation between IPS refutations and multiples.  We then use our lower bounds for multiples (\autoref{sec:lbs-mult}) to derive as corollaries lower bounds for \sumpowsum-IPS and roABP-IPS refutations.

\begin{lemma}\label{res:lbs-mult:strategy}
	Let $f,\vg,\vx^2-\vx\in\F[x_1,\ldots,x_n]$ be an unsatisfiable system of equations, where $\vg,\vx^2-\vx$ is satisfiable.  Let $C\in\F[\vx,y,\vz,\vw]$ be an IPS refutation of $f,\vg,\vx^2-\vx$.  Then
	\[
		1-C(\vx,0,\vg,\vx^2-\vx)
	\]
	is a nonzero multiple of $f$.
\end{lemma}
\begin{proof}
	That $C$ is an IPS refutation means that
	\[
		C(\vx,f,\vg,\vx^2-\vx)=1,\qquad
		C(\vx,0,\vnz,\vnz)=0
		\;.
	\]
	We first show that $1-C(\vx,0,\vg,\vx^2-\vx)$ is a multiple of $f$, using the first condition on $C$. Expand $C(\vx,y,\vz,\vw)$ as a univariate in $y$, so that
	\[
		C(\vx,y,\vz,\vw)=\sum_{i\ge 0} C_i(\vx,\vz,\vw)y^i
		\;,
	\]
	for $C_i\in\F[\vx,\vz,\vw]$. In particular, $C_0(\vx,\vz,\vw)=C(\vx,0,\vz,\vw)$. Thus,
	\begin{align*}
		1-C(\vx,0,\vg,\vx^2-\vx)
		&=C(\vx,f,\vg,\vx^2-\vx)-C(\vx,0,\vg,\vx^2-\vx)\\
		&=\textstyle\left(\sum_{i\ge 0} C_i(\vx,\vg,\vx^2-\vx)f^i\right)-C_0(\vx,\vg,\vx^2-\vx)\\
		&=\textstyle\sum_{i\ge 1} C_i(\vx,\vg,\vx^2-\vx)f^i\\
		&=\textstyle\left(\sum_{i\ge 1} C_i(\vx,\vg,\vx^2-\vx)f^{i-1}\right)\cdot f
		\;.
	\end{align*}
	Thus, $1-C(\vx,0,\vg,\vx^2-\vx)$ is a multiple of $f$ as desired.

	We now show that this is a \emph{nonzero} multiple, using the second condition on $C$ and the satisfiability of $\vg,\vx^2-\vx$.  That is, the second condition implies that $0=C(\vx,0,\vnz,\vnz)=C_0(\vx,\vnz,\vnz)$.  If $1-C(\vx,0,\vg,\vx^2-\vx)$ is zero, then by the above we have that $C_0(\vx,\vg,\vx^2-\vx)=1$, so that $C_0(\vx,\vz,\vw)$ is an IPS refutation of $\vg,\vx^2-\vx$, which contradicts the satisfiability of $\vg,\vx^2-\vx$ as IPS is a sound proof system.  So it must then be that $1-C(\vx,0,\vg,\vx^2-\vx)$ is nonzero.  
	
	That is, take an $\vaa$ satisfying $\vg,\vx^2-\vx$ so that $\vg(\vaa)=\vnz,\vaa^2-\vaa=\vnz$. Substituting this $\vaa$ into $C_0(\vx,\vg,\vx^2-\vx)$, we have that $C_0(\vx,\vg,\vx^2-\vx)|_{\vx\leftarrow\vaa}=C_0(\vaa,\vnz,\vnz)$, and because $C_0(\vx,\vnz,\vnz)\equiv0$ in $\F[\vx]$ via the above we have that $C_0(\vaa,\vnz,\vnz)=0$. Thus, we have that $1-C(\vx,0,\vg,\vx^2-\vx)=1-C_0(\vx,\vg,\vx^2-\vx)$ is a nonzero polynomial as its evaluation at $\vx\leftarrow\vaa$ is $1$.
\end{proof}

The above lemma thus gives a template for obtaining lower bounds for IPS\@.  First, obtain a ``hard'' polynomial $f$ whose nonzero multiples are hard for $\cC$, where $f$ is hopefully also computable by small (general) circuits.  Then find additional (simple) polynomials $\vg$ such that $\vg,\vx^2-\vx$ is satisfiable yet $f,\vg,\vx^2-\vx$ is unsatisfiable.  By the above lemma one then has the desired IPS lower bound for refuting $f,\vg,\vx^2-\vx$, assuming that $\cC$ is sufficiently general.  However, for our results we need to more careful as even though $C(\vx,y,\vz,\vw)$ is from the restricted class $\cC$, the derived polynomial $C(\vx,0,\vg,\vx^2-\vx)$ may not be, and as such we will need to appeal to lower bounds for stronger classes.

We now instantiate this template, first for depth-3 powering formulas, where we use lower bounds for multiples of the stronger \sumpow{2} model.

\begin{corollary}\label{res:lbs-ips:mult:sumpowsum}
	Let $\F$ be a field with $\chara(\F)=0$.  Let $f\eqdef x_1\cdots x_n$ and $g\eqdef x_1+\cdots+x_n-n$ with $f,g\in\F[x_1,\ldots,x_n]$.  Then $f,g,\vx^2-\vx$ are unsatisfiable and any $\sumpowsum$-IPS refutation requires size at least $\exp(\Omega(n))$.
\end{corollary}
\begin{proof}
	The hypothesis $\chara(\F)=0$ implies that $\{0,\ldots,n\}$ are distinct numbers.  In particular, the system $g(\vx)=0$ and $\vx^2-\vx=\vnz$ is satisfiable and has the unique satisfying assignment $\vno$.  However, this single assignment does not satisfy $f$ as $f(\vno)=\prod_{i=1}^n 1=1\ne 0$, so the entire system is unsatisfiable.  Thus, applying our strategy (\autoref{res:lbs-mult:strategy}), we see that for any \sumpowsum-IPS refutation $C(\vx,y,z,\vw)$ of $f,g,\vx^2-\vx$ that $1-C(\vx,0,g,\vx^2-\vx)$ is a nonzero multiple of $f$.  
	
	Let $s$ be the size of $C$ as a $\sumpowsum$ formula.  As $g$ is linear and the boolean axioms $\vx^2-\vx$ are quadratic, it follows that $1-C(\vx,0,g,\vx^2-\vx)$ is a sum of powers of quadratics ($\sumpow{2}$) of size $\poly(s)$.  As nonzero multiples of $f$ requires $\exp(\Omega(n))$-size as a \sumpow{2} formula (\autoref{res:lbs-mult:sumpowt}) it follows that $\poly(s)\ge\exp(\Omega(n))$, so that $s\ge\exp(\Omega(n))$ as desired.
\end{proof}

We similarly obtain a lower bound for roABP-IPS, where here we use lower bounds for multiples of read-\emph{twice} oblivious ABPs.

\begin{corollary}\label{res:lbs-ips:mult:roABP}
	Let $\F$ be a field with $\chara(\F)>n$. Let $f\eqdef \prod_{i<j}(x_i+x_j-1)$ and $g\eqdef x_1+\cdots+x_n-n$ with $f,g\in\F[x_1,\ldots,x_n]$.  Then $f,g,\vx^2-\vx$ are unsatisfiable and any roABP-IPS refutation (in any variable order) requires size $\ge\exp(\Omega(n))$.
\end{corollary}
\begin{proof}
	The hypothesis $\chara(\F)>n$ implies that $\{0,\ldots,n\}$ are distinct numbers.  In particular, the system $g(\vx)=0$ and $\vx^2-\vx=\vnz$ is satisfiable and has the unique satisfying assignment $\vno$.  However, this single assignment does not satisfy $f$ as $f(\vno)=\prod_{i<j}(1+1-1)=1\ne0$, so the entire system is unsatisfiable.  Thus, applying our strategy (\autoref{res:lbs-mult:strategy}), we see that for any roABP-IPS refutation $C(\vx,y,z,\vw)$ of $f,g,\vx^2-\vx$ that $1-C(\vx,0,g,\vx^2-\vx)$ is a nonzero multiple of $f$.  
	
	Let $s$ be the size of $C$ as an roABP, and we now argue that $1-C(\vx,0,g,\vx^2-\vx)$ has a small read-\emph{twice} oblivious ABP\@.  First, note that we can expand $C(\vx,0,z,\vw)$ into powers of $z$, so that $C(\vx,0,z,\vw)=\sum_{0\le i\le s} C_i(\vx,\vw)z^i$ (where we use that $s$ bounds the width and \emph{degree} of the roABP $C$).  Each $C_i(\vx,\vw)$ has a $\poly(s)$-size roABP (in the variable order of $C$ where $z$ is omitted) as we can compute $C_i$ via interpolation over $z$, using that each evaluation preserves roABP size (\autoref{fact:roABP:closure}).  Further, as $g$ is linear, for any $i$ we see that $g^i$ can be computed by a $\poly(n,i)$-size roABP (in any variable order) (\autoref{res:sumpowsum:roABP}).  Combining these facts using closure properties of roABPs under addition and multiplication (\autoref{fact:roABP:closure}), we see that $C(\vx,0,g,\vw)$, and hence $1-C(\vx,0,g,\vw)$, has a $\poly(s,n)$-size roABP in the variable order that $C$ induces on $\vx,\vw$.  Next observe, that as each boolean axiom $x_i^2-x_i$ only refers to a single variable, substituting $\vw\leftarrow\vx^2-\vx$ in the roABP for $1-C(\vx,0,g,\vw)$ will preserve obliviousness of the ABP, but now each variable will be read twice, so that $1-C(\vx,0,g,\vx^2-\vx)$ has a $\poly(s,n)$-size read-twice oblivious ABP.
	
	Now, using that nonzero multiples of $f$ requires $\exp(\Omega(n))$-size to be computed as read-twice oblivious ABPs (\autoref{cor:r2abp-multiples}) it follows that $\poly(s,n)\ge\exp(\Omega(n))$, so that $s\ge\exp(\Omega(n))$ as desired.
\end{proof}

We note that the above lower bound is for the \emph{size} of the roABP\@. One can also obtain the stronger result (for similar but less natural axioms) showing that the \emph{width} (and hence also the size) of the roABP must be large, but we do not pursue this as it does not qualitatively change the result.

\section{Discussion}\label{sec:open-problems}

In this work we proved new lower bounds for various natural restricted versions of the Ideal Proof System (IPS) of Grochow and Pitassi~\cite{GrochowPitassi14}. While existing work in algebraic proof complexity showed limitations of weak measures of complexity such as the degree and sparsity of a polynomial, our lower bounds are for stronger measures of circuit size that match many of the frontier lower bounds in algebraic circuit complexity.  However, our work leaves several open questions and directions for further study, which we now list.

\begin{enumerate}
	\item Can one obtain proof complexity lower bounds from the recent techniques for lower bounds for depth-4 circuits, such as the results of Gupta, Kamath, Kayal and Saptharishi~\cite{GuptaKKS14}? Neither of our approaches (functional lower bounds or lower bounds for multiples) currently extend to their techniques.

	\item Many proof complexity lower bounds are for refuting unsatisfiable $k$-CNFs, where $k=O(1)$, which can be encoded as systems of polynomial equations where each equation involves $O(1)$ variables.  Can one obtain interesting IPS lower bounds for such systems? Our techniques only establish exponential lower bounds where there is at least one axiom involving $\Omega(n)$ variables.

	\item Given an equation $f(\vx)=0$ where $f$ has a size-$s$ circuit, there is a natural way to convert this equation to $\poly(s)$-many equations on $O(1)$ \emph{extension variables} by tracing through the computation of $f$.  Can one understand how introducing extension variables affects the complexity of refuting polynomial systems of equations? This seems a viable approach to the previous question when applied to our technique of using lower bounds for multiples.

	\item We have shown various lower bounds for multiples by invoking the hardness of the determinant (\autoref{res:lbs-mult:pit:det}), but this does not lead to satisfactory proof lower bounds as the axioms are complicated.  Can one \emph{implicitly} invoke the hardness of the determinant?  For example, consider the hard matrix identities suggested by Cook and Rackoff (see for example the survey of Beame and Pitassi~\cite{BeamePitassi98}) and later studied by Soltys and Cook~\cite{SoltysCook04}.  That is, consider unsatisfiable equations such as $XY-\I_n,YX-2\cdot \I_n$, where $X$ and $Y$ are symbolic $n\times n$ matrices and $\I_n$ is the $n\times n$ identity matrix.  The simplest refutations known involve the determinant (see \Hrubes-Tzameret~\cite{HrubesTzameret15}, and the discussion in Grochow-Pitassi~\cite{GrochowPitassi14}), can one provide evidence that computing the determinant is intrinsic to such refutations?

	\item The lower bounds of this paper are for the \emph{static} IPS system, where one cannot simplify intermediate computations.  There are also \emph{dynamic} algebraic proof systems (see \autoref{sec:alg-proofs}), can one extend our techniques to that setting?
\end{enumerate}

\section*{Acknowledgments}

We would like to thank Rafael Oliviera for telling us of \autoref{res:lbs-mult:sparse:newton}, Mrinal Kumar and Ramprasad Saptharishi for conversations~\cite{ForbesKS16} clarifying the roles of functional lower bounds in this work, as well as Avishay Tal for pointing out how \autoref{res:subsetsum:multlin} implies an optimal functional lower bound for sparsity (\autoref{res:subsetsum:multlin:deg-sparse}).  We would also like to thank Joshua Grochow for helpful discussions regarding this work. We are grateful for the anonymous reviewers for their careful read of the paper and for their comments.

{\footnotesize
\bibliographystyle{customurlbst/alphaurlpp}
\ifthenelse{\equal{\me}{miforbes}}
{
	\bibliography{lbs-ips}
}
{
	\bibliography{PrfCmplx-Bakoma}
}
}

\appendix 

\newpage

\section{Relating IPS to Other Proof Systems}\label{sec:alg-proofs}

In this section we summarize some existing work on algebraic proof systems and how these other proof systems compare to IPS\@. In particular, we define the (dynamic) \emph{Polynomial Calculus} refutation system over circuits (related to but slightly different than the system of Grigoriev and Hirsch~\cite{GrigorievHirsch03}) and relate it to the (static) IPS system (\cite{Pitassi97,GrochowPitassi14}) considered in this paper.  We then examine the roABP-PC system, essentially considered by Tzameret~\cite{Tzameret11}, and its separations from sparse-PC\@.  Finally, we consider multilinear-formula-PC as studied by Raz and Tzameret~\cite{RazTzameret08a,RazTzameret08b} and show that its tree-like version simulates multilinear-formula-IPS, and is hence separated from sparse-PC.

\subsection{Polynomial Calculus Refutations}

A substantial body of prior work considers \emph{dynamic} proof systems, which are systems that allow simplification of intermediate polynomials in the proof.  In contrast, IPS is a \emph{static} system where the proof is single object with no ``intermediate'' computations to simplify.  We now define the principle dynamic system of interest, the \emph{Polynomial Calculus} system.  We give a definition over an arbitrary circuit class, which generalizes the definition of the system as introduced by Clegg, Edmonds, and Impagliazzo~\cite{CleggEI96}.

\begin{definition}\label{defn:pc}
	Let $f_1(\vx),\ldots,f_m(\vx)\in\F[x_1,\ldots,x_n]$ be a system of polynomials. A \demph{Polynomial Calculus (PC) proof} for showing that $p\in\F[\vx]$ is in the ideal generated by $\vf,\vx^2-\vx$ is a directed acyclic graph with a single sink, where
	\begin{itemize}
		\item Leaves are labelled with an equation from $\vf,\vx^2-\vx$. 
		\item An internal node $v$ with children $u_1,\ldots,u_k$ for $k>1$ is labelled with a linear combination $v=\alpha_1 u_1+\cdots +\alpha_ku_k$ for $\alpha_i\in \F$. 
		\item An internal node $v$ with a single child $u$ is labelled with the product $g\cdot u$ for some $g\in\F[\vx]$.
	\end{itemize}
	The \demph{value} of a node in the proof is defined inductively via the above labels interpreted as equations, and the value of the output node is required to be the desired polynomial $p$.  The proof is \demph{tree-like} if the underlying graph is a tree, and is otherwise \demph{dag-like}.  A \demph{PC refutation} of $\vf,\vx^2-\vx$ is a proof that 1 is in the ideal of $\vf,\vx^2-\vx$ so that $\vf,\vx^2-\vx$ is unsatisfiable.
	
	The \demph{size} of each node is defined inductively as follows.
	\begin{itemize}
		\item The size of a leaf $v$ is the size of the minimal circuit agreeing with the value of $v$.
		\item The size of an addition node $v=\alpha_1 u_1+\cdots +\alpha_ku_k$ is $k$ plus the size of each child $u_i$, plus the size of the minimal circuit agreeing with the value of $v$.
		\item The size of a product node $v=g\cdot u$ is the size of the child $u$ plus the size of the minimal circuit agreeing with the value of $v$.
	\end{itemize}
	The size of the proof is the sum of the sizes of each node in the proof. For a restricted algebraic circuit class $\cC$, a \demph{$\cC$-PC proof} is a PC proof where the circuits are measured as their size coming from the restricted class $\cC$.
\end{definition}

As with IPS, one can show this is a sound and complete proof system for unsatisfiability of equations.  Also as with IPS, in our definition of PC we included the boolean axioms $\vx^2-\vx$ as this in the most common regime.

An important aspect of the above proof system is that it is \emph{semantic}, as the polynomials derived in the proof are simplified to their smallest equivalent algebraic circuit.  This is a valid in that such simplifications can be efficiently verified (with randomness) using polynomial identity testing (which can sometimes be derandomized, see \autoref{sec:background}). In contrast, one could instead require a \emph{syntactic} proof system, which would have to provide a proof via syntactic manipulation of algebraic circuits that such simplifications are valid.  We will focus on semantic systems as they more naturally compare with IPS, which also requires polynomial identity testing for verification.

While many priors work (\cite{Pitassi97,RazTzameret08a,RazTzameret08b,Tzameret11,GrochowPitassi14}) considered algebraic proof systems whose verification relied on polynomial identity testing (because of semantic simplification or otherwise), we note that the system of Grigoriev and Hirsch~\cite{GrigorievHirsch03} (which they called ``formula-$\mathcal{PC}$'') is actually a \emph{syntactic} system and as such is deterministically checkable. Despite their system being restricted to being syntactic, it is still strong enough to simulate Frege and obtain low-depth refutations of the subset-sum axiom, the pigeonhole principle, and Tseitin tautologies.

\begin{remark}
	Note that our definition here varies slightly from the definition of Clegg, Edmonds, and Impagliazzo~\cite{CleggEI96}, in that we allow products by an arbitrary polynomial $g$ instead of only allowing products of a single variable $x_i$.  For some circuit classes $\cC$ these two definitions are polynomially equivalent (see for example the discussion in Raz and Tzameret~\cite{RazTzameret08a}).  In general however, using the product rule $f\vdash x_i\cdot f$ in a \emph{tree-like} proof can only yield $g\cdot f$ where $g$ is a small formula. However, we will be interested in algebraic circuit classes not known to be simulated by small formulas (such as roABPs, which can compute iterated matrix products which are believed to require super-polynomial-size formulas) and as such will consider this stronger product rule.
\end{remark}

We now observe that tree-like $\cC$-PC can simulate $\cC$-\lIPS for natural restricted circuit classes $\cC$.

\begin{lemma}\label{res:ips-vs-pc}
	Let $\cC$ be a restricted class of circuits computing polynomials in $\F[x_1,\ldots,x_n]$, and suppose that $\cC$-circuits grow polynomially in size under multiplication and addition, that is,
	\begin{itemize}
		\item $\size_\cC(f\cdot g)\le \poly(\size_\cC(f),\size_\cC(g))$.
		\item $\size_\cC(f+g)\le \poly(\size_\cC(f))+\poly(\size_\cC(g))$.
	\end{itemize}
	In particular, one can take $\cC$ to be sparse polynomials, depth-3 powering formulas (in characteristic zero), or roABPs.

	Then if $\vf,\vx^2-\vx$ are computable by size-$t$ $\cC$-circuits and have a $\cC$-\lIPS refutation of size-$s$, then $\vf,\vx^2-\vx$ have a tree-like $\cC$-PC refutation of size-$\poly(s,t,n)$, which is $\poly(s,t,n)$-explicit given the IPS refutation.
\end{lemma}
\begin{proof}	
	That the relevant classes obey these closure properties is mostly immediate. See for example \autoref{fact:roABP:closure} for roABPs.  For depth-3 powering formulas, the closure under addition is immediate and for multiplication it follows from Fischer~\cite{Fischer94}.
		
	Turning to the simulation, such an IPS refutation is an equation of the form $\sum_j g_jf_j+\sum_i h_i\cdot (x_i^2-x_i)=1$. Using the closure properties of $\cC$, one can compute the expression $\sum_j g_jf_j+\sum_i h_i\cdot (x_i^2-x_i)$ in the desired size, which yields the required (explicit) derivation of $1$.
\end{proof}

Note that the above claim does \emph{not} work for multilinear formulas, as multilinear polynomials are not closed under multiplication.  That tree-like multilinear-formula-PC simulates multilinear-formula-\lIPS is more intricate, and is given in \autoref{res:mult-form-PC:mult-form-lbIPS}.

The Polynomial Calculus proof system has received substantial attention since its introduction by Clegg, Edmonds, and Impagliazzo~\cite{CleggEI96}, typically when the complexity of the proofs are measured in terms of the number of monomials. In particular, Impagliazzo, \Pudlak and Sgall~\cite{IPS99} showed an exponential lower bound for the subset-sum axiom.

\begin{theoremwp}[Impagliazzo, \Pudlak and Sgall~\cite{IPS99}]\label{thm:IPS99}
	Let $\F$ be a field of characteristic zero. Let $\vaa\in\F^n$, $\beta\in\F$ and $A\eqdef\{\sum_{i=1}^n \alpha_i x_i : \vx\in\bits^n\}$ be so that $\beta\notin A$. Then $\alpha_1x_1+\cdots+\alpha_nx_n-\beta,\vx^2-\vx$ is unsatisfiable and any PC refutation requires degree $\ge\ceil{\nicefrac{n}{2}}+1$ and $\exp(\Omega(n))$-many monomials.
\end{theoremwp}

\subsection{roABP-PC}

The class of roABPs are a natural restricted class of algebraic computation that non-trivially goes beyond sparse polynomials.  In proof complexity, roABP-PC was explored by Tzameret~\cite{Tzameret11} (under the name of \emph{ordered formulas}, a formula-variant of roABPs, but the results there apply to roABPs as well).  In particular, Tzameret~\cite{Tzameret11} observed that roABP-PC can be deterministically checked using the efficient PIT algorithm for roABPs due to Raz and Shpilka~\cite{RazShpilka05}.

Given the Impagliazzo, \Pudlak and Sgall~\cite{IPS99} lower bound for the subset-sum axiom (\autoref{thm:IPS99}), our roABP-IPS upper bound for this axiom (\autoref{res:ips-ubs:subset:roABP}), and the relation between \lIPS and tree-like PC (\autoref{res:ips-vs-pc}), we can conclude the following exponential separation.

\begin{corollarywp}
	Let $\F$ be a field of characteristic zero. Then $x_1+\cdots+x_n+1,\vx^2-\vx$ is unsatisfiable, requires sparse-PC refutations of size-$\exp(\Omega(n))$, but has $\poly(n)$-explicit $\poly(n)$-size roABP-\lIPS and tree-like roABP-PC refutations.
\end{corollarywp}

This strengthens a result of Tzameret~\cite{Tzameret11}, who separated \emph{dag}-like roABP-PC from sparse-PC\@. However, we note that it is not clear whether sparse-PC can be efficiently simulated by roABP-\lIPS.

\subsection{Multilinear Formula PC}

We now proceed to study algebraic proofs defined in terms of multilinear formulas, as explored by Raz and Tzameret~\cite{RazTzameret08a,RazTzameret08b}.  We seek to show that the tree-like version of this system can simulate multilinear-formula-\lIPS.  While tree-like $\cC$-PC can naturally simulate $\cC$-\lIPS if $\cC$ is closed under multiplication (\autoref{res:ips-vs-pc}), the product of two multilinear polynomials may not multilinear.  As such, the simulation we derive is more intricate, and is similar to the efficient multilinearization results for multilinear formulas from \autoref{sec:multilinearization:mult-form}. We first define the Raz-Tzameret~\cite{RazTzameret08a,RazTzameret08b}  system (which they called \emph{fMC}).

\begin{definition}\label{defn:pc:multlin}
	Let $f_1(\vx),\ldots,f_m(\vx)\in\F[x_1,\ldots,x_n]$ be a system of polynomials. A \demph{multilinear-formula-PC$\,^\neg$ refutation} for showing that the system $\vf,\vx^2-\vx$ is unsatisfiable is a multilinear-formula-PC refutation of $\vf(\vx),\vx^2-\vx,\vnx^2-\vnx,\vx\circ\vnx$ in the ring $\F[x_1,\ldots,x_n,\nx_1,\ldots,\nx_n]$, where `$\circ$' denotes the entry-wise product so that $\vx\circ\vnx=(x_1\nx_1,\ldots,x_n\nx_n)$.
\end{definition}

That is, a multilinear-formula-PC$^\neg$ refutation of $\vf,\vx^2-\vx$ is a multilinear-formula-PC refutation with the additional variables $\vnx\eqdef(\nx_1,\ldots,\nx_n)$ which are constrained so that $\nx_i=1-x_i$ (so that `$\neg$' here is simply a modifier of the symbol `$x$' as opposed to being imbued with mathematical meaning).  These additional variables are important, as without them the system is not complete. For example, in attempting to refute the subset-sum axiom $\sum_i x_i+1,\vx^2-\vx$ in multilinear-formula-PC alone, one can never multiply the axiom $\sum_i x_i+1$ by another (non-constant) polynomial as it would ruin multilinearity.  However, in multilinear-formula-PC$^\neg$ one can instead multiply by polynomials in $\vnx$ and appropriately simplify.  We now formalize this by showing that tree-like multilinear-formula-PC$^\neg$ can simulate multilinear-formula-\lbIPS (which is complete (\autoref{res:multilin:simulate-sparse})).  

We begin by proving a lemma on how the $\vnx$ variables can help multilinearize products.  In particular, if we have a monomial $(\vno-\vnx)^\va$ (which is meant to be equal to $\vx^\va$) and multiply by $\vx^\vno$ we should be able to prove that this product equals $\vx^\vno$ modulo the axioms.

\begin{lemma}\label{res:mult-form-PC:mon-mon}
	Working in the ring $\F[x_1,\ldots,x_d,\nx_1,\ldots,\nx_d]$, and for $\vnz\le\va\le \vno$,
	\[
		(\vno-\vnx)^\va\vx^\vno-\vx^\vno
		=C(\vx,\vx\circ\vnx)
		\;,
	\]
	for $C(\vx,\vz)\in\F[\vx,\vz]$ where $C(\vx,\vx\circ\vnx)$ can be $\poly(2^d)$-explicitly derived from the axioms $\vx\circ\vnx$ in $\poly(2^d)$ steps using tree-like multilinear-formula-PC.
\end{lemma}
\begin{proof}
	\begin{align*}
		(\vno-\vnx)^\va\vx^\vno
		&=\vx^{\vno-\va}\cdot (\vx-\vx\circ\vnx)^\va\\
		&=\vx^{\vno-\va}\cdot \left(\sum_{\vnz\le\vb\le\va} \vx^{\va-\vb}(-\vx\circ\vnx)^{\vb}\right)\\
		&=\vx^{\vno-\va}\cdot \left(\vx^\va+\sum_{\vnz<\vb\le\va} \vx^{\va-\vb}(-\vx\circ\vnx)^{\vb}\right)\\
		&=\vx^\vno+\sum_{\vnz<\vb\le\va} \vx^{\vno-\vb}(-\vx\circ\vnx)^{\vb}\\
		&=\vx^\vno+C(\vx,\vx\circ\vnx)
		\;,
	\end{align*}
	where $C(\vx,\vz)$ is defined by
	\[
		C(\vx,\vz)\eqdef\sum_{\vnz<\vb\le\va} \vx^{\vno-\vb}(-\vz)^{\vb}
		\;.
	\]
	Now note that $C(\vx,\vx\circ\vnx)$ can be derived by tree-like multilinear-formula-PC\@.  That is, the expression $\vx^{\vno-\vb}(-\vx\circ\vnx)^{\vb}$ is multilinear (as the product is variable disjoint) and in the ideal of $\vx\circ\vnx$ as $\vb>\vnz$, and is clearly a $\poly(d)$-size explicit multilinear formula. Summing over the $2^d-1$ relevant $\vb$ gives the result.
\end{proof}

We now show how to prove the equivalence of $g(\vx)$ and $g(\vno-\vnx)$ modulo $\vx+\vnx-\vno$, if $g$ is computable by a small multilinear formula, where we proceed variable by variable.

\begin{lemma}\label{res:mult-form-PC:form-form}
	Working in the ring $\F[x_1,\ldots,x_n,\nx_1,\ldots,\nx_n]$, if $g\in\F[\vx]$ is computable by a size-$t$ multilinear formula, than
	\[
		g(\vx)-g(\vno-\vnx)=C(\vx,\vx+\vnx-\vno)
		\;,
	\]
	for $C(\vx,\vz)\in\F[\vx,\vz]$ where $C(\vx,\vx+\vnx-\vno)$ is derivable from $\vx+\vnx-\vno$ in size-$\poly(t,n)$ tree-like multilinear-formula-PC, which is $\poly(t,n)$-explicit given the formula for $g$.
\end{lemma}
\begin{proof}
	We proceed to replace $\vno-\vnx$ with $\vx$ one variable at a time.  Using $(\vx_{\le i},(\vno-\vnx)_{>i})$ to denote $(x_1,\ldots,x_i,1-\nx_{i+1},\ldots,1-\nx_{n})$, we see that via telescoping that
	\begin{align*}
		g(\vx)-g(\vno-\vnx)
		&=g(\vx_{\le n},(\vno-\vnx)_{>n})-g(\vx_{<1},(\vno-\vnx)_{\ge 1})\\
		&=\sum_{i=1}^n \Big (g(\vx_{\le i},(\vno-\vnx)_{>i}) - g(\vx_{<i},(\vno-\vnx)_{\ge i})\Big)\\
		&=\sum_{i=1}^n \Big (g(\vx_{<i},x_i,(\vno-\vnx)_{>i}) - g(\vx_{<i},1-\nx_i,(\vno-\vnx)_{> i})\Big)
		\;.
	\end{align*}
	Now note that $g(\vx_{<i},y,(\vno-\vnx)_{>i})$ is a multilinear polynomial, which as it is linear in $y$ can be written as
	\[
		g(\vx_{<i},y,(\vno-\vnx)_{>i})=\Big(g(\vx_{<i},1,(\vno-\vnx)_{>i})-g(\vx_{<i},0,(\vno-\vnx)_{>i})\Big)y+g(\vx_{<i},0,(\vno-\vnx)_{>i})
		\;.
	\]
	Thus, plugging in $x_i$ and $1-\nx_i$,
	\begin{align*}
		g(\vx_{<i},x_i&,(\vno-\vnx)_{>i})-g(\vx_{<i},1-\nx_i,(\vno-\vnx)_{>i})\\
		&=\left(\Big(g(\vx_{<i},1,(\vno-\vnx)_{>i})-g(\vx_{<i},0,(\vno-\vnx)_{>i})\Big)x_i+g(\vx_{<i},0,(\vno-\vnx)_{>i})\right)\\
		&\hspace{.3in}-\left(\Big(g(\vx_{<i},1,(\vno-\vnx)_{>i})-g(\vx_{<i},0,(\vno-\vnx)_{>i})\Big)(1-\nx_i)+g(\vx_{<i},0,(\vno-\vnx)_{>i})\right)\\
		&=\Big(g(\vx_{<i},1,(\vno-\vnx)_{>i})-g(\vx_{<i},0,(\vno-\vnx)_{>i})\Big)(x_i+\nx_i-1)
		\;.
	\end{align*}
	Plugging this into the above telescoping equation,
	\begin{align*}
		g(\vx)-g(\vno-\vnx)
		&=\sum_{i=1}^n \Big (g(\vx_{<i},x_i,(\vno-\vnx)_{>i}) - g(\vx_{<i},1-\nx_i,(\vno-\vnx)_{> i})\Big)\\
		&=\sum_{i=1}^n \Big(g(\vx_{<i},1,(\vno-\vnx)_{>i})-g(\vx_{<i},0,(\vno-\vnx)_{>i})\Big)(x_i+\nx_i-1)\\
		&\defeq C(\vx,\vx+\vnx-\vno)
		\;.
	\end{align*}
	Clearly each $g(\vx_{<i},b,(\vno-\vnx)_{>i})$ for $b\in\bits$ has a $\poly(t)$-size multilinear algebraic formula, so the entire expression $C(\vx,\vx+\vnx-\vno)$ can be computed by tree-like multilinear-formula-PC from $\vx+\vnx-\vno$ explicitly in $\poly(t,n)$ steps.
\end{proof}

Using the above lemma, we now show how to multilinearize a multilinear-formula times a low-degree multilinear monomial.

\begin{lemma}\label{res:mult-form-PC:form-mon}
	Let $g,f\in\F[x_1,\ldots,x_n,y_1,\ldots,y_d]$, where $g$ is computable by a size-$t$ multilinear formula and $y=\prod_{i=1}^d y_i$.  Then
	\[
		g(\vno-\vnx,\vno-\vny)\vy^\vno - \ml(g(\vx,\vy)\vy^\vno)
		=C(\vx,\vy,\vx+\vnx-\vno,\vy\circ\vny)
		\;,
	\]
	where $C(\vx,\vy,\vx+\vnx-\vno,\vy\circ\vny)$ can be derived from the axioms $\vx+\vnx-\vno,\vy\circ\vny$ in $\poly(2^d,t,n)$ steps using tree-like multilinear-formula-PC.
\end{lemma}
\begin{proof}
	Express $g$ as $g(\vx,\vy)=\sum_{\vnz\le\va\le \vno}g_\va(\vx)\vy^\va$ in the ring $\F[\vx][\vy]$, so that each $g_\va$ is multilinear. Then,
	\begin{align*}
		g(\vno-\vnx,\vno-\vny)\cdot \vy^\vno
		&=\sum_{\vnz\le\va\le\vno} g_\va(\vno-\vnx)(\vno-\vny)^\va \cdot \vy^\vno\\
		\intertext{appealing to \autoref{res:mult-form-PC:mon-mon} to obtain $(\vno-\vny)^\va \vy^\vno=\vy^\vno+C_\va(\vy,\vy\circ\vny)$ for $C_\va(\vy,\vy\circ\vny)$ derivable in tree-like multilinear-formula-PC from $\vy\circ\vny$ in $\poly(2^d)$ steps,}
		&=\sum_{\va} g_\va(\vno-\vnx)\left(\vy^\vno+C_\va(\vy,\vy\circ\vny)\right)
		\intertext{appealing to \autoref{res:mult-form-PC:form-form} to obtain $g_\va(\vno-\vnx)=g_\va(\vx)+B_\va(\vx,\vx+\vnx-\vno)$ for $B_\va(\vx,\vx+\vnx-\vno)$ derivable in tree-like multilinear-formula-PC from $\vx+\vnx-\vno$ in $\poly(t,n)$ steps,}
		&=\sum_{\va} \left(g_\va(\vx)+B_\va(\vx,\vx+\vnx-\vno)\right)\cdot \left(\vy^\vno+C_\va(\vy,\vy\circ\vny)\right)\\
		&=\sum_{\va} g_\va(\vx)\vy^\vno+\sum_\va \Big( B_\va(\vx,\vx+\vnx-\vno)\vy^\vno + g_\va(\vx)C_\va(\vy,\vy\circ\vny)\\
		&\hspace{1in}+B_\va(\vx,\vx+\vnx-\vno)C_\va(\vy,\vy\circ\vny)\Big)\\
		&=\ml(g(\vx,\vy)\vy^\vno)+C(\vx,\vy,\vx+\vnx-\vno,\vy\circ\vny)
		\;,
	\end{align*}
	by defining $C$ appropriately, and as
	\begin{align*}
		\ml(g(\vx,\vy)\vy^\vno)
		&=\textstyle\ml\left(\sum_{\vnz\le\va\le \vno}g_\va(\vx)\vy^\va\cdot \vy^\vno\right)\\
		&=\textstyle\ml\left(\sum_{\vnz\le\va\le \vno}g_\va(\vx)\vy^{\va+\vno}\right)\\
		&=\textstyle\sum_{\vnz\le\va\le \vno}g_\va(\vx)\vy^{\vno}
		\;.
	\end{align*}
	By interpolation, it follows that for each exponent $\va$ there are constants $\vaa_{\va,\vbb}$ such that $g_\va(\vx)=\sum_{\vbb\in\bits^d} \alpha_{\va,\vbb}g(\vx,\vbb)$.  From this it follows that $g_\va$ is computable by a multilinear formula of size-$\poly(t,2^d)$.  It thus follows that $C(\vx,\vy,\vx+\vnx-\vno,\vy\circ\vny)$ is a sum of $2^d$ terms, each of which is explicitly derivable in $\poly(2^d,t,n)$ steps in tree-like multilinear-formula-PC from $\vx+\vnx-\vno,\vy\circ\vny$ (as the multiplications are variable-disjoint), and thus $C(\vx,\vy,\vx+\vnx-\vno,\vy\circ\vny)$ is similar derived by tree-like multilinear-formula-PC\@.
\end{proof}

By linearity we can extend the above to multilinearization of a multilinear-formula times a sparse low-degree multilinear polynomial.

\begin{corollarywp}\label{res:mult-form-PC:form-sparse}
	Let $g,f\in\F[x_1,\ldots,x_n]$ be multilinear, where $g$ is computable by a size-$t$ multilinear formula and $f$ is $s$-sparse and $\deg f\le d$.  Then
	\[
		g(\vno-\vnx)\cdot f(\vx) - \ml(g(\vx)\cdot f(\vx))
		=C(\vx,\vx+\vnx-\vno,\vx\circ\vnx)
		\;,
	\]
	where $C(\vx,\vx+\vnx-\vno,\vx\circ\vnx)$ can be derived from the axioms $\vx+\vnx-\vno,\vx\circ\vnx$ in $\poly(2^d,s,t,n)$ steps using tree-like multilinear-formula-PC\@.
\end{corollarywp}

We now conclude by showing that tree-like multilinear-formula-PC$^\neg$ can efficiently simulate multilinear-formula-\lbIPS.  Recall that this proof system simply requires an IPS refutation that is linear in the non-boolean axioms, so that in particular $\sum_j g_jf_j\equiv 1\mod \vx^2-\vx$ for efficiently computable $g_j$.

\begin{corollary}\label{res:mult-form-PC:mult-form-lbIPS}
	Let $f_1,\ldots,f_m\in\F[x_1,\ldots,x_n]$ be degree at most $d$ multilinear $s$-sparse polynomials which are unsatisfiable over $\vx\in\bits^n$.  Suppose that there are multilinear $g_j\in\F[\vx]$ computable by size-$t$ multilinear formula such that
	\[
		\sum_{i=1}^m g_j(\vx)f_j(\vx)\equiv 1\mod\vx^2-\vx
		\;.
	\]
	Then there is a tree-like multilinear-formula-PC$\,^\neg$ refutation of $\vf,\vx^2-\vx$ of size $\poly(2^d,s,t,n,m)$, which is $\poly(2^d,s,t,n,m)$-explicit given the formulas for the $f_j,g_j$.

	In particular, if there is a size-$t$ multilinear-formula-\lbIPS refutation of $\vf,\vx^2-\vx$, then there is a tree-like multilinear-formula-PC$\,^\neg$ refutation of $\vf,\vx^2-\vx$ of size $\poly(2^d,s,t,n,m)$ which is $\poly(2^d,s,t,n,m)$-explicit given the refutation of $\vf,\vx^2-\vx$ and formulas for the $f_j$.
\end{corollary}
\begin{proof}
	By the above multilinearization (\autoref{res:mult-form-PC:form-sparse}), there are $C_j\in\F[\vx,\vz,\vw]$ such that
	\[
		g_j(\vno-\vnx)f_j(\vx)=\ml(g_j(\vx)f_j(\vx))+C_j(\vx,\vx+\vnx-\vno,\vx\circ\vnx)
		\;.
	\]
	where $C_j(\vx,\vx+\vnx-\vno,\vx\circ\vnx)$ is derivable from $\vx+\vnx-\vno,\vx\circ\vnx$ in $\poly(2^d,s,t,n)$ steps of tree-like multilinear-formula-PC\@. Thus, as $g_j(\vno-\vnx)$ has a $\poly(t)$-size multilinear formula, in $\poly(2^d,s,t,n,m)$ steps we can derive from $\vf(\vx),\vx+\vnx-\vno,\vx\circ\vnx$,
	\begin{align*}
		\sum_{j=1}^m \left (g_j(\vno-\vnx)f_j(\vx) - C_j(\vx,\vx+\vnx-\vno,\vx\circ\vnx)\right)
		&=\sum_{j=1}^m \ml(g_j(\vx)f_j(\vx))\\
		\intertext{as $\sum_{i=1}^m g_j(\vx)f_j(\vx)\equiv 1\mod \vx^2-\vx$ we have that 
			\[
				\sum_{j=1}^m \ml(g_j(\vx)f_j(\vx))
				=
				\ml\left(\sum_{i=1}^m g_j(\vx)f_j(\vx)\right)=1
				\;,
			\]
			where we appealed to linearity of multilinearization (\autoref{fact:multilinearization}), so that
		}
		\sum_{j=1}^m \left (g_j(\vno-\vnx)f_j(\vx) - C_j(\vx,\vx+\vnx-\vno,\vx\circ\vnx)\right)
		&=1
		\;,
	\end{align*}
	yielding the desired refutation, where the explicitness is clear.

	The claim about multilinear-formula-\lbIPS follows, as such a refutation induces the equation $\sum_{i=1}^m g_j(\vx)f_j(\vx)\equiv 1\mod\vx^2-\vx$ with the appropriate size bounds.
\end{proof}

Given this efficient simulation of multilinear-formula-\lbIPS by tree-like multilinear-formula-PC$^\neg$ (\autoref{res:mult-form-PC:mult-form-lbIPS}), our multilinear-formula-\lIPS refutation of the subset-sum axiom (\autoref{res:ips-ubs:subset:mult-form}), and the lower bound for sparse-PC of the subset-sum axiom (\autoref{thm:IPS99}), we obtain the following separation result.

\begin{corollarywp}
	Let $\F$ be a field of characteristic zero. Then $x_1+\cdots+x_n+1,\vx^2-\vx$ is unsatisfiable, requires sparse-PC refutations of size-$\exp(\Omega(n))$, but has $\poly(n)$-explicit $\poly(n)$-size multilinear-formula-\lIPS and tree-like multilinear-formula-PC$\,^\neg$ refutations.
\end{corollarywp}

This strengthens a results of Raz and Tzameret~\cite{RazTzameret08a,RazTzameret08b}, who separated \emph{dag}-like multilinear-formula-PC$^\neg$ from sparse-PC\@. However, we note that it is not clear whether sparse-PC can be efficiently simulated by multilinear-formula-\lbIPS.

\section{Explicit Multilinear Polynomial Satisfying a Functional Equation}\label{sec:appendix}

In \autoref{sec:lbs-fn:deg} we showed that any polynomial that agrees with the function $\vx\mapsto \nicefrac{1}{\left(\sum_i x_i-\beta\right)}$ on the boolean cube must have degree $\ge n$.  However, as there is a unique multilinear polynomial obeying this functional equation it is natural to ask for an explicit description of this polynomial, which we now give.

\begin{proposition}\label{res:subsetsum:multlin}
	Let $n\ge 1$ and $\F$ be a field with $\chara(\F)>n$.  Suppose that $\beta\in \F\setminus\{0,\ldots,n\}$. Let $f\in\F[x_1,\ldots,x_n]$ be the unique multilinear polynomial such that
        \[
                f(\vx)=\frac{1}{\sum_i x_i-\beta}
                \;,
        \]
        for $\vx\in\bits^n$.  Then
        \[
		f(\vx)=-\sum_{k=0}^n \frac{k!}{\prod_{j=0}^{k}(\beta-j)} S_{n,k}
                \;,
        \]
	where $S_{n,k}\eqdef \sum_{S\subseteq\binom{[n]}{k}}\prod_{i\in S} x_i$ is the $k$-th elementary symmetric polynomial.
\end{proposition}
\begin{proof}
	It follows from the uniqueness of the evaluations of multilinear polynomials over the boolean cube that
	\begin{align*}
		f(\vx)
		&=\sum_{T\subseteq[n]} f(\ind{T})\prod_{i\in T} x_i \prod_{i\notin T} (1-x_i)
		\intertext{where $\ind{T}\in\bits^n$ is the indicator vector of the set $T$, so that}
		&=\sum_{T\subseteq[n]} \frac{1}{|T|-\beta}\prod_{i\in T} x_i \prod_{i\notin T} (1-x_i)
		\;.
	\end{align*}
	Using this, let us determine the coefficient of $\prod_{i\in S} x_i$ in $f(\vx)$, for $S\subseteq[n]$ with $|S|=k$. First observe that setting $x_i=0$ for $i\notin S$ preserves this coefficient, so that
	\begin{align*}
		\coeff{\prod_{i\in S} x_i}\Big(f(\vx)\Big)
		&=\coeff{\prod_{i\in S} x_i}\left.\left(\sum_{T\subseteq[n]} \frac{1}{|T|-\beta}\prod_{i\in T} x_i \prod_{i\notin T} (1-x_i)\right)\right|_{x_i\leftarrow 0, i\in S}\\
		\intertext{and thus those sets $T$ with $T\not\subseteq S$ are zeroed out,}
		&=\coeff{\prod_{i\in S} x_i}\left(\sum_{T\subseteq S} \frac{1}{|T|-\beta}\prod_{i\in T} x_i \prod_{i\in S\setminus T} (1-x_i)\right)\\
		&=\sum_{T\subseteq S} \frac{1}{|T|-\beta}\coeff{\prod_{i\in S} x_i}\left(\prod_{i\in T} x_i \prod_{i\in S\setminus T} (1-x_i)\right)\\
		&=\sum_{T\subseteq S} \frac{1}{|T|-\beta} (-1)^{k-|T|}\\
		&=\sum_{j=0}^k \binom{k}{j} \frac{1}{j-\beta} (-1)^{k-j}\\
		&=-\frac{k!}{\prod_{j=0}^k (\beta-j)}
		\;,
	\end{align*}
	where the last step uses the below subclaim.
	\begin{subclaim}
		\[
			\sum_{j=0}^k \binom{k}{j} \frac{1}{j-\beta} (-1)^{k-j}
			=-\frac{k!}{\prod_{j=0}^k (\beta-j)}
			\;.
		\]
	\end{subclaim}
	\begin{subproof}
		Clearing denominators,
		\begin{align*}
			\prod_{i=0}^k (i-\beta)\cdot \sum_{j=0}^k \binom{k}{j} \frac{1}{j-\beta} (-1)^{k-j}
			&=\sum_{j=0}^k \binom{k}{j}(-1)^{k-j} \prod_{i\ne j} (i-\beta)
			\;.
		\end{align*}
		Note that the right hand side is a univariate degree $\le k$ polynomial in $\beta$, so it is determined by its value on $\ell\in\{0,\ldots,k\}$ (that $\F$ has large characteristic implies that these values are distinct). Note that on these values,
		\begin{align*}
			\sum_{j=0}^k \binom{k}{j}(-1)^{k-j} \prod_{i\ne j} (i-\ell)
			&=\binom{k}{\ell}(-1)^{k-\ell} \prod_{0\le i<\ell} (i-\ell) \cdot  \prod_{\ell<i\le k} (i-\ell)\\
			&=\binom{k}{\ell}(-1)^{k-\ell} \cdot (-1)^{\ell} \ell! \cdot  (k-\ell)!\\
			&=(-1)^k k!
			\;.
		\end{align*}
		Thus by interpolation $\sum_{j=0}^k \binom{k}{j}(-1)^{k-j} \prod_{i\ne j} (i-\beta)=(-1)^k k!$ for all $\beta$, and thus dividing by $\prod_{i=0}^k (i-\beta)$ and clearing $-1$'s yields the claim.
	\end{subproof}
	This then gives the claim as the coefficient of $\prod_{i\in S}x_i$ only depends on $|S|=k$.
\end{proof}

Noting that the above coefficients are all nonzero because $\chara(\F)>n$.  Thus, we obtain the following corollary by observing that degree and sparsity are non-increasing under multilinearization (\autoref{fact:multilinearization}).

\begin{corollarywp}\label{res:subsetsum:multlin:deg-sparse}
	Let $n\ge 1$ and $\F$ be a field with $\chara(\F)>n$.  Suppose that $\beta\in \F\setminus\{0,\ldots,n\}$. Let $f\in\F[x_1,\ldots,x_n]$ be a polynomial such that
	\[
		f(\vx)\left(\sum_i x_i-\beta\right)=1 \mod \vx^2-\vx
		\;.
	\]
	Then $\deg f\ge n$, and $f$ requires $\ge 2^n$ monomials.
\end{corollarywp}

\end{document}